\documentclass[11pt]{article}					
\usepackage{amsmath, amssymb, amsthm, bm, calrsfs}
\usepackage{multirow, booktabs}  
\usepackage{rotfloat}
\usepackage[longnamesfirst]{natbib}		%
\usepackage{jfe}
\usepackage{indentfirst}	

\usepackage[T1]{fontenc}
\usepackage[margin=1in,  a4paper, pdftex]{geometry}
\usepackage{setspace}
\usepackage{authblk}	

\usepackage{pgfplots}
\usetikzlibrary{decorations.markings}
\usetikzlibrary{3d}

\newtheorem{proposition}{Proposition}

\newcommand{\N}{\mathbb{N}} \newcommand{\Z}{\mathbb{Z}} \newcommand{\E}{\mathbb{E}}
\newcommand{\defeq}{\stackrel{\rm def}{=}}

\graphicspath{{./paperfig/}}
\makeatletter
\def\input@path{{./paperfig/}}
\makeatother

\expandafter\def\expandafter\normalsize\expandafter{%
    \normalsize
    \setlength\abovedisplayskip{8pt}
    \setlength\belowdisplayskip{8pt}
    \setlength\abovedisplayshortskip{1pt}
    \setlength\belowdisplayshortskip{3pt}
}

\newcommand{\abs}{\noindent \setstretch{1.85}
Using a fast numerical technique, we investigate a large database of investor
suboptimal non-exercise of short maturity American call options on dividend-paying stocks listed
on the Dow Jones. The correct modelling of the discrete dividend is essential for a correct
calculation of the early exercise boundary as confirmed by theoretical insights. Pricing with stochastic volatility
and jumps instead of the Black-Scholes-Merton benchmark cuts by a quarter the amount lost
by investors through suboptimal exercise. The remaining three quarters are largely unexplained
by transaction fees and may be interpreted as an opportunity cost for the investors to monitor
optimal exercise.}

\author[*]{Antonio Cosma}
\author[**]{Stefano Galluccio}
\author[***]{Paola Pederzoli}
\author[***]{Olivier Scaillet}
\affil[*]{University of Luxembourg}
\affil[**]{Incipit Capital, London}  
\affil[***]{University of Geneva and Swiss Finance Institute}

\begin{document}

\newcommand{\tit}{Early exercise decision in American options with dividends, stochastic volatility and jumps}

\title{{\bf \tit}}

%
\date{}
\maketitle

{\singlespacing
\centerline{\bf Abstract}
\medskip
\abs


}

{\let\thefootnote\relax\footnotetext{  We thank J. Detemple, D. Duffie, E. Gobet, J. Jackwerth, F. Moraux, D. Newton, A. Treccani and A. Valdesogo for valuable insight and help in addition to the participants at the 6th World Congress of the Bachelier Finance Society in Toronto, the 2011 European Econometric Society in Oslo, the 2013 Mathematical Finance Day in Montreal, the 2014 Mathematical and Statistical Methods for Actuarial Sciences and Finance in Salerno, the 2014 International
Symposium on Differential Equations and Stochastic Analysis in Mathematical Finance in Sanya, the 7th General Advanced Mathematical Methods in Finance and Swissquote Conference in Lausanne, the 2015 IEEE Symposium on Computational Intelligence for Financial Engineering and Economics in Cape Town, the 9th International Conference on Computational and Financial Econometrics in London, the 2015 International Conference on Computational Finance in London, SGF Conference 2016 in Zurich, the FERM 2016 in Guangzhou, the AFFI 2016 in Li\`ege, the Stochmod16 in Louvain-la-Neuve, the 9th World Congress of the Bachelier Society in New York, the EFA 2016 in Oslo, and seminars at the University of Geneva and the University of Orl\'eans. An earlier version of this work was circulated under the title ``Valuing American options using fast recursive projections''. O. Scaillet received support from the Swiss NSF through the NCCR Finrisk. Paola Pederzoli acknowledges the financial support of the Swiss NSF (grant $100018-149307$). Part of the research  was made when she was visiting  LSE. A previous version of this paper circulated under the title ``Valuing American options using fast recursive projections''.}}

\clearpage

\noindent
Holders of short maturity American call options on dividend-paying stocks are known to miss exercising their options in an apparently suboptimal way (see e.g. \cite{pool2008failure}).
Financial frictions are a possible explanation of this departure from the expected exercise behavior (e.g. \cite{jensen2016early}).
We investigate suboptimal exercise in the light of such frictions according to alternative models for the underlying stocks.
To do this, we compile data on 30 individual dividend-paying stocks listed on the Dow Jones, comprising a total of 101,295 series of short-term options amounting to approximately 9.5 million records. In order to tackle repeated calculations with this large database, we require an exceptionally fast option pricing technique, able to price contracts whose path is monitored at discrete moments in time.
The challenges in studying derivative products come from the degree of sophistication of the process of the underlying asset and the complexity of the payoff and exercise rules. Almost any departure from the plain vanilla European style options implies that closed-form pricing formulas are no longer available (for an extensive review, see \cite{detemple2005american}).
Only a few approaches, such as \cite{longstaff2001valuing}, \cite{haugh2004pricing}, \cite{fang2011fourier}, and \cite{chen2014advancing}, are readily able to tackle the problem of interest to us, but may lose in speed or ease of implementation if dividends are added in the picture.
Since we also require exceptional speed with our database, we develop a technique from the QUAD stable of option pricing (see \cite{andricopoulos2007extending} for treatment of path-dependency) which we refer to as a method of recursive projections to distinguish it from other variants. We are the first to characterize the convergence properties of a quadrature-based method in the presence of discrete dividends and with the underlying following a dynamics outside the Black-Scholes benchmark.
The recursive nature of our algorithm, which gives the name of the method, refers to the recursive relation of the price of the contract at different points in time.
This relationship translates the pricing problem in a sequence of simple matrix times vector multiplications.
This recursive property is not affected by intermediate cash flows as for instance dividend payments. This last feature, in addition to its speed and simplicity,  make our method well-suited to the calculations we need to perform in an empirical analysis with different pricing models on a large database.

We extend the empirical work of \cite{pool2008failure} on the observed suboptimal non-exercise of American call options written on dividend-paying stocks. We show that, by taking into account stochastic volatility and jumps in the process of the underlying asset, we can explain up to 25\% of the gain forgone due to non-optimal exercise decisions, as computed in \cite{pool2008failure}. Because financial frictions are a possible explanation of departure from the expected exercise behavior (e.g. \cite{jensen2016early}),
we also show that transaction costs cannot fully explain the non-exercise decisions. In the process, to our knowledge, we are the first to provide comprehensive descriptive statistics of the parameters driving the jumps and the stochastic volatility of the constituents of the Dow Jones Industrial Average Index (DJIA) traded in the period from January 1996 to December 2012.
In our calibration, we price by fully taking into account the discrete nature of the dividend distributed by the underlying stocks and the American style of the call options, and we do so for different specifications of the stock dynamics.
This feature is a peculiarity of our work, given that the standard empirical literature on options mainly focuses on European $S\&P 500$ options with a dividend yield (\cite{bakshi1997empirical}, \cite{Eraker2003}). \cite{broadie2007model} and \cite{Broadie2009} approximate  American prices with European ones, and show that transforming American options to European ones does not matter for calibration purposes when facing a continuous dividend yield since differences in early exercise premia are not so large in that case.
This is not true with multiple discrete dividend payments, and we provide an example on how neglecting the discrete nature of a dividend or its time of payment leads to an incorrect exercise decision.
We also provide new theoretical
insights into how the early exercise boundary changes, depending on the discrete or continuous nature of the dividend distributed by the underlying asset. In particular, we show that, for short maturities, the boundary is higher under the \cite{merton1976option} jump-diffusion and
\cite{heston1993closed} stochastic volatility models than under the Black-Scholes model if the dividend is discrete,  whereas we know it is lower in the case of a continuous dividend yield (\cite{Amin:1993a}, \cite{adolfsson2013representation}). The study of the early exercise boundary is important for investment decisions. For example \cite{Battauz2014real} characterize the double continuation region implied by an option with a negative interest rate, which occurs in  the case of gold loans.

To obtain a first intuition on how our method works, we can view the pricing of a derivative security essentially from two perspectives, with the link between the two being given by the Feymann-Kac theorem. The first perspective is solving the partial differential equation (PDE) to yield the price of the derivative assets. Numerically discretizing the differential operator leads to finite difference schemes
(see \cite{brennan1977valuation}, \cite{clarke1999multigrid}, \cite{ikonen2008efficient}).
This method is the most common approach in regard to numerically finding solutions to complex pricing problems, and we benchmark our method against the PDE outside the Black-Scholes setting. In the numerical examples of our paper we compare the recursive projections to the PDE in the Heston setting and not in the more general Bates framework because the implementation of the PDE in the former environment has been studied in detail, while this is not true for the latter setting.
Introducing jumps in the underlying process while keeping the finite differences viable from a computational point of view asks for specific techniques (see for instance \cite{d2005robust}) which are model specific and not yet implemented in conjunction with stochastic volatility. No off-the-shelf PDE method is available for the empirical analysis we are carrying out in this paper.

The second perspective is viewing the price of the derivative asset as the conditional expectation of the discounted future payoff. It exploits the knowledge of the discounted probability distribution (the so-called Green function) with respect to which the conditional expectation is taken.
The advantage of this class of methods over the previous class is that it does not introduce time stepping errors when the value function is evaluated only at specific points in time, typically potential exercise dates, while finite difference methods require a finer discretization in the time dimension to achieve satisfactory accuracy.
This paper opts for the second perspective.
We develop our technique after the line of quadrature-based methods that provide fast and effective routines to price path-dependent options.
Newton and co-authors follow up on an early intuition by \cite{sullivan2000valuing} and provide a pricing routine, called QUAD (\cite{Andricopoulos2003447}, \cite{andricopoulos2007extending}, \cite{chen2014advancing}).
The technique in the first two papers applies whenever the conditional probability density function is known, e.g. in the Black-Scholes-Merton framework for the underlying, \cite{merton1976option} jump-diffusion model and certain interest rate models. \cite{o2005path} observes that many useful processes without a well-known density function have a characteristic function and that we can obtain the density function as the inverse Fourier transform (through FFT) of the characteristic function to insert the output in the QUAD scheme. As \cite{chen2014advancing} point out, we cannot use this single-variable `FFT-QUAD' approach to price heavily path-dependent options in stochastic volatility frameworks, since it does not keep track of the evolution of the volatility process in moving from one observation point to the next. O'Sullivan's FFT-QUAD is improved considerably by the CONV techniques of \cite{lord2008fast}, referred to by later authors as CONV-QUAD (see \cite{chen2014advancing}).
This replaces the two integrals of FFT-QUAD with two fast Fourier transforms.
However, not all underlying processes could be covered by QUAD variants until \cite{chen2014advancing} introduce an approximation of the density function for cases were there is no characteristic function.
This allows pricing both under previously unavailable processes, such as SABR, but also under other processes, such as Heston, for which a characteristic function is available but approximation of the density function may be more convenient.
Meantime, pricing of path-dependent (Bermudan and barrier) options under Heston is solved by \cite{fang2011fourier} via Fourier cosine expansion and quadrature.
We follow the characteristic function route for Merton jump diffusion and for Bates (Heston with Merton jumps).

Although in its implementation the recursive projections resembles quadratures, we believe that the conceptual framework is more general.
We can express the value of the options on an equally spaced and time-homogenous grid of stock prices (as in \cite{simonato2016simplified}) as a function of the values of the option on the same grid at future times.
We show that choosing the grid at which we evaluate the contract is equivalent to projecting the value and Green functions on a basis of localized functions.
In this paper, the localized basis functions are indicator functions centred at the grid points. We can extend the function bases to other localized bases which are functions with a compact support but not necessarily indicator functions. In specific setups, this could ensure a faster convergence. Moreover, our function projection approach makes our method not dependent on FFT techniques.
While we do take advantage of the speed of FFT techniques in this specific application, our method would keep its main advantages if we know the Green function in the Laplace space, for instance.
Both in the case of different basis functions, or of a different transform space for the transition densities, we would still be able to represent conditional expectation operators in the simplest form of standard linear operators, i.e., matrices, and to price derivative contracts by means of linear algebra tools.
The output of a multiplication is the input of the following step, without the need for intermediate computations, such as the careful placing of nodes at discontinuities in the quadrature. The correct positioning of nodes in quadrature is itself the solution of the optimal control problem, since it requires the knowledge of the value function and the location of discontinuities in its derivative. That implementation can be cumbersome if payoffs include dividends, as in our leading application.
The localized nature of the function basis is key in allowing the inclusion of the discrete dividend case.
The recursive relation among coefficients of the delocalized Fourier cosine series expansion that makes the algorithm of  \cite{fang2011fourier} appealing in many situations, breaks down in the presence of discrete dividends.
This lack of intermediate computational overhead boosts the speed of our method and makes it feasible for empirics. We can accommodate all modelling choices for the underlying, provided that the transition density is either analytically known in the direct space or in some transform space (Fourier or Laplace), or if we can compute an approximation of the transition density at given grid values of the underlying processes (\cite{chen2014advancing}), for instance applying the estimation methods introduced by \cite{ai2006saddlepoint}, \cite{ait2007maximum}, \cite{li2013maximum}, \cite{Guay2016aa}.
Fourier transforms of transition probabilities describe price evolution in affine models
(\cite{Duffie:2000a}), quadratic models (\cite{Leippold:2002a}, \cite{cheng2007linear}), and variance gamma and Levy models (\cite{madan1998variance}, \cite{carr2003stochastic}).
We show how adding a jump component to a stochastic volatility diffusion for the underlying assets is as simple as an element by element multiplication of the matrices describing the transition probabilities of the two separate components.
In addition, projecting transition densities and value functions on appropriate function bases, allows us to naturally describe the filtration of information at given points in time, thus correctly addressing non-Markovian models, such as stochastic volatility models, in the stock price dimension.
This feature is the hurdle \cite{o2005path} and \citet{lord2008fast} could not cross but overcome by \cite{chen2014advancing}.
However, although the latter recognise the possibility of using characteristic functions and FFT with Heston (and this can also be achieved for Heston plus Merton jumps), they focus on demonstrating their transition density approximation technique.
The way we build the transition matrices makes the approximation error of these matrices independent of the time horizon. The number of time steps required is solely driven by features of the contracts, such as dates at which the contract needs to be monitored, and not by mesh size requirements in the time dimension. We can evaluate transitions of the value functions for arbitrary time steps, whereas \cite{chen2014advancing} can at most address time steps equal to 0.1 year if, instead of a characteristic function, they use their approximation which adds intermediate steps that slow down their pricing algorithm.
Moreover in our approach the transition matrices enjoy some space and time homogeneity properties that make the computation of their entries appealingly simple. The recursive nature and time and space homogeneity are not affected by intermediate cash flows as for instance dividend payments.
The numerical contribution of our paper is a general stochastic optimal control method with its matrix form having an appealing interpretation in terms of the stochastic discount factor. In a work that shares some intuitions with our approach, \cite{hodder2007} solve an optimal control problem for hedge fund management which is both time and variable dependent by discretising the univariate control variable on an equally spaced grid.

Finally, we think that the empirical study that we carry out speaks in favour of speed and ease of implementation of the recursive projections. 
To illustrate the speed and ease, we need only two matrix-per-vector multiplications and less than a second to evaluate an American call option and its Greeks on a stock paying one dividend before the expiry date in the Black-Scholes case.  The competing method, an improved tree methodology introduced by \cite{vellekoop2006efficient}, is roughly an order of magnitude slower. Similarly, a few seconds are sufficient to evaluate an American call option and its Greeks on a stock paying three dividends before the expiry date in the Heston model, whereas the finite-differences method alternative would again take an order of magnitude more time to achieve the same precision. To date, no empirical work on options written on dividend-paying stocks exists outside the Black-Scholes world, given that the existing methods are simply too time-consuming.
Indeed, \cite{broadie2007model}, state that ``The computation time required for American options makes calibration to a very large set of options impractical.'' Thanks to the recursive projections, we show that this is feasible\footnote{As reported earlier in the text, what \cite{broadie2007model} show in their appendix A is that transforming American options to European
ones does not matter for calibration purposes in their application, but we show in the following  that it does make a difference in our study.}.

The paper is organized as follows. The following section reviews the most recent advances in numerical option pricing.
In Section \ref{sec:example}, we develop an introductory example based on the Black-Scholes model and present some preliminary numerical results that show the advantages of our technique.
In Section \ref{sec:RecProj}, we study the general case of valuation by fast recursive projections
in order to include standard affine models.
We present numerical illustrations for the Black-Scholes and Heston models.
We also provide the theoretical convergence of the computed option price in terms of the sampling frequency, and characterize the convergence rate of the computed option price.
In Section \ref{sec:boundaries_empirics}, we present the innovative applications of our algorithm. In Section \ref{sec:early_boundaries}, we characterize the early exercise boundary under different modelling assumptions. In Section \ref{sec:empirics}, we present the empirical results concerning the cost of failing to optimally exercise American call options.
Section \ref{sec:conclusion} gathers some concluding remarks. The supplementary online Appendix contains the proofs of the propositions, additional comparisons with existing methods and gives a detailed description of the data and the calibration procedure.

\section{Review of alternative pricing methods}\label{sec:littReview}
As explained in the introductory section above, the numerical methods for option pricing fall into the following two categories: differential methods and integral methods. Differential methods provide the solution to the pricing problem by numerically approximating the associated partial differential equation (see the references of the previous section). Integral methods approximate the conditional expectation, giving the arbitrage-free value of a financial contract. Some methods lie at the boundary of the two groups.
The binomial tree technology, and lattice methods in general (\cite{cox1979option}, \cite{Broadie:1996aa}, \cite{vellekoop2006efficient}),
although similar to finite differences in their implementation (\cite{rubinstein2000relation} shows that they are basically the same), belong to the integral family of methods, given that they provide a discrete approximation of a conditional expectation.
\cite{Barone-Adesi:1987aa} start from the partial differential equation representation in the Black-Scholes case, and contribute a closed-form approximation for an American option with continuous dividend yield.

In this section, we focus on integral methods because the main recent contributions, including ours, fall within this family. The integral representation arises from the observation that option valuation in arbitrage-free economies amounts to using linear operators that assign prices today to payoffs at futures dates. These linear operators are conditional expectation operators. In multi-period economies, the temporal consistency in valuation ensures that the family of such operators satisfies a semigroup property (\cite{Garman:1985aa}). Because most valuation models correspond to Markov environments for the price of the underlying assets, the semigroup property is formalized by the law of iterated expectations restricted to Markov processes (\cite{Hansen:2009aa}), where the Markov states are indexed by the time horizons of potential exercise. From a computational perspective, pricing algorithms attempt to approximate the following: (i) the conditional expectation with respect to the information at one potential exercise date and (ii) the recursive relationship between conditional expectations at different potential exercise dates.
\cite{Geske:1984aa} express the problem of pricing an option with discrete exercise dates as a multivariate integral.
\cite{Kim:1990aa}, \cite{jamshidian1992analysis}, and \cite{carr1992alternative} extend this approach to continuous time and express the early exercise premium of American contracts as an integral term whose interpretation is the risk-neutral valuation of a continuous cash flow. This additional integral is a function of the exercise boundary.
\cite{Bunch:2000aa}, \cite{huang1996pricing}, \cite{Ju:1998aa}, and \cite{ibanez2003robust} contribute algorithms to compute the free boundary and the option prices in the Black-Scholes model environment.
\cite{detemple2002valuation} extend these contributions to diffusion processes with stochastic interest rates. These approaches cannot easily accommodate discrete dividends.
\cite{roll1977analytic}, \cite{geske1979note}, and \cite{whaley1981valuation} provide a closed-form approximation of pricing an
American option paying a single discrete dividend in the Black-Scholes framework. \cite{medvedev2010pricing} develop an approximation of American option prices under stochastic volatility and stochastic interest rates using short-term asymptotics (for the Black-Scholes case, see \cite{Lamberton:2003aa}). These approximations are accurate and flexible, but cannot accommodate discrete dividends. They are mainly suited for index and exchange rate options.

Another approach in approximating the pricing operator is simulating the trajectories of the underlying asset to compute the conditional expectation. One of the first simulation-based analyses of American options in the Black-Scholes case is that of \cite{broadie1997pricing}. A recent stream of Monte Carlo methods started with the work of \cite{longstaff2001valuing}, developing into the so-called duality approach of \cite{haugh2004pricing}, \cite{rogers2002monte}, and \cite{andersen2004primal}. The duality approach method uses simulations to find a lower bound and an upper bound of the true price.
In this sense, it is an extension of \cite{Broadie:1996aa}. As shown in the numerical experiments of \cite{andersen2004primal}, the interval can be tight, which makes the method very precise. The advantage of this approach is that it can handle virtually any type of process dynamics, state variable structure, and payoff specification. However, to achieve precision, as in all simulation-based methods, the discretization step in the time dimension must be small (except trivial cases like the geometric brownian motion), making computations lengthy. \cite{desai2012pathwise} refine the duality approach by contributing a new algorithm to compute the upper bound, improving it in terms of speed.

We can express conditional expectations in terms of series developments on appropriate basis functions, which is the approach we follow when projecting the Green and value functions on localized basis functions in the subsequent sections. We can already find simple and tractable projections of complex payoffs in the literature via the representation on a set of basis functions (such as polynomials in \cite{madan1994contingent}, \cite{lacoste1996wiener}, \cite{darolles2000approximating}).
\cite{chiarella1999evaluation} suggest a fast recursion projection method based on Hermite polynomials in the Black-Scholes model (for an extension of the model in \cite{merton1976option}, see \cite{chiarella2005pricing}).
They also explain how their method provides a viable numerical method for the implementation of the path integral approach to option pricing as described in \cite{linetsky1997path}, for example. Our paper is an extension of the former in that we use sampling instead of projections to improve speed. This family of models typically assumes a geometric Brownian motion for the underlying asset, with some exceptions accommodating for jumps but not for stochastic volatility.

In a more recent application of Hermite polynomials to pricing,  \cite{xiu2014hermite} provides a general, closed-form approximation of European, but not path-dependent, option prices by using a finite-term expansion of the transition density (for related work on expansions, see also \cite{ait2007maximum, ait2010estimating}). \cite{kristensen2011adding} approximate the option price by expanding the difference between the true model price and the Black-Scholes price.
They can price some specific path-dependent options, such as barrier options but cannot address more general path-dependent payoff features, such as discrete dividends.
Our recursion projection method shares some features with the dynamic programming approach of \cite{ben2009dynamic} in that they also approximate the value function on a fixed grid and interpolate with local polynomials to reconstruct the value surface. They assume a GARCH process for the stochastic volatility and obtain closed-form formulas for the conditional expectations to compute the grid of values at a previous time. Though computationally efficient, their method is restricted to the GARCH family of processes and does not naturally accommodate dividends.

As explained above, our method shares features with the quadrature family of models. In addition to the contributions cited above, \cite{jackson2008fourier} introduce a Fourier space time stepping algorithm to price Bermudan options under Levy processes. This efficient method is not extendible to the stochastic volatility framework.

\section{An introductory example}\label{sec:example}

In this section, we show how a pricing problem for a European payoff in the Black-Scholes model translates into a functional projection. The derivation uses elementary calculus concepts. Then, we explain how we exploit the projection to build fast recursive schemes to value path-dependent options such as Bermudan and American options when the stock pays discrete dividends. We design this introductory example to emphasize the intuition underpinning our approach.

\subsection{Description of the method: the European case }\label{sec:example_european}
Let $S_t$ be the price of the underlying asset at date $t$ and assume that interest rates are constant to facilitate
exposition. For $S_t = x$, the value function $V(x,t)$ for a European option is given by the following
conditional expectation:
\begin{equation}\label{eq:condexp}
	V(x, t)=\mathbb{E}\big[e^{-r(T-t)}H(S_T, T)\big|S_t = x\big],
\end{equation}
where $H(S_T,T)$ is the payoff function expressed as a function of time $T$ and of the value of
the underlying asset $S_T$ at maturity date  $T$, and $r$ is the constant risk-free interest rate. When the  pricing operator in \eqref{eq:condexp}
admits a state price density $G(x,t; y, T)$, the so-called Green function, which is the discounted value of the transition probability density from point $x$ at time $t$
to point $y$ at time $T$, we obtain the familiar integral form:
\begin{equation}\label{eq:integral}
	V(x, t) = \int \mspace{-3mu}  G(x,t; y, T) H(y, T) dy.
\end{equation}
Now consider a regularly spaced grid of points $\{y_1, y_2, \dots, y_N\}$ that defines a finite interval $D=[y_1-\Delta y/2, y_N +\Delta y/2]$, with $\Delta y= y_{j} - y_{j-1}$.
We know that, under appropriate regularity conditions, the integral \eqref{eq:integral} restricted to the interval $D$ can be approximated by the Riemann sum as follows:
\begin{equation}\label{eq:Riem_sum}
	V(x, t) \sim \sum_{j=1}^{N-1}  G(x, t; y_j, T)\, H(y_j, T)\,\Delta y, \qquad t<T,
\end{equation}
where the $\sim$ symbol means that the right hand term converges to the left hand term as \mbox{$\Delta y \to 0$}.
The representation \eqref{eq:Riem_sum} is known as the `rectangle method' in standard integral calculus.
If we are interested in computing the value of $V(x,t)$ on a regularly spaced grid of points
$\{x_1, \dots, x_M\}$, for instance to plot the value function of the contract,
or to compute the Greeks, we can express \eqref{eq:Riem_sum} in a matrix form as follows:
\begin{equation}\label{eq:matrix}
	\mathbf{V}(t)  \sim  \mathbf{v}(t)=\mathbf{G}(t;T)\mathbf{H}(T),
\end{equation}
where \mbox{$\mathbf{V}(t) =  \big(V(x_1,t), \dots, V(x_M,t))'$},  $\mathbf{v}(t)$ is the approximation of $\mathbf{V}(t)$  obtained from the $N \times 1$ vector
$\mathbf{H}(T)$  with entries
 $\mathbf H_j= H(y_j,T)$
and the $M\times N$ matrix $\mathbf{G}$  with entries \mbox{$\mathbf{G}_{ij} = \Delta y\, G(x_i,t; y_j,T)$}.
Equation \eqref{eq:matrix} describes a discretization of the economy.
We can think of the matrix $\mathbf{G}$ as a matrix of Arrow-Debreu prices,
in which the rows represent discrete states $\{x_i\}_{i=1,\dots,M}$ at
the current date $t$, and the columns represent discrete states
$\{y_j\}_{j=1,\dots, N}$ at the future date $T$.
We can interpret the column vectors $\mathbf{V}(t)$ and $\mathbf{H}(T)$ as vectors of state contingent values at dates $t$ and $T$, respectively.
The matrix operator $\mathbf{G}$ discounts future payoffs at date $T$ to current prices at date $t$ (\cite{Garman:1985aa}).

Let us more closely examine the elements of $\mathbf{H}(T)$. For readability, we define
$\underline{y}_i = y_i - \frac{\Delta y}{2}$ and $\overline{y}_i = y_i + \frac{\Delta y}{2}$.
We can interpret the value $ H(y_j, T)$  as an approximation of the quantity
\begin{equation}\label{eq:InnerProduct}
	(1/\Delta y)\int_{\underline{y}_j}^{\overline{y}_{j} }\mspace{-9mu}H(y,T)dy= (1/\sqrt{\Delta y})\int \frac{1}{\sqrt{\Delta y}} \mathbb{I}_{\underline{y}_j, \overline{y}_{j} } H(y,T)dy \defeq \frac{1}{\sqrt{\Delta y}}\int\mspace{-3mu} e_{j}(y)H(y,T)dy,
\end{equation}
where $e_{j}(y)=\frac{1}{\sqrt{\Delta y}}\mathbb{I}_{\underline{y}_j, \overline{y}_{j} }$, and where $\mathbb{I}_{\underline{y}_j, \overline{y}_{j} }$ is the indicator function of the interval $[\underline{y}_j, \overline{y}_{j})$.  Because
$\{e_{j}(y)\}_{j=1,\dots,N}$ is an orthonormal set given
the standard $L^2$ inner product $\langle f,g\rangle = \int \mspace{-6mu}f(x)g(x)dx$, we can view the entries of the vector $\mathbf{H}(T)$ as an approximation
of $(1/\sqrt{\Delta y})$ times the coefficients of the decomposition of the payoff function $H(y, T)$
on the set of orthonormal indicator functions defined by the grid $\{y_1, \dots,y_N\}$.
A similar argument can be applied to the coefficients of the $\mathbf{G}(t;  T)$ matrix: every
row of $\mathbf{G}(t;T)$ is given by an approximation of $\sqrt{\Delta y}$ times the coefficients of the projection of the conditional density
$G(x_i, t; y, T)$ on the same orthonormal set.
The different choices in  the normalization factors for the entries of $\mathbf{G}$ and $\mathbf{H}$
are justified by our wish to interpret all of the quantities appearing in \eqref{eq:matrix} as prices.
We have that $V(x,t)$ and $H(y,T)$ are value functions and $G(x,t;y,T)$ is a state price density.

Overall, we can interpret the numerical approximation of the integral
 \eqref{eq:integral} as a projection of the functions on an orthonormal basis.
Such an interpretation is key
in the generalization of the recursive projection approach to more sophisticated models in the next section. In our case, the functional projection boils down to sampling the given functions
on a grid of $N$ points $\{y_j\}_{j=1,\dots,N}$ via a discrete transform.
From a computational perspective, the entries of $\mathbf{H}(T)$ summarize how payoffs depend on the price of the underlying assets at future exercise dates,
and the entries of $\mathbf{G}(t;T)$ summarize how the price of the underlying assets transits from one state to another according to the elapsed time between exercise dates. Figure \ref{fig:projection-1d} gives a graphical representation of the two computational steps of our fast recursive projection approach, (i) the projection step, and (ii) the recursive step. The value function at $t$, on the right,
 and the state price density for a given value of $x=S_{t}$, on the left, are sampled (the projection step); the obtained arrays of values are multiplied element by element, and the products are summed to obtain the value of
 $V(x, t)$ (the recursive step).

	\medskip
	{\centering	
	[Figure  \ref{fig:projection-1d} about here]\medskip

	}

\subsection{Description of the method: the path-dependent case }\label{sec:exampl_PathDep}
Let us now address the valuation of path-dependent contracts. We start by considering a Bermudan option. We consider a set $\{t_1=t, \dots, t_L=T\}$ of exercise dates. At each $t_l$, the holder of a Bermudan option may decide whether to exercise. He exercises if the intrinsic value \mbox{$H(S_{t_l}, t_l) \defeq (S_{t_l} - K)_+= \max\{S_{t_l} -K,0\}$} is higher than the value of keeping the option, i.e., the continuation value. Bermudan options are the ideal building blocks for studying American call
options on dividend-paying stocks. It is well known that it can be optimal to exercise American call options immediately before ex-dividend dates
$\{t_h, t < t_h < T\}_{h=1,\dots, H}$; for instance, see \cite{pool2008failure} for a discussion on early exercise strategies. The implication is that we must monitor the option value function $V(x,t)$ immediately before the ex-dividend dates, when the intrinsic value $(S_{t_h-\epsilon} - K)_+$  for a small $\epsilon >0$ can be larger than the continuation value $V(S_{t_h},t_h)$. Then an American call option shares with a Bermudan option the feature that its value function must be evaluated only at a finite number of dates.

 The semigroup property of the pricing operator ensures that we compute the value function $V(x,t)$ of a Bermudan option recursively.
The recursion consists of moving backwards in time and computing at each $t_l, l=1, \dots , L-1$:
\begin{equation}\label{eq:rec_exp}
	V(x,t_l) = \max\bigl\{H(x,t_l), \mathbb{E}\bigl[e^{-r(t_{l+1} - t_{l})}V(S_{t_{l+1}}, t_{l+1})|S_{t_l}=x\bigr]\bigr\},
\end{equation}
with the boundary condition $V(y, t_L)=H(y,t_L) $.
To speed up the recursion, we impose the condition that the grid of values  $\{y_j\}_{j=1,\dots,N}$ at which we sample the function $V(y, t_{l+1})$  and the grid  $\{x_i\}_{i=1,\dots, M}$ at which we compute the function $V(x, t_l)$, coincide at each exercise date, which means that  $M=N$. From now on, we use the $x$ variable as a generic conditioning value, e. g. the value of the underlying at date $t$, as in $V(x,t)$. If $x$ takes a specific value in the grid $\{y_{i}\}_{i=1,\dots,N}$, then we write $V(y_{i},t)$.
Then, in the matrix notation of the approximation, we obtain the following:
\begin{equation}\label{eq:rec_matr}
	\mathbf{V}(t_l)  \sim \mathbf{v}(t_l)= \max \bigl\{\mathbf{H}(t_l), \mathbf{G}(t_l;t_{l+1})\mathbf{v}(t_{l+1})\bigr\},
\end{equation}
and the approximation $\mathbf{v}({t_l})$ is the input for the following recursion step to compute an approximation $\mathbf{v}({t_{l-1}})$  of  $\mathbf{V}({t_{l-1}})$.
Figure \ref{fig:projection-1dmatrix} shows how combining projection and recursion steps with identical grids at each exercise date translates the pricing problem of path-dependent options into a sequence of matrix time vector multiplications.

	\medskip
	{\centering	
	[Figure  \ref{fig:projection-1dmatrix} about here]\medskip
	
	}

The convergence properties of $\mathbf{v}({t_l})$ to $\mathbf{V}({t_l})$  are formally established in Proposition \ref{prop:convBS} of Section \ref{sec:RecProj}.
From \eqref{eq:rec_matr}, it is clear that computations only occur at the exercise dates defined in the Bermudan contract, and do not require any input at any other point in time.
Furthermore, under the previous additional constraints on the grids, if the time
interval $\tau = t_{l+1} - t_l, l=1, \dots, L-1$, is constant and if the pricing operator
enjoys a stationarity property (time translation invariance), then the matrix $\mathbf{G}(t_l;t_{l+1})=\mathbf{G}(\tau)$
has constant entries, and the algorithm only involves one single computation of the matrix.

The methodology easily extends  in the presence of discrete dividends
paid on potential exercise dates\footnote{We can easily extend to the case in which the ex-dividend date does not belong to the set of exercise dates. We do not explicitly state this case, given that we are mainly interested in Bermudan options, as a building block for studying American options.}.
The implication is that the set of ex-dividend dates is a subset of the Bermudan exercise dates and that we have
$\{t_{h}\}_{h=1,\dots,H} \subset \{t_{l}\}_{l=1,\dots, L}$.
    We only need to add the dividend $\delta$ to the continuation value in  (\ref{eq:rec_exp}).
Hence, in order to price an American option on a dividend-paying stock, Equation \eqref{eq:rec_matr} must be modified by sampling the state price density $G(x,t_{l};y,t_{l+1})$ at
the grid $\{y_i-\delta(y_{i})\}_{ i=1,\dots, N}$ for the conditioning  value $x$, whenever $t_{l}\in \{t_{h}\}_{h=1,\dots,H}$. The entries of the matrix $\mathbf{G}(t_{l}; t_{l+1})$ then become $\mathbf{G}_{ij}=G(y_i-\delta, t_l; y_j, t_{l+1})\Delta y$. Given the freedom in choosing where to sample $G$, $\delta(x)$ could be any function of $x$. If $\delta(x)=r_d \, x$, then we can accommodate for a proportional dividend. If $\delta(x)=d$, then we
can accommodate for  a discrete dividend amount $d$\ \footnote{When we consider discrete dividends, we rule out arbitrage situations where the dividend is too large with respect to the stock price (see \cite{haug2003back} for a discussion).}. If $\delta(x)=0$, then we return to the Bermudan option case. The value function $\mathbf{V}(t_{h})$ still gives the value of the contract at the grid points $\{y_{1}, \dots, y_{N}\}$; thus we can use its approximation $\mathbf{v}(t_{h})$ as the input for the following step of the algorithm, and the recursive property of the algorithm is maintained.
Figure \ref{fig:dividends}, shows how the recursive scheme changes to accommodate for dividends.

	\medskip
	{\centering	
	[Figure  \ref{fig:dividends} about here]\medskip

	}

The early exercise decision for American long put holders is more complicated. We consider a Bermudan put with exercise dates $\{t_1,\dots,t_L\}$, and assume that the dates $\{t_{h}\}$ at which dividends are paid form a subset of the exercise dates. Then, by taking $L$ large,
our approach also provides a quick approximation for pricing American put options on single stocks
paying discrete dividends.
In the Black-Scholes model, the state price density is known in the following closed form:
\begin{equation}\label{eq:GreenBS}
	G(x,t; y, T)= \frac{1}{y}\frac{e^{-r(T-t)}}{\sqrt{2\pi \sigma^2(T-t)}}\exp\biggl(-\frac{\bigl(\log y - \log x - (r-\sigma^2/2)(T-t)\bigr)^2}{2 \sigma^2(T-t)}\biggr),
\end{equation}
where $\sigma$  is the volatility and the pricing operator satisfies the stationarity property. Hence, we expect to obtain a fast, simple and accurate numerical algorithm.
Before reporting the numerical results, we stress that we have kept the setup of this introductory example as simple as possible to emphasize intuition by limiting technical details. Although the computational results already speak in favor of our construction, it is in the following sections that the recursive projections method deploys its full potentiality with more complex price dynamics.

\subsection{Numerical illustrations in the Black-Scholes model}\label{sec:BSillustration}

As a first numerical example in the Black-Scholes framework, we compare the convergence speed of a binomial tree and of the recursive projections method in pricing an American call option on a dividend-paying stock\footnote{All of the codes are written in C++. The codes are available from the authors upon request.}.
Two popular modeling choices for the dividend payment are a known  cash amount $d$ or a known dividend yield $r_d$. The latter is computationally friendly because it leads to a recombining tree. However, the empirical evidence shows that corporations tend to commit to paying out fixed amounts at regular dates and to smooth their dividends rather than adjusting them  downwards and signaling a decrease in cash flows (for a signaling based theory on dividend policy see, for instance, \cite{Miller1985zr}). These arguments justify our preference for modeling dividends as fixed known amounts rather than as given yields. The known dividend amount assumption does not lead to a recombining tree, and a new tree is originated at each node following an ex-dividend date, increasing the numerical complexity of the problem. The work of \cite{vellekoop2006efficient} provides a recent enhancement of the classical binomial tree method which incorporates discrete dividend payments through an approximation of the continuation value of the option at the ex-dividend dates. This new algorithm has been proven to be substantially faster than the standard non-recombining binomial tree, and is therefore a reliable benchmark for this simulation exercise.


%

	\medskip
	{\centering	
	[Figures  \ref{fig:Vellekoop} and  \ref{Figure:BS_amcall_div} about here]\medskip

	}

Figure  \ref{fig:Vellekoop} compares the convergence speed of the enhanced binomial tree and that of the recursive projections method in pricing an American call option on a discrete dividend-paying stock. The option has a maturity of $T=3$ years and a dividend $d=2$ is paid out at the end of each year.  Other parameters, namely the interest rate, volatility and strike price,  are set equal to $r=0.05$, $\sigma=0.2$,  and $K=100$, respectively. We compute 3 prices: at-the-money, in-the-money and out-of-the-money, corresponding to $S_0=80, 100$, and $120$, respectively.
The \emph{true} values of 7.180, 18.526, and 34.033 are obtained with $10000$ time steps in the binomial tree. The graphs show that, across the three different values of $S_{0}$, the recursive projections enjoy an increase of speed of approximately a factor $10$ for a comparable level of precision. The speed advantage is even larger if we consider that a new tree is needed for each value of $S_{0}$. Instead, the recursive projections method delivers the entire value function $\mathbf{v}(0)$ at once in a straightforward manner. This feature is particularly useful in computing Greeks through numerical differentiation.
As an additional benchmark, Figure \ref{Figure:BS_amcall_div} displays the convergence speed of the recursive projections jointly with the one of a standard non-recombining tree. Even though the non-recombining tree is known to be an inefficient method, it is still used as a common reference point in the literature, and we show this graph for comparison purposes. We can see that the gain of speed of the recursive projection is of the order of $10^{4}$.
In Section E of the supplementary online Appendix, we extend the above pricing exercise, and further compare the recursive projections with the Monte Carlo based method by \cite{longstaff2001valuing}. The recursive projections are faster than the Longstaff Schwartz method by at least four orders of magnitude.
As an aside,  for $S_{0}=100$, if we approximate the known constant dividend $d=2$ with a known continuous dividend yield\footnote{The yield is obtained by considering the dividends paid at $t=1$ and $t=2$ only, because the dividend paid at $t=3$ has no impact on the price of the option. Considering a dividend yield of $2\%$ would provide an option value of 16.857, which is a much larger error.} $r_d=0.013$, then a binomial tree with $10000$ steps delivers a value of 18.213 instead of 18.526, with a relative error of approximately $169bp$.
This error is far above observed bid-ask spreads. This simple example points to the importance of using models that can explicitly address discrete dividends in empirical analysis, instead of using approximations based on continuous dividend yields.
In Section \ref{sec:boundaries_empirics}, we more extensively analyze the impact on the exercise boundary of the choice of a continuous dividend yield or a discrete dividend.

Moreover, we have chosen a sampling scheme that is equivalent to projecting the payoff function on a set of basis functions that are well localized, in the sense that their support is a closed interval. The implication is that local features of the payoff function, such as a discontinuity, are described by the coefficients relative to one or at most two basis functions, those lying next to the discontinuity.
This description avoids a noisy approximation induced by spurious oscillations when projecting discontinuities on basis functions defined on the entire domain, such as the Fourier sine-cosine basis or the Hermite polynomial basis.
From a computational perspective, this property translates into an accurate approximation even for payoffs with strong discontinuities, such as a digital payoff  $H(S_{t_l}, t_l) = \mathbb{I}_{S_{t_l>K}}$ in a Bermudan digital call option. The discontinuity may introduce noise at most in the coefficient relative to the indicator function of the interval in which the discontinuity is located.
The noise is completely eliminated if we make sure that that the strike value lies in between two consecutive grid points, so that the discontinuity is at the boundary between two consecutive indicator functions.
In this numerical example, we use the standard binomial tree as a benchmark, since the method of \cite{vellekoop2006efficient} provides no advantage in the absence of dividends.
Figure \ref{Figure:BS_berm_digitalcall} (see the caption of the table for the values of the parameters of the example) shows that the binomial tree has problems capturing the discontinuity in the payoff function. Consequently, an extremely slow convergence of the tree method for at-the-money Bermudan digital call options is yielded. The recursive projections are also at least an order of magnitude faster in pricing the out-of-the-money options. The apparent non-monotonic convergence of the binomial tree for $S_{0} = 120$ is because both methods achieve a quick convergence for in-the-money options, and the graph only displays small oscillations on the order of half a basis point around the true value.

	\medskip
	{\centering	
	[Figure \ref{Figure:BS_berm_digitalcall} about here]\medskip

	}

\section{Valuation by fast recursive projections}\label{sec:RecProj}

In this section, we generalize the approach developed in the introductory example of Section \ref{sec:example_european}. We consider a set of exercise dates $\{t_l\}_{l=1,\dots, L}$, with $t_{L}=T$.
The starting point is the following identity:
\begin{align}\notag
	V(x,t)
&=
 \int \mspace{-3mu} G(x,t ; y, T)H(y,T) dy
=\mspace{-9mu}
\int \mspace{-6mu}
\Bigl(\sum_{i=-\infty}^{\infty} \mspace{-6mu}
\langle G , \varphi_{i}\rangle  \varphi_{i}(y)\Bigr)
\Bigl(\sum_{j=-\infty}^{\infty} \mspace{-6mu}
\langle H , \varphi_{j}\rangle \varphi_{j}(y)\Bigr)
 dy \\
&=\mspace{-9mu}
\sum_{i,j=-\infty}^{\infty} \mspace{-6mu}
\langle G , \varphi_{i}\rangle
\langle H , \varphi_{j}\rangle
 \mspace{-6mu}
 \int \mspace{-6mu} \varphi_{i}(y)\varphi_{j}(y) dy  \label{eq:generalprojection}
	 =
 \sum_{j=-\infty}^{\infty}
\langle G,\varphi_{j}\rangle
\langle H ,\varphi_{j}\rangle,
\end{align}
where $\{\varphi_{j}\}_{i\in \Z}$ is a generic localized orthonormal basis in $L^2$,  that is the support of each function $\varphi_{j}$ is a finite closed interval.  $\langle G(x,t ; \cdot, T) \varphi_{j}(\cdot)\rangle$ and $\langle H(\cdot, T),\varphi_{j}(\cdot)\rangle$ (abbreviated in $\langle G,\varphi_{j}\rangle$ and $\langle H, \phi_{j}\rangle$ respectively) are the coefficients of the linear projection of $G$ and $H$ on $\{\varphi_{j}\}_{i\in \Z}$. The valuation by fast recursive projections is based on two steps: a projection step and a recursive step.
The projection step consists of projecting $G$ and $H$ at time $T$ on the basis $\{\varphi_{j}\}_{j\in \Z}$ and computing the coefficients $\langle G,\varphi_{j}\rangle$ and $\langle H,\varphi_{j}\rangle$. Due to finite support of the localized basis functions $\{\varphi_{j}\}_{j\in \Z}$, all the coefficients $\langle H,\varphi_{j}\rangle$ computed as the inner product between the function $H$ and the localized basis function $\varphi_{j}$ are automatically well defined, even if $H$ is not $L^2$-measurable. The recursive step consists of transmitting the coefficients back in time by computing the final sum in Equation (\ref{eq:generalprojection}).
We develop the method by considering the orthonormal set of indicator functions $\{e_{j}\}_{j\in \Z}$, $e_{j}=1/\sqrt{\Delta y}\, \mathbb{I}_{\underline{y}_{j}, \overline{y}_{j}}$ for an equally spaced grid $\{y_{i}\}$ with step $\Delta y$ and $\underline{y}_{j}=y_{j}-\Delta y/2$, $\overline{y}_{j}=y_{j}+\Delta y/2$ . The set $\{e_{j}\}_{j\in \Z}$ is a basis in $L^2$ when $\Delta y \to 0$.
Then, as shown in Equation \eqref{eq:InnerProduct}, we can approximate the coefficients $\langle H, e_{j}\rangle$ and $\langle G, e_{j}\rangle$ with $\sqrt{\Delta y}H(y_{j},T) $ and  with $\sqrt{\Delta y}\, G(x,t ; y_{j}, T)$. Thus, the projection step simply consists of sampling the relevant functions at the grid points $\{y_{j}\}_{j\in \Z}$.
The recursion step stems from the observation that, if $x$ in Equation \eqref{eq:generalprojection} takes values in the grid $\{y_{j}\}_{j\in \Z}$, then $v(y_{i},t_{L-1}) = \sum_{j=-\infty}^{\infty} G(y_{i},t_{L-1} ; y_{j}, T)H(y_{j}, T)\Delta y$ and $v(y_{i},t_{L-2}) = \sum_{j=-\infty}^{\infty} G(y_{i},t_{L-2} ; y_{j}, t_{L-1})v(y_{j}, t_{L-1})\Delta y$ are approximations of $V(y_{i}, t_{L-1})$ and $V(y_{i}, t_{L-2})$. That is, there is a recursive linear expression that connects the value of the contract at consecutive times and at different points of $\{y_{j}\}_{j\in \Z}$.
Here, all the functions are sampled at the same grid $\{y_{j}\}_{j\in \Z}$ for each potential exercise date $\{t_{l}\}_{l=1,\dots,L}$. Doing so allows the matrix representations of Equation \eqref{eq:rec_matr}.
In the following sections, we characterize the projection step and the recursive step in two different frameworks.
In Sections \ref{sec:ProjectionBS} and \ref{sec:RecursBS}, we consider the case in which the state price density is known in closed form, as in the Black-Scholes example developed in Section \ref{sec:example} or in the Merton jump-diffusion model. In Sections \ref{sec:ProjectionHes} and \ref{sec:RecursHes}, we characterize the projection step and the recursion step when the characteristic function of the state price density has a known analytic form, as in the Heston model.
In our presentation, we generically refer to this second class of models as to the ``stochastic volatility'' case. The methodology developed in Sections \ref{sec:ProjectionHes} and \ref{sec:RecursHes} covers affine jump-diffusion models and Levy models, in which we price by transform analysis.
Here, the key requirement is to be able to numerically compute the characteristic function of the state price density. In Section \ref{sec:NumericalHeston}, we provide numerical examples in the Heston case.

\subsection{Merton-Black-Scholes: the projection step}\label{sec:ProjectionBS}

The projection step is based on an approximation of the payoff function and of the state price density by the set of orthonormal functions $\{ e_{j} (y) \}_{j\in \Z}$:
\begin{align} \label{eq:projection}
	\tilde H(y,T) &= \sqrt{\Delta y}\, \sum_{j=-\infty}^{\infty} H(y_{j},T) e_{j}(y), \\ \label{eq:projection2}
	\tilde G(y_{i},t; y, T) & = \sqrt{\Delta y}\, \sum_{j=-\infty}^{\infty} G(y_{i},t; y_{j}, T) e_{j}(y).
\end{align}
The values of $\tilde H(y,T)$ and $\tilde G(y_{i}, t; y, T)$ in each interval $[\underline{y}_{j}, \overline{y}_{j})$ are given by the values of the functions $G$ and $H$ sampled in the mid point $y_{j}$ of each interval. Then, we have that $\tilde H(y,T)$ and $\tilde G(y_{i}, t; y, T)$ are piecewise constant approximations of $H(y,T)$ and $G(y_{i}, t; y, T)$.
In the algorithm, we only need the quantities $\{H(y_{j},T)\}_{j\in \Z}$ and  $\{G(y_{i},t; y_{j}, T)\}_{j\in \Z}$.
We use the full representations \eqref{eq:projection} and \eqref{eq:projection2} in Section A of the supplementary online Appendix to prove Proposition \ref{prop:convBS}.
There, we show that relying on the sampling approximations \eqref{eq:projection} and \eqref{eq:projection2} instead of the orthogonal projections given by the inner products as in Equation \eqref{eq:generalprojection}, leaves the convergence properties of the pricing algorithm unaffected.

%

\subsection{Merton-Black-Scholes: the recursive step}\label{sec:RecursBS}

The following proposition gives a recursion formula that relates the approximated values of the option at different points of the grid $\{y_{j}\}_{j=-\infty}^{\infty}$ and at different points in the set of exercise dates $\{t_{l}\}_{l=1,\dots,L}$. It also states the rate of convergence of the algorithm.

\begin{proposition}\label{prop:convBS}
Let $H(y,T)$ be such that $|H(y,T)-H(y',T)|<C|y-y'|$ for a positive constant $C$, and for $|y-y'|<\Delta y$.
Furthermore, let $v_{i}(t_{l})$ be defined for a set of dates $\{t_{l}\}_{l=1,\dots,L}$, with $t_{L}=T$, as follows:
\begin{alignat}{3}\label{eq:recursiveBS}
	v_i (t_{l}) &= \max\Bigl\{H(y_{i}, t_{l}), \sum_{j \in \Z} G(y_{i}, t_{l};y_{j}, t_{l+1})H(y_{j}, t_{l+1}) \Delta y \Bigr\},		&	\quad  	& \mathrm{for}	& \quad l&=L-1 ,	\\
	\label{eq:recursiveBS2}
	v_i (t_{l}) &= \max\Bigl\{H(y_{i},t_{l}), \sum_{j\in \Z} G(y_{i}, t_{l};y_{j}, t_{l+1})v_{j}(t_{l+1})\Delta y \Bigr\}, 	&			& \mathrm{for}	& l&=1,\dots, L-2.
\end{alignat}
Then, for each $t_{l}$ in $\{t_{1},\dots, t_{L-1}\},$ the approximated values $v_{i}(t_{l})$ defined in \eqref{eq:recursiveBS} and \eqref{eq:recursiveBS2} converge to the true value $V(y_{i}, t_{l})$ with an approximation error of the order $O\bigl((\Delta y)^{2}\bigr)$.
\end{proposition}
\begin{proof}
See Section A of the supplementary online Appendix.
\end{proof}

In principle, the continuity condition on $H(y,T)$ rules out digital payoffs. We need the condition to hold within each interval $\Delta y$ centred on the grid $\{y_{j}\}_{j=-\infty}^{\infty}$. We can still price digital options by ensuring that the strike value lies in between two consecutive grid points. Then the convergence properties stated in Proposition \ref{prop:convBS} remain true.
This procedure is not the same as placing nodes in the quadrature method because the grids remain the same for all dates  $\{t_{l}\}_{l=1,\dots,L}$
provided that the strike price does not change with time.
With a fixed grid, we can price  even more exotic options, e.g. a digital call with a down-and-out feature, provided that the barrier occurs at the same value of the underlying for each $t_{l}$.

The main difference between Equations \eqref{eq:recursiveBS} and \eqref{eq:recursiveBS2} is provided by the following explanation. In the right-hand side of  \eqref{eq:recursiveBS}, we find the exact values taken by the payoff function $H(y,T)$ on the grid $\{y_{j}\}_{j\in \N}$, and there is no approximation of the payoff. On the right-hand side of Equation \eqref{eq:recursiveBS2}, we find the values  $\{v_{j}(t_{l+1})\}_{j\in \N}$ obtained in the previous step of the algorithm, and these are approximations of the true values $\{V(y_{j}, t_{l+1})\}_{j\in \N}$.
Regardless of this fundamental difference, Proposition \ref{prop:convBS} states that the convergence rate is the same for both cases.
In the case of a European option with maturity $T$, by taking $t=t_{l}$ in \eqref{eq:recursiveBS}, we obtain the approximated price of the option at $t$. Equation  \eqref{eq:recursiveBS2} allows us to recursively compute the values of the option at different points in time, and thus to price Bermudan options, American options, and other types of path-dependent options.

In the implementation of Section \ref{sec:example_european}, we project the payoff function on a finite number of indicator functions. Thus, the values $\{H(y_{j}, T)\}_{j=1,\dots,N}$ and $\{\Delta y\, G(y_{i},t; y_{j}, t_{l+1})\}_{j=1,\dots,N}$ are then the entries of the finite dimensions matrices $\mathbf{H}(T)$ and $\mathbf{G}(t_{l};t_{l+1})$, respectively. Hence, we can express Equations \eqref{eq:recursiveBS} and \eqref{eq:recursiveBS2} in exactly the same matrix forms as in \eqref{eq:matrix} and \eqref{eq:rec_matr}.
By choosing $N$ sufficiently large, we can make the error introduced by truncating the infinite summations in  \eqref{eq:projection} and \eqref{eq:projection2} arbitrarily small. The truncation only suppresses the indicator functions localized in regions where the state price density vanishes, and where we can neglect the local contribution to the computed expectation\footnote{Truncating $N$ is equivalent to choosing an $y_{min}$ and $y_{max}$, and to restrict our analysis to $L^2([y_{min}, y_{max}])$.
This is the standard assumption in each quadrature based approach. This also insures that $H\in L^2([y_{min}, y_{max}])$, even though the proofs in Appendices A and B of the supplemental file do not need the measurability of $H$.}.
When the time step is a constant $\tau=t_{l+1} - t_{l}$, we obtain the fast and easily implementable algorithm of the introductory example.

Furthermore, as in Section \ref{sec:exampl_PathDep}, whenever an exercise date $t_{l}$ coincides with a dividend-paying date $t_{h}$,
we only need to replace the entries $\{\Delta y\, G(y_{i},t_{l}; y_{j}, t_{l+1})\}_{j=1,\dots,N}$ of  $\mathbf{G}( t_{l}; t_{l+1})$ with the values $\{\Delta y\, G(y_{i}-d, t_{l}; y_{j}, t_{l+1})\}_{j=1,\dots,N}$, whenever $t_{l}\in \{t_{h}\}_{h=1,\dots,H}$, to accommodate for a discrete dividend $d$.

\subsection{Stochastic volatility: the projection step}\label{sec:ProjectionHes}
In the class of  stochastic volatility models, there are two state variables, the underlying asset $S_{t}$ and the variance $\sigma^{2}_{t}$.
The bivariate state price density $G_{2}(S_{t}, \sigma^{2}_{t}, t; y, w, T)$ describes the discounted transition probability density from the asset level $S_{t}$ and variance level $\sigma^{2}_{t}$ at time $t$ to the asset level $y=S_{T}$ and variance level $w=\sigma^{2}_{T}$ at time $T$.
Its Fourier transform is denoted by $\hat{G}_{2}(S_{t}, \sigma^{2}_{t}, t; \lambda, \kappa, T)$, so that $G_{2}(S_{t}, \sigma^{2}_{t}, t; y, w, T) =\displaystyle \frac{1}{4\pi^{2}}{\iint} d\lambda d\kappa e^{-\iota( \lambda y +\kappa w)}\hat{G}_{2}(x, \xi, t; \lambda, \kappa, T)$, where $\iota$ is the imaginary unit.

For the values taken by the underlying asset, let the grid $\{y_{j}\}_{j\in \N}$ and the orthonormal set $\{e_{j}(y)\}_{j\in \Z}$ be defined as in the Black-Scholes case.  The Fourier transforms of $\{e_{j}(y)\}_{j\in \Z}$ are denoted  by $\{\hat{e}_j(\lambda)\}_{j\in \Z}$ , such that (see Section C of the supplementary online Appendix for the analytic form of $\hat{e}_j(y)$): $e_{j}(y)=\frac{1}{2\pi}\int d\lambda e^{-\iota \lambda y}\hat{e}_{j}(\lambda)$.

For the variance, we use the same equally spaced grid $\{w_{q}\}_{q\in \N}$ for the values taken by both the variables $\sigma^{2}_{t}$ and $w=\sigma^{2}_{T}$, such that $\Delta w= w_{q+1} - w_{q}$. Let $\{\varepsilon_{q}(w)\}_{q\in \N}$ be the normalized indicator functions centered on the grid $\{w_{q}\}_{q\in \N}$ and of width $\Delta w$, and $\{\hat \varepsilon_{q}(\kappa)\}_{q\in \N}$ be their Fourier transforms. Furthermore, let $\{\lambda_{r}\}_{r\in \Z}$ and $\{\kappa_{z}\}_{z\in \Z}$ be two regularly spaced grids of values taken by the transformed variables $\lambda$ and $\kappa$, with constant widths $\Delta \lambda$ and $\Delta \kappa$, respectively.

The approximation of the payoff function $H(y,T)$ is the same as in Equation \eqref{eq:projection}. Because we do not have the analytical form of $G_{2}(S_{t}, \sigma^{2}_{t}, t; y, w, T)$, we cannot directly sample the transition density, thus we must rely on the approximation given below.
The projection step for the bivariate state price density $G_{2}(S_{t}, \sigma^{2}_{t}, t; y, w, T)$ when the conditioning variables take the values $S_{t}=y_{i}$ and $\sigma^{2}_{t} =w_{p}$, is based on the following approximation:
\begin{align}\label{eq:projectionHesG}
	&\tilde G_{2}(y_{i}, w_{p}, t; y, w, T)\\ \notag
	&=\sqrt{\Delta y \Delta w} \mspace{-6mu}\sum_{j=-\infty,q=1}^{\infty} \Bigl( \frac{1}{4\pi^{2}}\mspace{-9mu}\sum_{r,z=-\infty}^{\infty}\mspace{-9mu} \hat{G}_{2}(y_{i},w_{p}, t; \lambda_{r}, \kappa_{z}, T)\hat{e}_{j}(-\lambda_{r})\hat{\varepsilon}_{q}(-\kappa_{z})\Delta \lambda \Delta \kappa\Bigr)e_{j}(y)\varepsilon_{q}(w)	\\ \notag
	& \defeq \sqrt{\Delta y \Delta w} \mspace{-6mu}\sum_{j=-\infty, q=1}^{\infty}\Gamma_2(y_{i}, w_{p},t; y_{j}, w_{q}, T)e_{j}(y)\varepsilon_{q}(w),
\end{align}
where the second equality in \eqref{eq:projectionHesG} defines the quantities $\{\Gamma_2(y_{i}, w_{p},t; y_{j}, w_{q}, T)\}_{j \in \Z, q \in \N}$.
To parallel the discussion following Equation \eqref{eq:projection2} in the Black-Scholes case, the quantities $\{\Gamma_2(y_{i}, w_{p},t; y_{j}, w_{q}, T)\}_{j \in \Z, q \in \N}$ play the same role as $\{G(y_{i}, t;y_{j}, T)\}_{j \in \Z}$.
We can motivate these approximations as follows. The orthogonal projection of the state price density $G_{2}(y_i,w_{p},t; y,w, T)$ on the two orthonormal sets
$\{e_{j}(y)\}_{j\in \N}$ and $\{\varepsilon_{q}(w)\}_{q\in\N}$ is given by the inner products $\iint \mspace{-3mu}dydw\mspace{1mu} G_{2}(y_i,w_{p},t;y,w,T)e_{j}(y)\varepsilon_{q}(w)$. Because we only know the closed form\footnote{For the closed form of $\hat{G}_{2}(y_i,w_{p}, t; \lambda, \kappa, T)$, see \cite{griebsch2013evaluation}.} of
$\hat{G}_{2}(y_i,w_{p}, t; \lambda, \kappa, T)$ and not $G_{2}(y_i,w_{p},t;y,w,T)$, we exploit the following key relationship:
\begin{align}\notag
	&\iint dydw\mspace{1mu} G_{2}(y_i,w_{p},t;y,w,T)e_{j}(y)\varepsilon_{q}(w) \\ \notag
	& = \iint dydw \frac{1}{4\pi^{2}}{\iint} d\lambda d\kappa e^{-\iota( \lambda y +\kappa w)}\hat{G}_{2}(y_{i}, w_{p}, t; \lambda, \kappa, T)e_{j}(y)\varepsilon_{q}(w) \\ \notag
	&= \frac{1}{4\pi^{2}}{\iint} d\lambda d\kappa\, \hat{G}_{2}(y_i,w_{p},t;\lambda, \kappa, T) \int dy e^{-\iota\lambda y}e_{j}(y)\int dw e^{-\iota\kappa w}\varepsilon_{q}(w) \\
	\label{eq:GammaTwo}
	&= \frac{1}{4\pi^{2}}{\iint} d\lambda d\kappa\, \hat{G}_{2}(y_i,w_{p},t;\lambda, \kappa, T) \hat{e}_{j}(-\lambda) \hat{\varepsilon}_{q}(-\kappa).
\end{align}

\noindent
Each $\Gamma_2(y_{i}, w_{p},t; y_{j}, w_{q}, T)$ is an approximation of the last integral appearing in Equation \eqref{eq:GammaTwo}, obtained by a direct sampling of the Fourier tranforms $\hat G_{2}(x_i,w_{p},t;\lambda, \kappa, T)$, $\hat{e}_{j}(-\lambda)$ and $\hat{\varepsilon}_{q}(-\kappa)$ on the bivariate grid $\{(y_{j},w_{q})\}_{j\in\Z,q\in \N}$, and on the univariate grids $\{\lambda_{r}\}_{r\in \Z}$ and $\{\kappa_{z}\}_{z\in \Z}$.

As in Section \ref{sec:ProjectionBS}, only the quantities $\{\Gamma_2(y_{i}, w_{p},t; y_{j}, w_{q}, T)\}_{j\in \N,q \in \Z}$ are the inputs for the pricing algorithm. The representation \eqref{eq:projectionHesG} is only used in Section B of the supplementary online Appendix to prove the convergence properties of the algorithm.

\subsection{Stochastic volatility: the recursive step}\label{sec:RecursHes}
In the stochastic volatility framework, the recursion for a Bermudan option consists of moving backwards in time as in Equation \eqref{eq:rec_exp} with:
\begin{equation}\label{eq:condexpStochVol}
	V(x,\xi, t_{l}) = \max\bigl\{H(x,t_{l}),\E\bigl[e^{-r(t_{l+1}-t_{l})}V(S_{t_{l+1}},\sigma^{2}_{t_{l+1}}, t_{l+1})\big| S_{t_{l}}=x, \sigma^{2}_{t_{l}}=\xi\bigr]\bigr\}.
\end{equation}
Thus, the recursive step in the Heston model is the sampling counterpart of \eqref{eq:condexpStochVol}.
 The following proposition gives a recursion formula that relates the approximated values of the option at different points of the bivariate grid $\{(y_{j}, w_{q})\}_{j\in Z, q \in \N}$ and at different points in time $t_{l+1}$ and $t_{l}$. It also states the rate of convergence of the algorithm. Define \mbox{$\Gamma_{1}(y_{i}, w_{p},t_{l}; y_{j}, t_{l+1})=\displaystyle\sum_{q=1}^{\infty}\Gamma_2(y_{i}, w_{p},t_{l}; y_{j}, w_{q}, t_{l+1})\sqrt{\Delta w}$}.
\begin{proposition}\label{prop:convHes}
Let $H(y,T)$ be such that $|H(y,T)-H(y',T)|<C|y-y'|$ for a positive constant $C$, and for $|y-y'|<\Delta y$. Let $v_{ip}(t_{l})$ be defined for a set of dates $\{t_{l}\}_{l=1,\dots,L}$, with $t_{L}=T$, as follows:
\begin{alignat}{4}\label{eq:recursiveHes}
	v_{ip} (t_{l}) &= \max\Bigl\{ H(y_{i},t_{l}), \sum_{j=1}^{\infty} \Gamma_{1}(y_{i}, w_{p},t_{l}; y_{j}, t_{l+1}) H(y_{j}, t_{l+1}) \sqrt{\Delta y} \Big\}, &	\ & \mathrm{for}	 & \ & l=L-1 ,	\\
	\label{eq:recursiveHes2}
	v_{ip} (t_{l}) &= \max\Big\{H(y_{i},t_{l}), \sum_{j,q=1}^{\infty}\mspace{-3mu}\Gamma_2(y_{i}, w_{p},t_{l}; y_{j}, w_{q}, t_{l+1})\mspace{1mu}v_{jq}(t_{l+1})\mspace{1mu}\sqrt{\Delta y \Delta w} \Big\}, & & \mathrm{for } & & l=1,\dots, L-2.
\end{alignat}

Then, for each $t_{l}$ in $\{t_{1},\dots, t_{L-1}\},$ the approximated values $v_{ip}(t_{l})$ defined in \eqref{eq:recursiveHes} and \eqref{eq:recursiveHes2} converge to the true value $V(y_{i}, w_{p}, t_{l})$ with an approximation error of the order $O\bigl(\underline\Delta^{2}\bigr)$, with $\underline \Delta = \sqrt{(\Delta y)^{2} + (\Delta w)^{2}}$.
\end{proposition}
\begin{proof}
See Section B of the supplementary online Appendix.
\end{proof}

For $t=t_{L-1}$,  Equation \eqref{eq:recursiveHes} gives the price of a European option in the Heston model.
Because the payoff function $H(y,T)$ only depends on the value $y$ taken by the underlying asset at $t_{L}=T$, the computed price $v_{ip}(t_{L-1})$ depends on the stochastic variance only through the conditioning value $\sigma^{2}_{t_{L-1}}=w_{p}$.
For this reason, we can use $\{ \Gamma_{1}(y_{i}, w_{p},t; y_{j}, T)\}_{j \in \N}$ instead of $\{\Gamma_2(y_{i}, w_{p},t; y_{j}, w_{q}, T)\}_{j, q \in \N}$.
Using  the values $\{ \Gamma_{1}(y_{i}, w_{p},t; y_{j}, T)\}_{j \in \N}$ in \eqref{eq:recursiveHes} is equivalent to applying the projection step on the Fourier transform $\hat{G}_{1}(y_{i}, w_{p}, t; \lambda,T) = \int \mspace{-3mu}d\kappa\mspace{1mu} \hat{G}_{2}(y_{i}, w_{p}, t; \lambda,\kappa, T)$. It is the univariate function $\hat{G}_{1}$ and not the bivariate  $\hat{G}_{2}$ that appears, for instance, in the original work by Heston (1993) for a European option.
Figure \ref{fig:projection-2d}  graphically presents the projection and recursion steps in the bivariate case.
 	
	\medskip
	{\centering	
	[Figure  \ref{fig:projection-2d} about here]

	}

In the implementation, we truncate the summations in \eqref{eq:recursiveHes} and \eqref{eq:recursiveHes2}, so that the grid $\{(y_{j},w_{q})\}_{j=1,\dots,N;q=1\dots, W}$ has $N\times W$ points.
The $N\times W$ matrix of computed prices at time $t=t_{l}$  is denoted by $\bm v_{2}(t_{l})$ , that is $\bm{v}_{2, jq}(t_{l})=v_{jq}(t_{l})$. Let $\mathbf{\Gamma}_{2}(y_{i}, w_{p}, t_{l}; t_{l+1})$ be the $N\times W$ matrix of the approximated transition probabilities from the initial point $(y_{i}, w_{p})$ to the end points of the entire grid $\{(y_{j},w_{q})\}_{j=1,\dots,N;q=1\dots, W}$. As in Section \ref{sec:example_european}, we integrate the normalization parameter $\sqrt{\Delta y \Delta w}$ in the definition of the transition matrix. We then have that $\mathbf{\Gamma}_{2, jq}(y_{i}, w_{p}, t_{l}; t_{l+1})=\Gamma_2(y_{i}, w_{p},t_{l}; y_{j}, w_{q}, t_{l+1})\sqrt{\Delta y \Delta w}$.
Let $\phi_{j}=\{\hat{e}_{j}(-\lambda_{r})\}_{r=1,\dots, R}$ and $\varphi_{q}=\{\hat{\varepsilon}_{q}(-\kappa_{z})\}_{z=1,\dots, Z}$ be the values of the functions $\hat e_{j}(-\lambda)$ and $\hat \varepsilon_{q}(-\kappa)$ sampled at the grids $\{\lambda_{r}\}_{r=1,\dots, R}$ and $\{\kappa_{z}\}_{z=1,\dots, Z}$, respectively.
Furthermore, we define the $R \times N$ matrix $\bm{\phi} = (\phi_{1}, \dots, \phi_{N})$, the
$Z \times W$ matrix $\bm{\varphi} = (\varphi_{1}, \dots, \varphi_{W})$, and the $R \times Z$ matrix $\mathbf{\hat G}_{2}(y_{i}, w_{p}, t_{l}; t_{l+1})$ with entries $\mathbf{\hat G}_{2, rz}(y_{i}, w_{p}, t_{l}; t_{l+1})=\hat G_{2}(y_{i}, w_{p}, t_{l}; \lambda_{r}, \kappa_{z}, t_{l+1})$.
Then, we can write the coefficients of the projection step \eqref{eq:projectionHesG}  in matrix form as:
\begin{equation}\label{eq:projectionHesGMatr}
	\mathbf{\Gamma}_{2}(y_{i}, w_{p}, t_{l}; t_{l+1}) = \bm\phi'  \mathbf{\hat G}_{2}(y_{i}, w_{p}, t_{l}; t_{l+1}) \bm\varphi\sqrt{\Delta y \Delta w}.
\end{equation}
The recursive step \eqref{eq:recursiveHes2} becomes the following:
\begin{align}\notag
	v_{ip}(t_{l}) = \max \Bigl\{H(y_i, t_l), & \sum_{j=1}^{N}\sum_{q=1}^{W}  \Gamma_2(y_{i}, w_{p},t_{l}; y_{j}, w_{q}, t_{l+1})v_{jq}(t_{l+1})\sqrt{\Delta y\Delta w}\,\Bigr\}\\
	\label{eq:recursiveHes2Matr}
	=\mspace{9mu} &\max \bigl\{H(y_i, t_l),  \mathbf{\Gamma}_{2}(y_{i}, w_{p}, t_{l}; t_{l+1}) \bm : \bm v_{2}(t_{l+1})\, \bigr\},
\end{align}
where the symbol ``$\bm :$'' denotes the Frobenius, or entry-wise, product.

In Section D of the supplementary online Appendix, we show how to speed up the computation of the matrices $\mathbf{\Gamma}_{2}(y_{i}, w_{p}, t_{l}; t_{l+1})$ by taking advantage of the space translation invariance property of transition densities.
Our method contains the Fast Fourier Transform (FFT) as a special case. In the FFT, the univariate grids for  $\lambda$ and $\kappa$ are automatically set, which can sometimes cause an imprecise reconstruction of the $\mathbf{\Gamma}_{2}(y_{i}, w_{p}, t_{l}; t_{l+1})$ matrices.

\subsection{Numerical illustrations in the Heston model}\label{sec:NumericalHeston}
We investigate the performance of our method in a standard affine model such as the \cite{heston1993closed} model.  We study an American option, written on an asset $S_t$, which pays discrete dividends and that evolves according to the following stochastic volatility model:
\begin{equation}\label{eq:stochvol}\begin{split}
        dX_{t} &= \Bigl(r -\frac{1}{2}\sigma_{t}^{2}\Bigr)dt + \sqrt{\sigma_{t}^{2}}\cdot dW_{1t},\\
        d\sigma_{t}^{2} &= \beta\bigl(\sigma_{LT}^2 - \sigma_{t}^{2}\bigr)dt + \omega \sqrt{\sigma_{t}^{2}}\cdot dW_{2t}, \quad E(dW_{1t}\cdot dW_{2t})=\rho dt.
\end{split}\end{equation}
In Equation \eqref{eq:stochvol}, $X_{t}=\log S_{t}$, and $\sigma_{t}^{2}$ is the variance process. We work with $X_t$ to be able to implement the space invariance property of the transition matrices, as outlined in Section D of the supplementary online Appendix. We conduct two simulation studies.
In the first, the American call has a time to maturity of one year, and $3$ dividends worth $d=2$ are distributed at $t_{h}=0.25,0.5,0.75$.
In the second, the time to maturity remains one year, but a single large dividend $d=10$ is paid out after six months.
The process parameter values are the following: $r=0.05$,  $\sigma_{LT}=0.2$, $\beta=2$ and $\omega=0.2$. Moreover we choose the parameter $\rho$ to be equal to zero. We compute the price for an at-the-money option ($S_{0}=K=100$).
The benchmark method in this analysis is a finite-difference (hereafter \emph{FD}) numerical solution of the partial derivatives equation (\textit{PDE}) that describes the evolution
of the price process $V_t$ of the American call. We implement an alternating direction implicit (\emph{ADI}) variant of the finite-difference scheme.
For a recent discussion of schemes similar to  \emph{FD}, see, for instance, \cite{in2010adi}. This implementation is equivalent to a Crank-Nicolson scheme, which in standard problems converges at a rate $O\bigl((\Delta t)^{2}\bigr)$, where $\Delta t$ is the temporal discretization interval.
In both the \emph{FD} scheme and the recursive projections, the evolution of the option price $V_t$ is studied on a rectangular grid in the space $(X,\sigma^{2})$, with $X \in [\log(K) - 10\sigma_{LT}\sqrt{T}, \log(K) + 10\sigma_{LT}\sqrt{T}]$ and $\sigma^{2} \in [0, 0.3]$. In the \emph{FD} scheme, the parameter $m_s$ gives the number of equally spaced grid points in the $X$ direction, and $m_v$ gives the number of equally spaced grid points in the $\sigma^{2}$ direction, so that the grid points are $\{(X_i, \sigma^2_p)\}_{i=1,\dots,m_s; p=1,\dots,m_v}$.  The parameter $L_{T}$ gives the number of time steps used.
In the recursive projections, under a sampling scheme we define $\Delta y=2^{-J}a$, where $a$ is a positive constant that gives the step of the $\{y_{j}\}_{j=1,\dots,N}$ grid when $J=0$. Describing the convergence of the recursive projections in terms of the parameter $J$ emphasizes how the approximation error decreases each time the number of grid points is doubled. Similarly, $\Delta w = 2^{-J_{w}}a_{w}$, where $a_{w}$ is the step with $J_{w}=0$ of the $\{w\}_{p=1,..,W}$ grid in which the $\sigma^{2}_{t}$ variable takes values.

Assuming the contemporaneous correlation $\rho=0$ simplifies the implementation of the \emph{FD} scheme, in the sense that neglecting the correlation between $X_{t}$ and $\sigma^{2}_{t}$ makes the \emph{FD} scheme easier to code and faster.
On the other hand, the speed and complexity of the recursive projection method are unaffected by the value chosen for the parameter $\rho$. The correlation is addressed in the Green function  $G_2(x, \sigma^{2}_{t}, t; y, w, T )$ and consequently in the coefficients of the matrix $\hat{\mathbf{G}}_{2}$. Because the speed of the method depends on the number of entries in the $\hat{\mathbf{G}}_{2}$ matrix, and not on the values taken by the entries, it is clear that the choice of $\rho$
does not affect the convergence rate of the recursive projections. This feature is the first advantage of the recursive projection over finite-difference schemes. This simulation study will then give a lower bound to the difference in speed between the recursive projections and the \emph{FD} scheme.
To price an American option on dividend-paying stocks, we should implement the \emph{FD} scheme-equivalent of the recombining tree. Doing so is practically unfeasible because it would mean computing at each ex-dividend date a new option price at each point of the grid.
Instead,  at each ex-dividend date $t_{h}$ and at each grid point $(X_{i}, \sigma^2_{p})$, we opt for comparing  the intrinsic value $H(X_{i}, t_h)$ with the continuation value $V(\widetilde{X}^{d}_{i}, \sigma^2_{p}, t_h)$, where $\widetilde{X}^{d}_{i}$ is the value of the $X$ grid closest to $\log(e^{X_{i}}-d)$.
This choice amounts to perturbating the \emph{FD} scheme at each ex-dividend rate, which could translate into a convergence slower than the theoretical $O\bigl((\Delta t)^{2}\bigr)$. This feature is a second advantage of the recursive projection over the finite-difference schemes, because, as we explained in Section \ref{sec:BSillustration}, the recursive projections can easily adapt to discrete dividends without their affecting the convergence properties of the algorithm.
The recursive projections achieve convergence quickly in the $\sigma^{2}$ direction. The method does not seem to improve by setting a resolution level greater than $J_{w}=4$; thus, we keep this value fixed throughout our simulations.
The \emph{FD} scheme is also not very sensitive to the number of points used in the $\sigma^{2}$ direction. We find no improvement beyond $m_{v}=31$.

Figure \ref{fig:Hes3divConv} shows the results for the 3-dividend case. The true value used to compute the pricing errors is $7.397$, obtained with the resolution level $J=13$. The graph on the right displays the pricing error of the \emph{FD} scheme as a function of the time discretization parameter $L_{T}$. Each line is relative to a different value of the spatial discretization parameter $m_{s}$. The time labels are all relative to the $m_{s}=3200$ curve.
The \emph{FD} scheme with $L_{T} =2048$ and $m_{s}=6400$ delivers a value within 1$bp$; thus, we assume that the methods have converged when the absolute value of the relative error is within 1$bp$ of $7.397$. The graph on the left plots the relative pricing error of the recursive projections against the resolution level $J$. The regression line on the left graph shows that the estimated slope is almost exactly the slope of $-2$ predicted by the theoretical convergence results of Proposition \ref{prop:convHes}.
 The \emph{FD} is at least one order of magnitude slower. Compare, for instance, the computation time needed to deliver a 4$bp$ error (2$s$ against 65$s$), or a 1$bp$ error (8$s$  against 130$s$). Figure \ref{fig:Hes1divConv} compares the convergence speed of the two methods in the 1-dividend case. The true value of 7.302 is obtained by the recursive projection method with $J=13$. The \emph{FD} scheme requires 48 seconds to reach a $5 bp$ relative error,  with parameters $m_{s}=400$ and $L_{T}=2048$. The bottom curve, relative to $m_{s}=200$, shows that the method does not converge for smaller values of the space discretization parameter.
The small $5 bp$ bias of the \emph{FD} is due to the large value of the dividend $d$ and the perturbation of the scheme at each dividend date. The rate as a function of the resolution level $J$ at which the recursive projections attain the $1 bp$ error band is approximately -2, as theoretically predicted.

The reason for the difference in speed between the recursive projections and the \emph{FD} scheme lies in the fundamentally different way  finite differences and quadrature methods deal with time stepping.
Both methods achieve time stepping through matrix multiplications. But while the number of time steps in the \emph{FD} is of the order of $2^{9}$ or higher, the recursive projections only need 3 or 4 time steps, one per divided payment, plus the expiry date.
The size of the parameters $L_{T}$, $m_{s}$ and $m_{v}$ determines the efficiency of the implementation of the \emph{FD} scheme.
If we compare the magnitude of the parameters $L_{T}$, $m_{s}$ and $m_{v}$ that we need to obtain convergence with the values of the equivalent parameters in \cite{in2010adi}, we find that our implementation is close to the most recent ones in the literature.
While specific implementations could marginally improve on ours, we think that we give a fair representation of the potential of the two techniques.
We remind that the computational time per time step is underestimated in our simulation, since the assumption of $\rho = 0$ reduces the number of intermediate steps in the $ADI$ implementation of the \emph{FD} scheme. Finally, if we include jumps in the process of the underlying stock, as we do in our empirical application, the numerical complexity of the recursive projections remains exactly the same as in the stochastic volatility case.
Introducing jumps in the underlying process while keeping the finite differences viable from a computational point of view asks for technical devices (see for instance \cite{d2005robust}) which are model specific and not yet implemented in conjunction with stochastic volatility\footnote{In Appendix H of the supplemental online Appendix we elaborate further on the implications of the choice $\rho=0$, on the relative advantages of recursive projections and $ADI$, and compare in detail our implementation with the one of \cite{d2005robust}.
}.

Another notable difference between the \emph{FD} and the recursive projection method is that the latter demands far fewer changes to adapt to different pricing problems.
In Equation \eqref{eq:matrix}, the matrix $\mathbf{G}(t;T)$ depends only on the dynamics of the underlying asset and not on the payoff. We can compute it once for all and use it to price different options with different payoffs, because the payoff functional form only impacts the vector $\mathbf{H}(T)$.  Such a design is particularly suited for object-oriented programming, which is often used in quant desks.
In finite-difference schemes, we cannot price options with different payoffs through the use of the same transition matrices, as boundary conditions affect the way the matrices are computed.

\section{Numerical applications and empirics}\label{sec:boundaries_empirics}
\subsection{Numerical comparison of early exercise boundaries}\label{sec:early_boundaries}

In this section, we compare the early exercise boundary implied by the Black-Scholes model with those implied by the Heston stochastic volatility and the Merton jump-diffusion models. We study two cases in which i) the stock distributes a continuous dividend yield and ii) the stock distributes a discrete dividend. Combining cases i) and ii) with the different modelling assumptions for the underlying asset and different maturities leads to very different patterns. For instance, in the discrete dividend case, the early exercise boundary is lower under the Black-Scholes model than under the Heston model, whereas in the continuous dividend case, the opposite is true. Hence, by modelling a discrete dividend as a continuous yield, we can draw misleading conclusions in an empirical evaluation of suboptimal non-exercise.
The exercise boundary $S^*_{t}$ for an American call with a continuous dividend yield is defined as the lowest value of $S_{t}$ such that $S_{t}-K \geq C(S_{t},T,K)$. If the value of the current stock is above $S^*_{t}$, then it is optimal for the call holder to exercise his option. With a discrete dividend, it can only be optimal to exercise the call option on the days immediately before the ex-dividend dates $t_{h}$. The exercise boundary $S^*_{t}$ for an American call option is then defined as the lowest value of $S_{t_{h}}$ such that $S_{t_{h}}-K \geq C(S_{t_{h}}-d,T,K)$.

	\medskip
	{\centering	
	[Figure  \ref{Heston_exercise_boundary} about here]
	
	}

In Panel A of Figure \ref{Heston_exercise_boundary}, we plot the early exercise boundary for the Heston and Black-Scholes models for an American call option with a continuous dividend yield $r_d=0.03$ (right graph) and with an equivalent quarterly discrete dividend $d=1.38$ (left graph). We choose $d=1.38$ to have an equivalent total annual dividend between the continuous dividend yield $r_d=0.03$ and the discrete dividend case. Indeed, $1.38=0.03 S^*/4$, where $S^*=184$ is the critical stock price under the Black-Scholes model in the dividend yield case for maturity $T=0.5$. We use the following set of representative parameters: $T=0.5$, $K=100$, $r=0.05$, $\sigma_0=0.2$, $\omega=0.1$, $\sigma_{LT}=0.3$, $\beta=4$, and $\rho=-0.5$ (\cite{adolfsson2013representation}).
For comparison, we follow \cite{heston1993closed}, and we use the Black-Scholes model with a volatility parameter that matches the (square root of the) variance of the spot return over the life of the option in the Heston model.
When the stock distributes a regular quarterly dividend, there are only two dividend payments during the life time of an option with maturity $T=0.5$, and it is immediately before the payment dates that it can be optimal to exercise the option.
In our example, the two dividend payments occur immediately, at $t=0$, and at $t=0.25$,  corresponding to a time to maturity of $0.5$ and $0.25$, respectively.
At both dates, the value of the exercise boundary is lower under the Black-Scholes model. With a continuous dividend yield, the Heston early exercise boundary is always below the Black-Scholes boundary, whereas with discrete dividends, the opposite is the case.

Although the findings in the continuous dividend case are in line with those of \cite{adolfsson2013representation}, the findings in the discrete dividend case are entirely new.
This difference warrants further intuitive discussion.
Assume there is only one discrete dividend to be paid. The continuation value of the call option immediately after the ex-dividend date is that of a European call with the remaining time to maturity. When $\rho \leq 0$, the price of European options for a deep in-the-money call, where early exercise could be optimal, is higher in the Heston case than in the Black-Scholes case (see \cite{heston1993closed,hull1987pricing}).
For instance, in the left graph of Panel B of Figure \ref{Heston_exercise_boundary}, for a time to maturity of $0.25$, this would be the case in the range of stock prices of approximately $150$. Even by taking into account the dividend drop in computing the continuation value, the ex-dividend stock price should remain in the region where the Heston price is higher.
We can repeat the same argument for a number of discrete dividends sufficiently small (typically of the order of a couple percent) to prevent the stock price from falling in the price range where the call has more value under the Black-Scholes model. The behavior of the boundary with a continuous dividend is less straightforward to grasp.
Following \cite{Kim:1990aa} and \cite{jamshidian1992analysis}, we can decompose the value $V(S_{t},t)$ of an American option into two components, namely, the European value $V^{E}(S_{t},t)$ and the early exercise premium $V^{A}(S_{t},t)$, such that:
\begin{align}\label{eq:EarlyExPremium}
	V(S_{t},t) & = V^{E}(S_{t},t) + V^{A}(S_{t},t)	\\ \notag
			& =e^{-r(T-t)} \mathbb{E}\bigl[ (S_{T}-K)_{+}\big | S_{t}, \sigma^{2}_{t} \bigr] +\int_{t}^{T}\mspace{-9mu}e^{-r(s-t)} \mathbb{E}\bigl[ (r_{d}S_{s} -rK)\mathbb{I}_{(S_{s}>S^{*}_{s})}\big | S_{t}, \sigma^{2}_{t} \bigr]ds,
\end{align}
where $S^{*}_{s}$ is the early exercise boundary at time $s$ and $\mathbb{I}_{(S_{s}>S^{*}_{s})}$ equals one if, at time $s$, the stock is in the exercise region, otherwise zero. We can interpret $V^{A}(S_{t},t)$ as a  continuum of European call options with maturity $T-s$, strike price $S^{*}_{s}$, and payoff $r_{d}S_{s} -rK$.
For each of these European options, we can apply the results of Table II in \cite{hull1987pricing} who compare the values of
European options under general stochastic volatility dynamics with the Black-Scholes price.
Call values under the stochastic volatility assumption are lower when the contracts are at-the-money and $\rho \leq 0$.
The continuum of contracts composing the $V^{A}(S_{t},t)$ are at-the-money when $S_{s}=S^{*}_{s}$. As confirmed from our numerical simulations, the $S^{*}_{s}$ values are distributed in the region immediately above $S=150$, that is, exactly where the price of the American option under the Heston model is lower than that under the Black-Scholes model, and explain the negative bump in the right graph of Panel B of Figure \ref{Heston_exercise_boundary}.

Similarly, we can characterize the early exercise under the Merton jump-diffusion model, where the asset $S_{t}$ evolves according to the following jump-diffusion process:
\begin{equation}
    \frac{dS_t}{S_t}=(r-r_d-\gamma \upsilon)dt + \sigma_M dW_t +(\psi-1) dq_t,\label{eq:Merton}
\end{equation}
 where $r_d$ is the continuous dividend yield paid by the asset, and $\sigma_M^2$ is the instantaneous variance of the return conditional on no jump arrivals. The Poisson process, $q(t)$, is independent of $W_{t}$, and such that there is a probability $\gamma dt$ that a jump occurs in $dt$, and $1-\gamma$ probability that no jump occurs.
The parameter $\gamma$ represents the mean number of jumps per unit of time. The random variable $\psi$ is such that $\psi-1$ describes the percentage change in the stock price if the Poisson event occurs, and $\upsilon=\mathbb{E}[\psi-1]$ is the mean jump size.
We further make the standard assumption (for instance, see \cite{Amin:1993a,bakshi1997empirical}) that $\log(\psi) \sim N(\mu_{\psi},\sigma_{\psi}^2)$. If $\gamma=0$, then we recover the standard Black-Scholes model with no jumps.
We use the following set of representative parameters: $K=40$, $r=0.08$, $\gamma=5$, $\sigma_M^2=0.05$, $\sigma_{\psi}^2=0.05$, $\mu_{\psi}=0$ (\cite{Amin:1993a}).
We set the volatility parameter in the Black-Scholes model  equal to the volatility of the underlying return over the life of the option in the Merton model.

	\medskip
	{\centering	
	[Figure  \ref{Merton_exercise_boundary} about here]
	
	}

In the Panel A of Figure \ref{Merton_exercise_boundary}, we plot the early exercise boundary for the Merton and Black-Scholes models for an American call option with a continuous dividend yield $r_d=0.05$ (right graph) and in the case in which the stock pays an equivalent quarterly discrete dividend \mbox{$d=1.125 $}~\footnote{As before, we take $d=1.125$ because $1.125=0.05 S^*/4$, where $S^*=90$ is the critical stock price with the dividend yield $r_d=0.05$ for maturity $T=0.5$.} (left graph). As for the Heston case, the results on the continuous dividend case are in line with the existing literature, e.~g., \cite{Amin:1993a}, and we provide new insights into the discrete dividend case.
To interpret the graphs in Figure \ref{Merton_exercise_boundary}, we have to make an important distinction. For short maturity options, the jump component in Equation \eqref{eq:Merton} dominates the diffusion component. As explained in \cite{Amin:1993a} and \cite{merton1976option}, the result is higher prices for short maturities under the Merton model than under the Black-Scholes model. We call this effect the jump effect. For longer maturities, the jump effect no longer dominates the diffusion component but instead creates an interplay that makes the jump-diffusion process observationally similar to a stochastic volatility process. For a discrete dividend, both the jump effect and the stochastic volatility effect, as previously discussed in the Heston case, predict a higher boundary in the Merton case than in the Black-Scholes case, which holds true for all maturities. This result is exactly what we find in the left graph of Panel A of Figure \ref{Merton_exercise_boundary}.
For a continuous dividend case, we have that the jump effect dominates for short maturities, giving a higher boundary under the Merton model than under the Black-Scholes model. For longer maturities, the jump effect diminishes, and the boundary behaves as in the stochastic volatility model, that is, taking lower values than in the Black-Scholes case. These insights explain the crossing of the early exercise boundary that we observe in the right graph of Panel A, Figure \ref{Merton_exercise_boundary}.

A key numerical finding of this section is that the early exercise is more likely under the Black-Scholes model when we face discrete dividends. In the next section, we assess the empirical consequences of this finding for the cost of suboptimal non-exercise.

\subsection{Empirical analysis}\label{sec:empirics}
In this section, we apply the recursive projection method to characterize the early exercise boundary of a large sample of call options. The options have a maturity of less than six months and are written on dividend-paying stocks, which are part of the Dow Jones Industrial Average Index (DJIA). The sample comprises daily observations between January 1996 and December 2012. We investigate the early exercise decision of call holders in light of the different values
that the exercise boundary can take under distinct modelling assumptions for the underlying asset.  Following the procedure suggested by \cite{pool2008failure}, we first check which contracts should be exercised by comparing the intrinsic value immediately before the dividend payment with the continuation value on the ex-dividend day. We quantify how much is economically lost in the case of a suboptimal non-exercise decision.
This amount depends on the continuation value and is model-specific. We compare the results obtained under three modelling environments, namely,
Black-Scholes, Merton jump-diffusion, and Merton jump-diffusion with the stochastic volatility dynamics of the Heston model.
\cite{Bates1996} was the first to suggest combining the Merton and Heston models, and therefore, we refer to this process specification as the Bates model.
Finally, whenever we find evidence of a suboptimal non-exercise decision, we show that trading costs alone cannot justify the behavior of investors.

In our empirical analysis, we price by fully taking into account the discrete nature of the dividend distributed by the underlying stocks and the American style of the call options, and we do so for the three pricing models. This feature is a peculiarity of our work, given that the standard empirical literature on options mainly focuses on European $S\&P 500$ options with a dividend yield or limits itself to the Black-Scholes model. We have already shown in an example in Section \ref{sec:BSillustration} how neglecting the discrete cash flow feature of the dividend payment, and approximating it with a continuous dividend yield, leads to a pricing error of $169bp$. In our empirical analysis, we also need to correctly take into account  the time of payment. A popular approach, when dealing with American options on dividend-paying stocks, is the so called ``escrow dividend'' model. Under this approximation, the option is priced as if it were a European contract, valued at the prevailing stock price minus the present value of all the dividends to be distributed during its remaining life. This technique correctly models that a long holder of a call option is unprotected from dividend distributions, but does not properly integrate the early exercise premium of an American option, and tends to underprice the option. As a consequence, an investor following this approach, could underestimate the value of the early exercise boundary, and exercise a contract that he would be better off by keeping alive. For example, on  May 10th, 2006, the Dupont stock closes at $45.71$ dollars. The call option with $K=30$ and $T=0.45$ should not be exercised if the continuation value is computed correctly, but the option should be exercised if the continuation value is approximated with a European price. In this case, if an investor exercises his option wrongly, he suffers a loss of $200 bp$ if the spot price is modelled under the Merton dynamics\footnote{We elaborate further on this example in Section G of the supplementary online Appendix where we show that a correct modelling of the dividend as a discrete cash flow does matter when measuring the cost of sub-optimality in the empirics of the next section.}.


 Our choice for alternative modelling environments follows the empirical findings of \cite{bakshi1997empirical}, who suggest that jumps and stochastic volatility play a dominant role in pricing short-term options whereas modelling stochastic interest rates does not seem to significantly improve the pricing performance.
In addition, the choice of the jump arrival distribution in Equation \eqref{eq:Merton} is motivated by the work of \cite{bajgrowicz2016jumps}, who show that high-frequency data on individual stocks support the hypothesis that jump arrivals follow a simple low-intensity Poisson process.
Their findings also support the assumption in Merton (1976) that the jump component is nonsystematic, i.e., diversifiable, because of the absence of co-jumps affecting all stocks. This explains our choice of the Bates model, where $X_{t}=\log(S_{t})$, with:
\begin{align}\label{eq:MRT+HST}
    & dX_t =(r-r_d-\gamma \upsilon-\frac{1}{2}\sigma^2_{t})dt + \sigma_{t} dW_{1,t} +log(\psi) dq_t,\\ \nonumber
    & d\sigma_{t}^{2} =\beta \bigl(\sigma_{LT}^{2} - \sigma_{ t}^{2}\bigr)dt + \omega \sqrt{\sigma_{t}^{2}}\cdot dW_{2,t}, \quad E(dW_{1,t}\cdot dW_{2,t})=\rho dt,
\end{align}
and $\log(\psi) \sim N(\mu_{\psi},\sigma^2_{\psi})$.
To implement the recursive projections in the Bates model, we need the Green function implied by Equation \eqref{eq:MRT+HST}. Because the jump process is independent of the Brownian motions $ dW_{1,t}$ and $ dW_{2,t}$, we use the property that the characteristic function of a bivariate process is the product of the characteristic functions of the independent univariate processes. The transition matrices are then obtained as in Section \ref{sec:RecursHes}.
The Bates model shows how straightforward it is to adapt the recursive projections to more complex models.
The number of entries in the transition matrices for the Bates model is exactly the same as it is as for the Heston model.
In the calibrations, we have compared the computation of the transition matrices with both the FFT and the full sampling method outlined in Section \ref{sec:ProjectionHes}. Because we have found that the FFT correctly reproduces the $\mathbf{\Gamma}_{2}(y_{i}, w_{p}, t_{l}; t_{l+1})$ matrices, we have opted for using the FFT throughout the empirical exercise to take advantage of the further increase in speed provided by the FFT algorithm.

The daily data on all option attributes, the stock price, and the dividend distribution details are from Optionmetrics. We obtain the daily data on the interest rates from the Treasury constant maturities of the H15 report of the Federal Reserve\footnote{We access the Optionmetrics and the H15 databases through the Wharton Research Data Services (WRDS) research platform.}.
A total of 101,295 series of short-term options written on 30 stocks enter our database. The total number of records is approximately 9.5 million.
This number stresses the importance of a fast and versatile numerical method.
Table \ref{Table:Db_description} reports the number of quotes for each stock with a breakdown by maturity and moneyness. Our study focuses on the early exercise behavior of investors; hence, we focus on the in-the-money options, which are the category of options for which the number of quotes is the highest.

	\medskip
	{\centering	
	[Tables  \ref{Table:Db_description} and \ref{Table:Param} about here]
	
	}

\subsubsection{Estimating the cost of suboptimal non-exercise}\label{sec:cost}

Table \ref{Table:Param} reports the results of the calibration for the three modelling frameworks. We obtain the parameters through the minimization of the implied volatility mean squared error, as in \cite{christoffersen2004importance}\footnote{ Section F of the supplementary online Appendix gives a detailed description of the data and the calibration procedure.}.
The first line of Table  \ref{Table:Param} displays the average values of the parameters calibrated on our sample of single stocks, whereas the second line reports the average values that \cite{bakshi1997empirical} obtain for the same parameters from contracts written on the  $S\&P 500$ index.
The parameters that rule the level of the volatility smile, namely, the Black-Scholes volatility $\sigma_{BS}$, the long term volatility $\sigma_{LT}$, and the spot volatility $\sigma_0$, are much higher in our single stock calibration than in the index calibration, which reflects the well known fact that an index is less volatile than its components.
Indeed, in our sample, the average Black-Scholes volatility is $29\%$, $\sigma_0$ is $28\%$, and the average long-term implied volatility is $32\%$, whereas for the index options, the same parameters take the values of $18.15\%$, $20\%$, and again $20\%$.
The jump parameters in the Bates model show that jumps are on average less frequent in the single stock case than in the index case ($\gamma_{Stocks}=0.5$ against $\gamma_{SP500}=0.61$), but the amplitude and variability are higher for single stocks ($\mu_{\psi,Stocks}=-0.12$ and $\sigma_{\psi,Stocks}=0.18$ respectively) than for the index  ($\mu_{\psi,SP500}=-0.09$ and $\sigma_{\psi,SP500}=0.14$). Given that the index is a diversified portfolio, it displays a jump whenever one of its constituents jumps, but the non-systematic nature of the jumps of single stocks (see \cite{bajgrowicz2016jumps})  attenuates the variability and magnitude at an aggregate level.
The remaining two parameters of the stochastic volatility component of the Bates model,  the correlation parameter $\rho$ and the volatility of volatility $\omega$, have a specific impact on the shape of the implied volatility smile (\cite{hagan2002managing,west2005calibration}). A negative $\rho$ implies a negatively sloped smile. The correlation parameter is in absolute value lower in the single stock case ($\rho_{SP500}=-0.52$ against $\rho_{Stocks}=-0.35$), meaning that the
implied volatility smile for the index is more negatively sloped than for individual stocks.
This finding is consistent with the findings of  \cite{bakshi2003stock} who describe the same relationship between the slopes of the index and of the individual stocks implied volatility smiles. \cite{bollen2004does} also find the same pattern, and explain it by relating the slope of the index smile to the buying pressure for index puts, with the demand for call options driving the shape of the smile of single stocks. The volatility of volatility $\omega$ determines the convexity of the implied volatility smile. The difference in the values taken by $\omega$ is striking. We find $\omega$ to be $75\%$ for stock options, whereas \cite{bakshi1997empirical} find a much smaller value of $40\%$ for short-term index options. This difference is due to the higher convexity of the implied volatility smiles of stock options versus that of index options, another feature also documented in \cite{bollen2004does}. The parameters $\rho$ and $\omega$  are related to the smile shape through the higher moments of the distributions of the returns of the underlying. A more negative $\rho$ generates a more negatively skewed distribution of index returns with respect to stock returns, whereas a higher $\omega $ in the single stock returns leads to a higher kurtosis than in the index return distribution.

In Table I of the supplementary online Appendix, we provide more detailed results, including a breakdown of the calibration by stock, and we show that the values of the calibrated parameters are homogeneous across stocks.

	\medskip
	{\centering	
	[Table  \ref{Table:Param} about here]
	
	}

After having calibrated the models, we are able to compute the price $C(S_{t_{h}-1}-d, K, T)$ on the day previous to the dividend payment date $t_{h}$, by using the recursive projections. The price $C(S_{t_{h}-1}-d, K, T)$ is the continuation value of the option at date $t_{h}$, when the dividend $d$ is distributed. By comparing it with the intrinsic value $S_{t_{h}-1}-K$, we can assess which options should be exercised on $t_{h}-1$. If an option should be exercised (i.e., $C(S_{t_{h}-1}-d,K,T) \leq S_{t_{h}-1}-K$), then a positive open interest at the end of the day before ex-dividend ($OI_{t_{h}-1}>0$) measures the failure of investors to exercise the option. In this case, we calculate the suboptimal non-exercise percentage as the following ratio:
\begin{equation}\label{eq:percentage_unexercise}
    NE_{\%}=\frac{OI_{t_{h}-1}}{OI_{t_{h}-2}},
\end{equation}
i.e., the number of contracts outstanding at the end of the day $t_{h}-1$ to the total number of contracts outstanding at the end of day $t_{h}-2$.
The quantity defined in Equation \eqref{eq:percentage_unexercise} is an approximation of the actual non-exercise ratio, because it neglects a possible issue of new contracts on date $t_{h}-1$. This event is unlikely; indeed \cite{pool2008failure} test the approximation on a subsample of contracts for which they have the real exercise data. They conclude that the approximation (\ref{eq:percentage_unexercise}) is a precise description of the actual exercise behavior of option investors. The total amount of money that is left on the table due to suboptimal
non-exercise is given by the following formula:
\begin{equation}\label{eq:TotalLoss}
TL=100 \times OI_{t-1} \times [(S_{t-1}-K)-C(S_{t-1}-d,K,T)].
\end{equation}
The continuation value $C(S_{t-1}-d,K,T)$ depends on the model used for pricing; hence, the total loss due to suboptimal non-exercise (TL) is itself model-specific.

Table \ref{table:NII_lost} presents the results on the suboptimal non-exercise behavior of investors.

	\medskip
	{\centering	
	[Table  \ref{table:NII_lost} about here]
	
	}

\noindent Table \ref{table:NII_lost} clearly shows that the optimal early exercise decision depends on the model used for the stock price. Under the Black-Scholes model, approximately $9.5\%$ of the outstanding contracts should be exercised, and the percentage decreases (approximately $7.5\%$) under the alternative models.
This result is consistent with the numerical findings of Section \ref{sec:early_boundaries},
where we show that, in the case of discrete dividends, the early exercise boundary under the Black-Scholes model is lower compared to that implied by the Merton and Heston models. As a general rule, contracts that should be exercised under the Merton or Bates models should also be exercised under the Black-Scholes model. In our sample, we find some exceptions to this rule because,
in Section \ref{sec:early_boundaries},
we choose the model parameters such that the total variance of the returns over the life of the option is the same in all models, whereas in real data, this condition may not hold. To give some examples, $4680$ options should be exercised under Black-Scholes but not under the Bates model, whereas the reverse is true only for $249$ contracts.
Similarly, we find that $2872$ options should be exercised under the Black-Scholes model but not under the Merton model, whereas the opposite occurs with only $53$ options. The first important lesson we learn is that, by allowing for more sophisticated models than the Black-Scholes model, the number of contracts that should be optimally exercised decreases by almost $25\%$. The suboptimality figures are model-dependent and may be a consequence of the calibration procedure.
The comparison between our calibration results and the results of \cite{bakshi1997empirical} are reassuring in terms of the reliability of our calibration method. To justify the suboptimal behavior found in our sample, we should obtain unreasonably high values for the jump and intensity parameters.

A second piece of evidence that stands out from Table \ref{table:NII_lost} is that the percentage of investors who leave the options suboptimally non-exercised is higher under the Black-Scholes model than under the other models, $39\%$ versus approximately $30\%$. We compute these percentages in accordance with Definition (\ref{eq:percentage_unexercise}).
If we restrict our attention to the $1965$ contracts in our sample that should be exercised under the Black-Scholes model but not under the Merton model or the Bates model, we find a striking $81\%$ of no-exercise. These results may suggest that investors do not limit themselves to a Black-Scholes world when evaluating their options but rely on more sophisticated models that include jumps or stochastic volatility.
Even if this evidence is a considerable step towards understanding the investor decision-making process, it does not fully solve the puzzle. Indeed, even in the Merton and Bates models, we still find a high percentage of suboptimal non-exercises, which leads to a global loss of approximately $130-140$ million dollars, down approximately
$30\%$ from the loss of $206$ million dollars in the Black-Scholes model.

A second possible explanation of the early exercise puzzle is that investors wait until the options are deeper in-the-money before exercising. If we restrict ourselves to the options that should be exercised only under the Black-Scholes model but not under the Merton or Bates models, then the average moneyness of the contracts, measured by the delta of the options, is $0.962$. If we consider all the contracts that should be exercised under the Black-Scholes model, including those that should also be exercised under the other two model specifications, the average moneyness rises to $0.986$.
Figure \ref{Figure:moneyness} shows how relevant the moneyness of the contract is in the non-exercise decision of investors: the more in-the-money the option is, the smaller the number of suboptimal non-exercises.
Investors may not respond immediately to favorable stock price movements and may take some time before reacting and optimally exercising their option, which would be in line with the behavior on the early prepayment of mortgages documented by \cite{stanton1995rational}.

	{\centering	
	[Figure \ref{Figure:moneyness} about here]
	
	}

Our summary results regarding the exercise decision in the Black-Scholes model are in line with those obtained by \cite{pool2008failure}.
In that work, the authors apply the early exercise decision rule to all options series by using the Black-Scholes model with historical volatility and find that $53.1\%$ of investors leave their options unexercised when instead they should have been exercised. Their data span over ten years (from 1996 to 2006) and to compare our results with theirs, we divide our sample into two subsamples, the first spanning the years 1996-2006 and the second spanning the years 2006-2012. Then, we calculate the average percentage of suboptimal non-exercise in the two subsamples and find that the percentage of suboptimal non-exercise under the Black-Scholes model is approximately $47\%$ in the first subsample and $37\%$ in the second. The decrease in the non-exercise behavior with time intimates that investors become more attentive in monitoring their investments. There is a small difference between our results ($47\%$) and the $53.1\%$ found in
\cite{pool2008failure}.  The explanation is most likely our focus on the constituents of the Dow Jones Industrial Average, whereas \cite{pool2008failure} consider all option series. It is likely that, for large-cap companies, stock and option prices are monitored more closely than they are on average.

Throughout our empirical investigation, we choose a model-based approach to calculate the continuation value of the option $C(S_{t_{h}-1}-d,K,T)$. We could have also used a market-based approach where the continuation value is the market price of the option. The market-based approach checks whether the quantity $C_{MKT}(S_{t_{h}-1}, K, T)-(S_{t_{h}-1}-K)_+$ equals $0$, where $C_{MKT}(S_{t_{h}-1}, K, T)$ is the observed market price at $t=t_h-1$. As discussed in \cite{pool2008failure} and in \cite{barraclough2012early}, the market-based approach has shortcomings. The most important is that it does not make it possible to calculate the total loss due to suboptimal non-exercise, which we do in Equation (\ref{eq:TotalLoss}). In addition, the bid-ask spread and the discreteness of the prices make it difficult to decide which $C_{MKT}$ should be used. For all of these reasons, we follow \cite{pool2008failure} and \cite{barraclough2012early} and use a model-based approach with Equation \eqref{eq:percentage_unexercise} to account for the actual exercise behavior of options investors. \cite{barraclough2012early} only use the market-based approach as a useful model-free test to confirm the presence of suboptimal non-exercise behavior. They find that the market-based approach gives a magnitude of suboptimality that is comparable to that implied by the model-based approach. This last piece of evidence is an additional argument against the possible objection that an incorrect calibration of the model parameters may be the source of the suboptimal exercise figures.

\subsubsection{The role of fees}

According to the recent literature on option prices (\cite{jensen2016early, christoffersen2014illiquidity}), trading costs and financial frictions in general strongly affect both the option prices and the early exercise decision of American options. In this section, we investigate whether the suboptimal non-exercise behavior of investors is due to the trading costs that investors face when exercising their options.

Following \cite{pool2008failure}, we model the costs of exiting a long call position as a per share lump sum $\mathcal{F}$ that the investor must pay at the moment he decides to exercise. The specific value of $\mathcal{F}$ depends on how the exit is accomplished according to the different possible objectives of the investor. The most expensive value of the fee $\mathcal{F}$ is attained when the investor wants to exercise the option and reenter into the same call position. \cite{pool2008failure} estimate an average value for the rollover costs $\mathcal{F}$ by using the commissions of the high-cost brokers, and they obtain a very conservative amount of $\mathcal{F}=0.4446$ dollar per share. A detailed description of the components of the fee $\mathcal{F}$ can be found in \cite{pool2008failure}.

To understand the role of the fees in the early exercise decision, we perform two different empirical exercises.
As a first check, we re-perform the exercise of Section \ref{sec:cost} and compute the loss due to a suboptimal non-exercise decision, but this time using $C(S_{t_{h}-1}-d,K+\mathcal{F},T) $ as the continuation value, and $(S_{t_{h}-1}-K-\mathcal{F}) $ as the intrinsic value. The fee value  is  $\mathcal{F}=0.4446$.
The fee $\mathcal{F}$ enters both in the exercise proceeds and in the continuation value. Indeed, at the moment of the exercise decision, the investor should decide whether to exercise and hence pay the exercise fee immediately or not exercise and postpone the payment of the exercise fee to a future date.
Accordingly, the calculation of the total amount of money that is lost due to suboptimal non-exercise is given by the following formula:
\begin{equation}\label{eq:TotalLossFee}
TL_{\mathcal{F}}=100 \times OI_{t-1} \times [(S_{t-1}-K-\mathcal{F})-C(S_{t-1}-d,K+\mathcal{F},T)].
\end{equation}
The second column in Table \ref{table:NII_lost} shows the summary results including the fee. They are not very different from those obtained without considering the fee (first column of Table \ref{table:NII_lost}). We can conclude that the inclusion of trading costs does not change the big picture on the suboptimal non-exercise of investors, as outlined in the previous paragraph.

As a second empirical exercise, we calculate the value of the fee that would justify the non-exercise decision of investors to detect possible additional costs that are not taken into consideration in the fee  $\mathcal{F}$. To do so, for each option for which $C(S_{t_{h}-1}-d,K,T)<(S_{t_{h}-1}-K)$, but that is not optimally exercised by some of the investors, we compute the value of the implied fee $\mathcal{F}$ that would justify the non-exercise decision. It amounts to numerically finding the zero of the following function:
\begin{equation}
    f(\mathcal{F})=C(S_{{t_{h}-1}}-d,K+\mathcal{F},T)-(S_{t_{h}-1}-K-\mathcal{F}).
\end{equation}
The results are reported in Table \ref{table:Fee}. The average implied fee is between $7$ and $8$ dollars per share, an incredibly high amount compared to the already conservative fee of $0.4446$ dollar per share estimated by \cite{pool2008failure}. No realistic hidden fees can sum up to 7 dollars per share, and the trading costs of exiting a long call option position cannot fully justify the suboptimal non-exercise behavior of investors.
We can interpret the difference between the implied fee of 8 dollars and the conservative fee of 0.4446 as an implied opportunity cost for the holder of the option to monitor the optimal exercise of the American option. The holder of the option chooses to spend this amount on alternative activities.

	{\centering	
	[Figure \ref{table:Fee} about here]
	
	}

\section{Concluding remarks}\label{sec:conclusion}

We investigate the exercise behaviour of investors in a large database of 101,295 series of short-term American call options. \cite{pool2008failure} find that more than 40\% of the investors fail to optimally exercise their contracts.
We extend their analysis by including stochastic volatility and jumps to the process of the underlying stock.  In order to deal with the large option database and the repeated calculations required for the calibration and pricing, we develop an option pricing technique which is at the same time fast, precise and which can handle both multidimensional dynamics and cash dividend distributions. We start from the observation that, by monitoring the value of an option at discrete times and by sampling the value function of the contract on a finite grid of values of the underlying asset, we can describe the evolution of the price process in terms of elementary matrix operators.
The interpretation of the elements of such a matrix in terms of a functional projection allows us to extend the matrix approach to the pricing of contracts written on assets following processes whose transition probabilities do not have a known analytical expression in the direct space. The recursive projection method owes its speed to the simplicity of its algorithm, which is based on sampling. Moreover, our approach allows us to derive transition matrices that relate the option prices at points in time that are arbitrarily distant.
The number of time steps required is solely driven by features of the contracts, such as dates at which the contract needs to be monitored, and this is the main feature that makes our method faster than existing alternatives.

By applying our technique to the dataset, we can explain up to 25\% of the gain forgone due to suboptimal exercise decisions, as computed in \cite{pool2008failure}.
This result confirms the insights we obtain from the theoretical part of the paper. Indeed, we show that the exercise boundary is higher under the Merton and Heston models than under the Black-Scholes model if the dividend is discrete.
This result underlines the importance of the correct modelling of the dividend distribution.
We show that by modelling the dividend as a continuous yield instead of a discrete cash flow, the exercise frontier in the Bates model would have been lower instead of higher, and the suboptimal exercise behavior of  \cite{pool2008failure} would have been reinforced instead of mitigated.
We further try to check whether we can explain the remaining 75\% of suboptimal behavior in terms of transaction costs (\cite{jensen2016early}). We show that hidden transaction costs would need to be unrealistically large to explain the entire amount foregone by investors.
This observation leads us to interpret the implied transaction fee as a monitoring cost.

We can envisage two lines of research for further investigation: other possible applications of our method and improving the algorithm.
American option pricing is only a particular case of stochastic optimal control problems. We can think of applying the recursive projection method to other problems, such as the optimal portfolio allocation involving complex and path-dependent financial assets.
Currently, these types of complex problems are solved by using Monte Carlo simulations  (\cite{Detemple2003monte}) and our method could offer a more efficient computational alternative. Similarly our method makes possible other types of analysis with options based on individual stocks paying discrete dividends and sophisticated pricing models such as, e.g., the recent studies of tail risk (\cite{andersen2015risk, Andersen2016aa}), political uncertainty (\cite{Kelly2016a}), sector-wide government guarantees (\cite{kelly2016price}), and factor structure (\cite{christoffersen2015factor}). We can also use our method to extend the optimal control problem of \cite{hodder2007} to more complex dynamics for the decision variable.
The second possible extension of our work is the use of different function bases in the projection step. Candidate basis systems are functions that share the localization property with the indicator function but that may display higher regularity and enhance the convergence speed.
Faster convergence should not come at the cost of a more complicated numerical implementation, and therefore, research should go in the direction of functional projections that result in linear transformations of the sampled values (\cite{Sweldens1996lifting, Sweldens1998lifting}).


\bibliography{FRPpaper}
\bibliographystyle{jfe}





\begin{sidewaysfigure}
	\begin{center}
	\begin{tikzpicture}[scale=0.9]
		\input{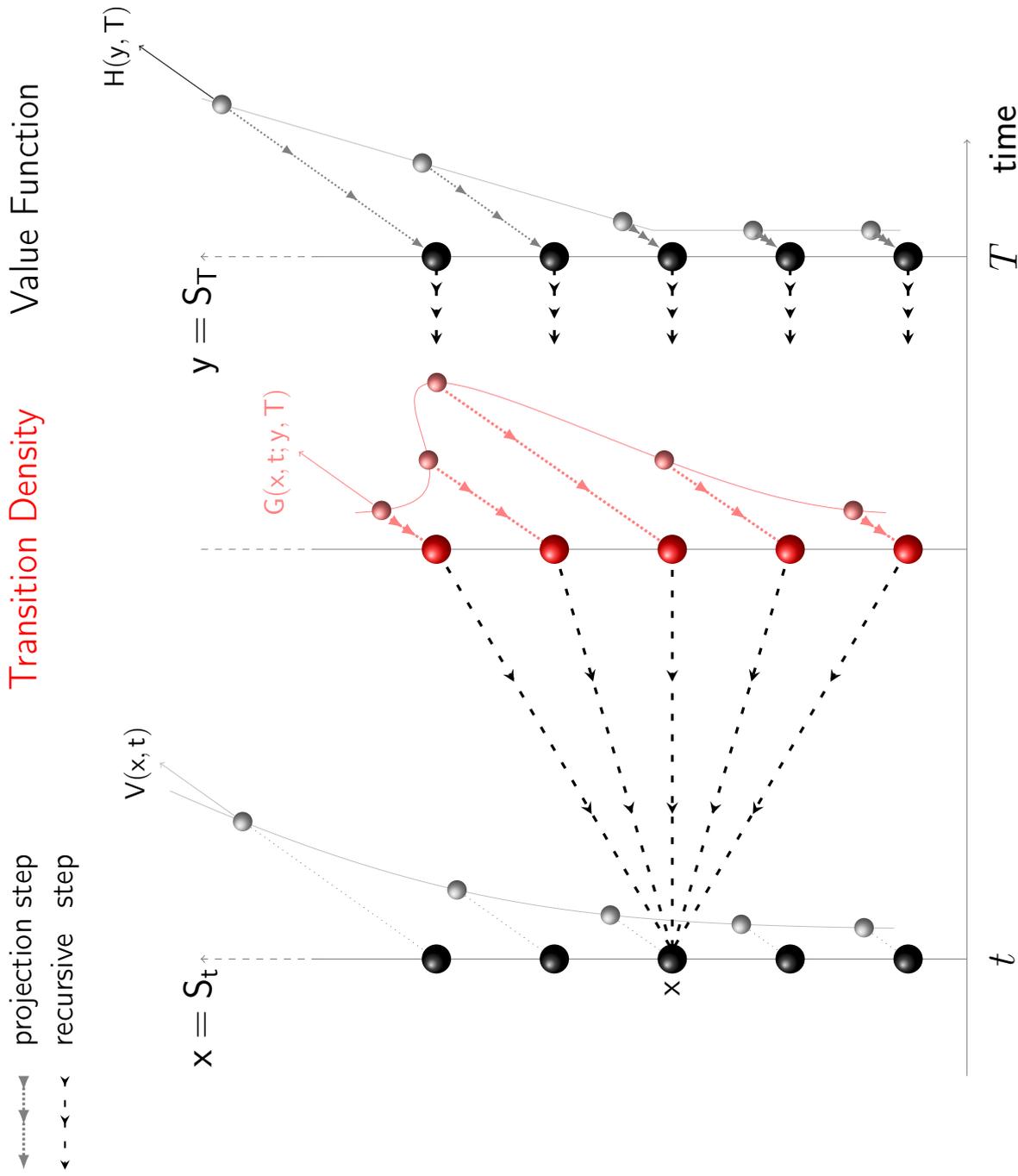}
	\end{tikzpicture}
	    \caption{\label{fig:projection-1d}Recursive Projection scheme in the Black-Scholes case is composed of two steps. First the projection step: the value function $H(y, T)$ at $T$ and state price density functions $G(x,t; y, T)$ are  sampled (\emph{dotted lines}). Second the recursive step: the sampled values are multiplied by the transition weights to obtain the value function at $t$  (\emph{dashed lines}), which in turn will be the input for the following step of the algorithm.}
	\end{center}
\end{sidewaysfigure}



\begin{sidewaysfigure}
  \begin{center}
  \begin{tikzpicture}[xscale=2.5, yscale=3] 
    \input{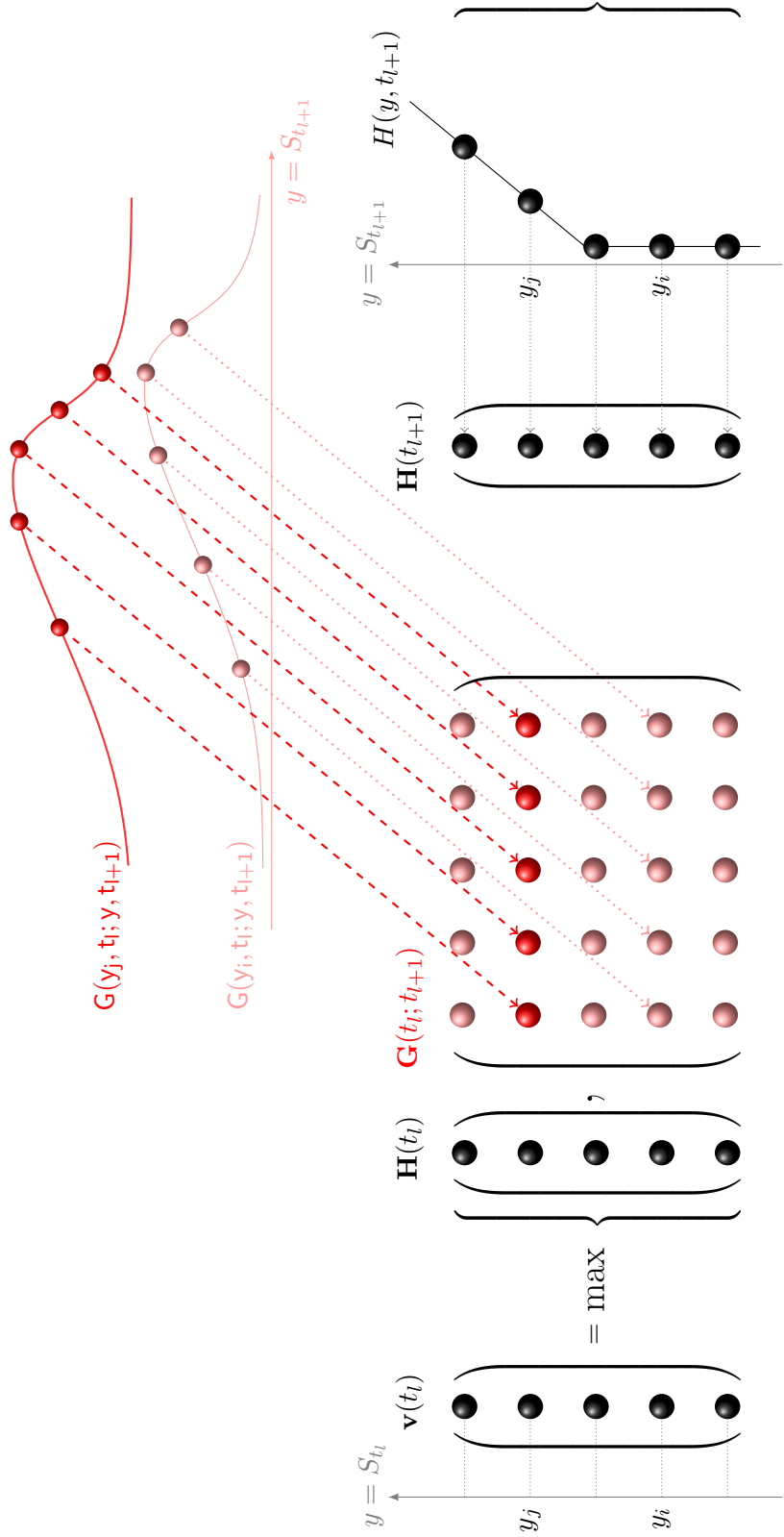}
  \end{tikzpicture}
  \caption{\label{fig:projection-1dmatrix}Recursive Projection scheme in the Black-Scholes case when grids are the same at each exercise date $t_{l}$. The rows of the $\mathbf{G}(t_{l}; t_{l+1})$ matrix are given by the values sampled from the transition densities $G(y_{i}, t_{l}; y_{j}, t_{l+1})$ for conditioning values $\{y_{1}, \dots, y_{N}\}$. In this way, the output of the vector times matrix multiplication contains the continuation values at the grid points $\{y_{1}, \dots, y_{N}\}$. Each continuation value at $t_{l}$ is compared with the intrinsic value $H(y_{i}, t_{l})$, to obtain the vector $\mathbf{v}(t_{l})$, whose elements are the approximation of $V(y,t_{l})$ at the same grid points $\{y_{1}, \dots, y_{N}\}$, and that is the input for the following recursive step.}
  \end{center}
\end{sidewaysfigure}



\begin{figure}
	\begin{center}
	\Large{\bf Panel A}
	\end{center}
	\begin{tikzpicture}[scale=0.7]
		\input{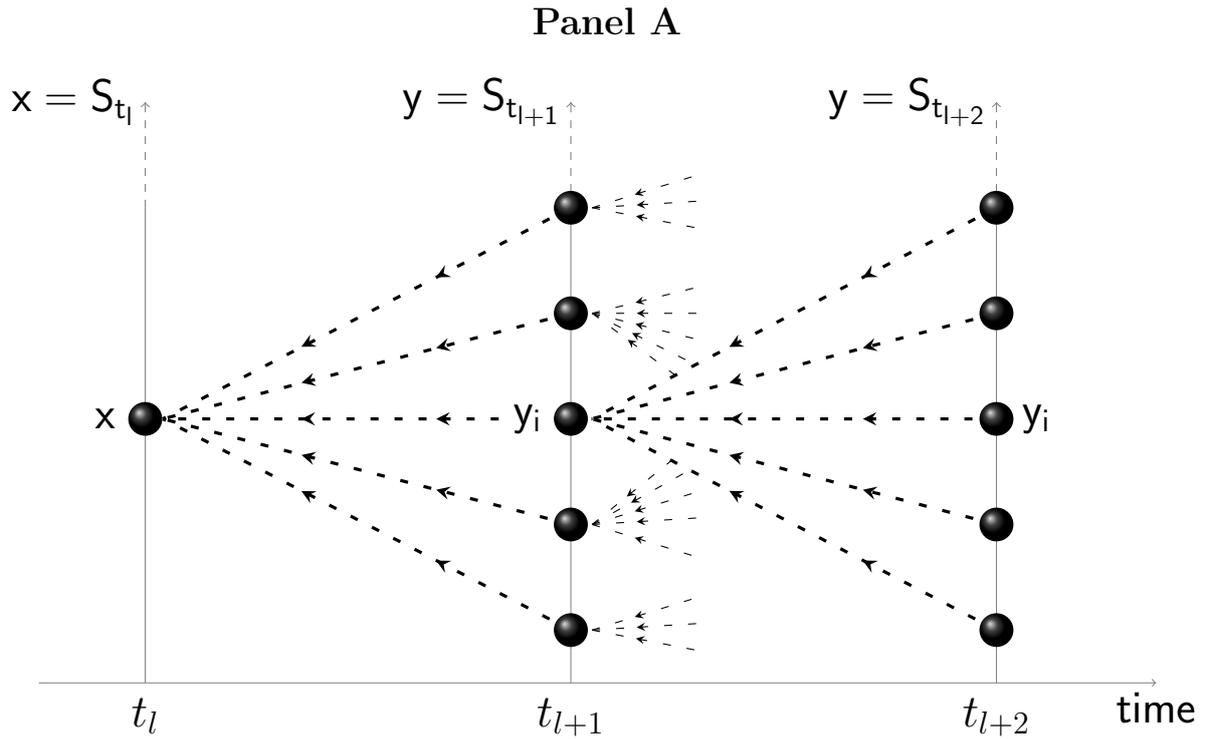}
	\end{tikzpicture}
	\bigskip
	\begin{center}
	\Large{\bf Panel B}
	\end{center}
	\begin{tikzpicture}[scale=.7] %
		\input{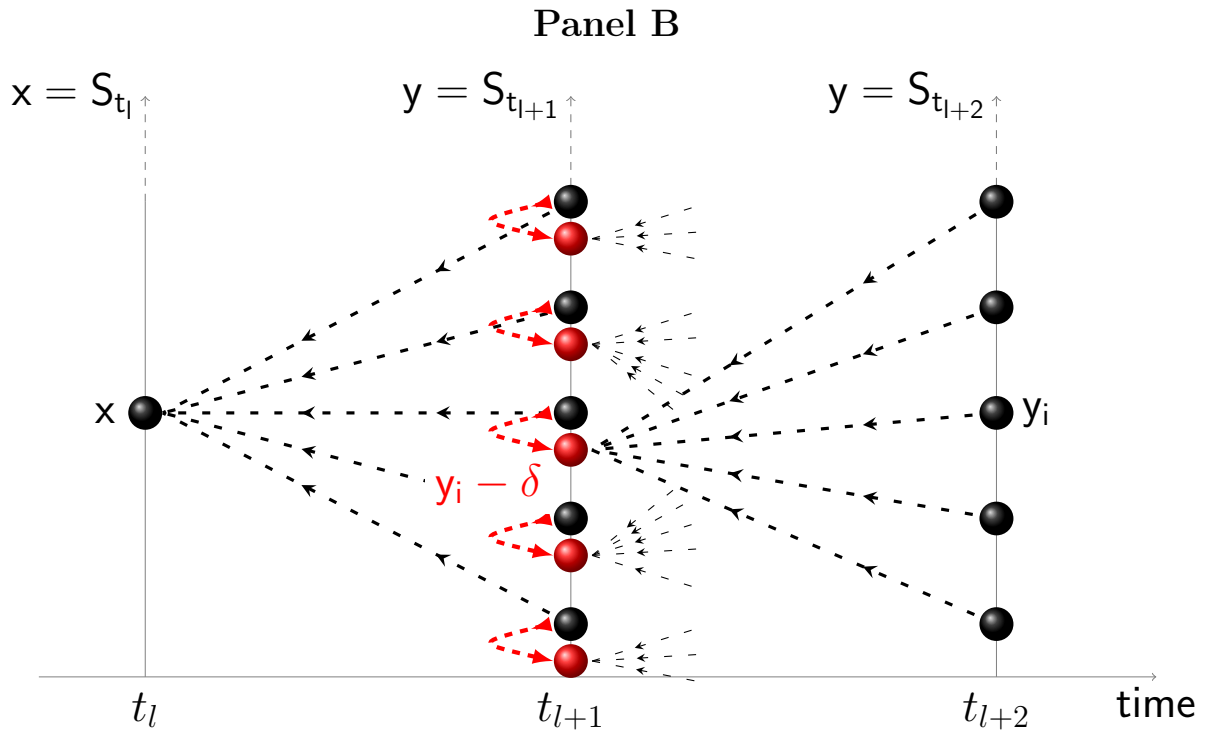}	
	\end{tikzpicture}
	\caption{\label{fig:dividends}Recursive scheme without dividends (Panel A) and with discrete dividends (Panel B). In Panel A, at date $t=t_{l+1}$, the intrinsic value $H(y_{i}, t_{l+1})$ is compared with the continuation value $V(y_{i}, t_{l+1})$ computed at the same grid point $y_{i}$ (black ball). In Panel B, at the ex-dividend date $t_{h}=t_{l+1}$, the intrinsic value $H(y_{i}, t_{l+1})$ at the grid point $y_{i}$ (black ball) is compared with the continuation value $V(y_i -  d, t_{l+1})$ at $y_{i}- d$ (red ball).}
\end{figure}



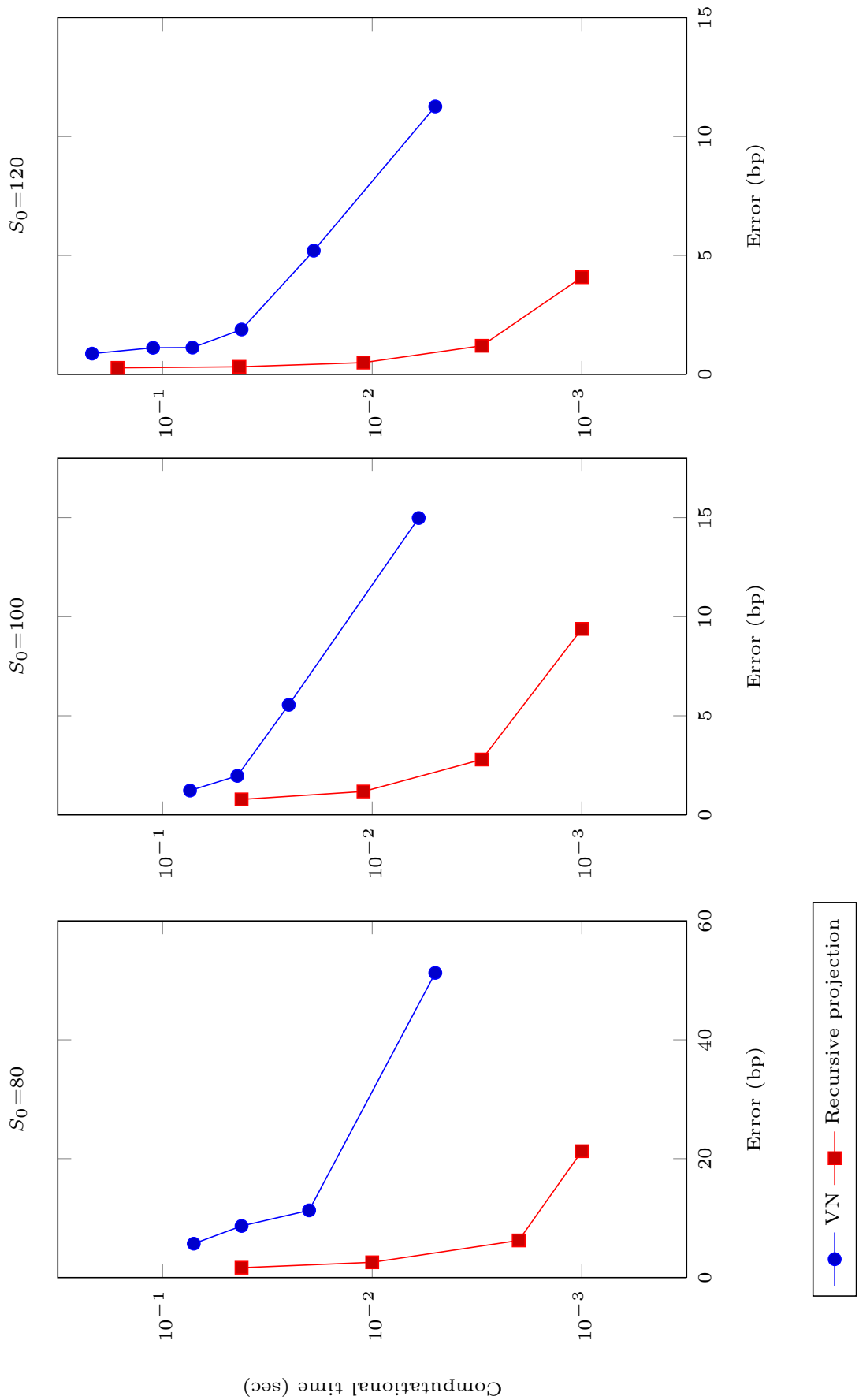
\begin{sidewaysfigure}
    \centering
   {		\pgfplotsset{domain=-1:1,
		legend style={
		at={(0.5,-0.2)}, anchor=north, legend columns=2, font=\tiny,
		},
		}
		\begin{tikzpicture}[baseline, scale=1.5]
		    \begin{axis}[
		        title = {$S_{0}$=80},
		        scale only axis,
                xmin=0,
                xmax=60,
                xlabel={Error (bp)},
                ymin=-3.5,
                ymax=-0.5,
                ytick={-3,-2,-1},
                yticklabels={$10^{-3}$,$10^{-2}$,$10^{-1}$},
		        font=\tiny,
		        xlabel=Error (bp),
		        ylabel= Computational time (sec),
		        width=4cm, height=7cm
		            ]
		    \addplot plot coordinates {
            (51.2491, -2.301029996)
            (11.3064, -1.698970004)
            (8.6846, -1.37675071)
            (5.7013, -1.148741651)
             };

		    \addplot plot coordinates {
            (21.261956, -3)
            (6.250285, -2.69897)
            (2.574474, -2)
            (1.664213, -1.37675071)
             };

		    \legend{VN\\Recursive projection\\}			\end{axis}
		\end{tikzpicture}
		\begin{tikzpicture}[baseline, scale=1.5]
		    \begin{axis}[
		        title = {$S_{0}$=100},
		        scale only axis,
                xmin=0,
                xmax=18,
                xlabel={Error (bp)},
                ymin=-3.5,
                ymax=-0.5,
                ytick={-3,-2,-1},
                yticklabels={$10^{-3}$,$10^{-2}$,$10^{-1}$},
		        font=\tiny,
		        xlabel=Error (bp),
		        width=4cm, height=7cm,
		        legend style={legend columns=1}
		    ]
		
		    \addplot coordinates {
            (14.977, -2.22184875)
            (5.5498, -1.602059991)
            (1.9682, -1.356547324)
            (1.2252, -1.13076828)};
		    \addplot coordinates {
            (9.384926, -3)
            (2.793496, -2.522878745)
            (1.179349, -1.958607315)
            (0.779497, -1.37675071)};
		    \end{axis}
		\end{tikzpicture}
		\begin{tikzpicture}[baseline, scale=1.5]
		    \begin{axis}[
		        title = {$S_{0}$=120},
		        scale only axis,
                xmin=0,
                xmax=15,
                xlabel={Error (bp)},
                ymin=-3.5,
                ymax=-0.5,
                ytick={-3,-2,-1},
                yticklabels={$10^{-3}$,$10^{-2}$,$10^{-1}$},
		        font=\tiny,
		        xlabel=Error (bp),
		        width=4cm, height=7cm,
		        legend style={legend columns=1}
		    ]
		
		    \addplot coordinates {
            (11.2649, -2.301029996)
            (5.1999, -1.721246399)
            (1.8857, -1.37675071)
            (1.1255, -1.142667504)
            (1.1173, -0.954677021)
            (0.8728, -0.663540266)};
		    \addplot coordinates {
            (4.082874, -3)
            (1.201872, -2.522878745)
            (0.49386, -1.958607315)
            (0.317633, -1.366531544)
            (0.273384, -0.785156152)};
		    \end{axis}
		\end{tikzpicture} }
    \caption{Comparison between the approximated binomial tree method (VN) of Vellekoop and Nieuwenhuis (2009) and the recursive projections on an American call option written on a dividend-paying stock in the Black-Scholes case. The option has a maturity of 3 years and a dividend $d=2$ is paid at the end of each year. Other parameters set equal to $r=0.05, \sigma=0.2, K=100$.}
    \label{fig:Vellekoop}
\end{sidewaysfigure}

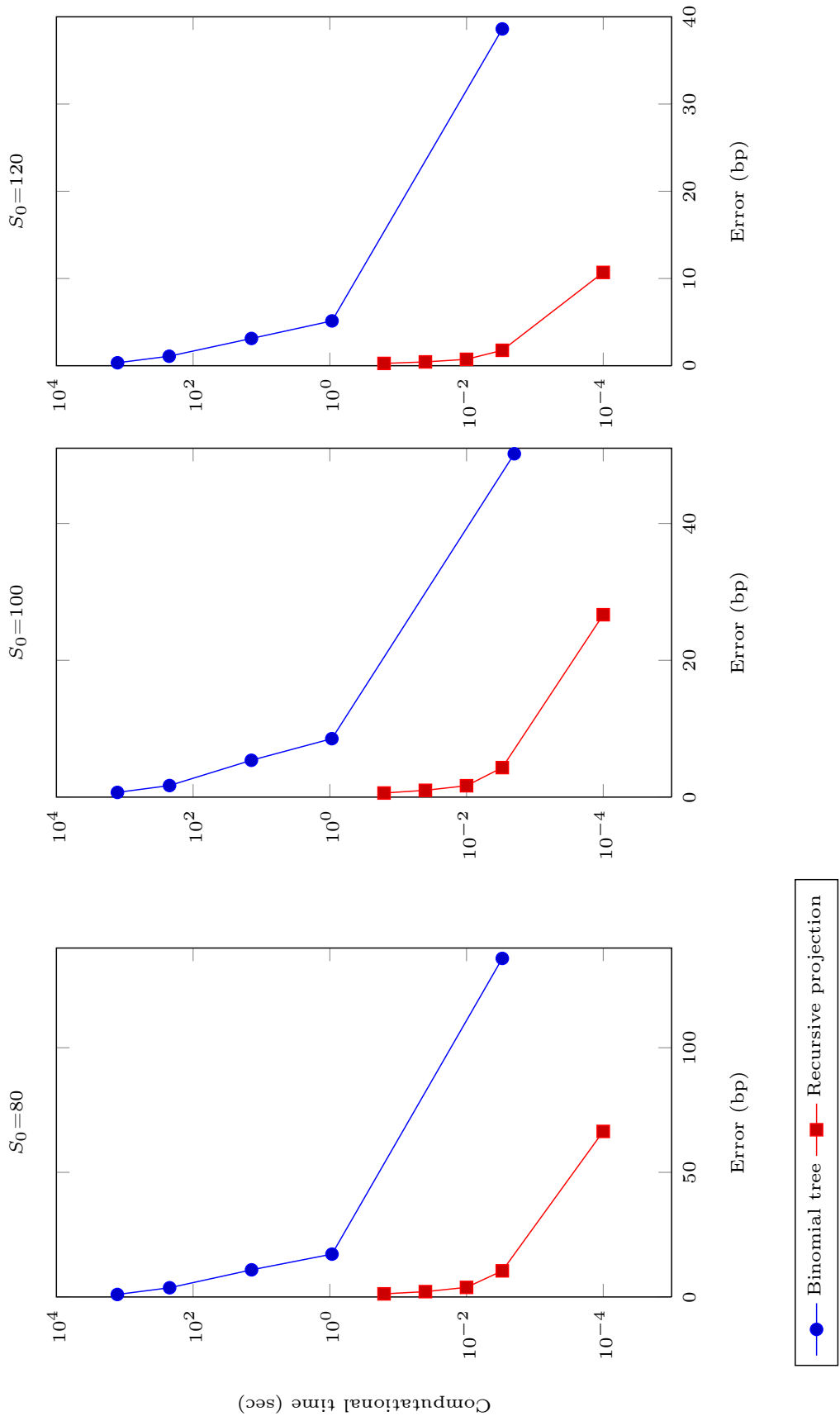
\begin{sidewaysfigure}
    \centering
 {		\pgfplotsset{domain=-1:1,
		legend style={
		at={(0.5,-0.2)}, anchor=north, legend columns=2, font=\tiny,
		},
		}
		\begin{tikzpicture}[baseline, scale=1.4]
		    \begin{axis}[
		        title = {$S_{0}$=80},
		        scale only axis,
                xmin=0,
                xmax=140,
                xlabel={Error (bp)},
                ymin=-5,
                ymax=4,
                ytick={-4,-2,0,2,4},
                yticklabels={$10^{-4}$,$10^{-2}$,$10^0$,$10^{2}$,$10^4$},
		        font=\tiny,
		        xlabel=Error (bp),
		        ylabel= Computational time (sec),
		        width=4cm, height=7cm
		            ]
		    \addplot plot coordinates {
            (135.834, -2.522878745)
            (17.181, -0.031517051)
            (10.883, 1.14500984)
            (3.696, 2.343770847)
            (1, 3.108838232)
             };

		    \addplot plot coordinates {
            (66.422510, -4)
            (10.485074, -2.522878745)
            (3.843822, -2)
            (2.137740, -1.397940009)
            (1.244239, -0.793174124)
             };

		    \legend{Binomial tree\\Recursive projection\\}		
		
		\end{axis}
		\end{tikzpicture}
		\begin{tikzpicture}[baseline, scale=1.4]
		    \begin{axis}[
		        title = {$S_{0}$=100},
		        scale only axis,
                xmin=0,
                xmax=51,
                xlabel={Error (bp)},
                ymin=-5,
                ymax=4,
                ytick={-4,-2,0,2,4},
                yticklabels={$10^{-4}$,$10^{-2}$,$10^0$,$10^{2}$,$10^4$},
		        font=\tiny,
		        xlabel=Error (bp),
		        width=4cm, height=7cm,
		        legend style={legend columns=1}
		    ]
		
		    \addplot coordinates {
                (50.186, -2.698970004)
                (8.544, -0.029653124)
                (5.385, 1.147738141)
                (1.692, 2.345307062)
                (0.691, 3.107624387)};
		    \addplot coordinates {
                (26.662105, -4)
                (4.317101, -2.522878745)
                (1.665226, -2)
                (0.995119, -1.397940009)
                (0.599089, -0.793174124)};
		    \end{axis}
		\end{tikzpicture}
		\begin{tikzpicture}[baseline, scale=1.4]
		    \begin{axis}[
		        title = {$S_{0}$=120},
		        scale only axis,
                xmin=0,
                xmax=40,
                xlabel={Error (bp)},
                ymin=-5,
                ymax=4,
                ytick={-4,-2,0,2,4},
                yticklabels={$10^{-4}$,$10^{-2}$,$10^0$,$10^{2}$,$10^4$},
		        font=\tiny,
		        xlabel=Error (bp),
		        width=4cm, height=7cm,
		        legend style={legend columns=1}
		    ]
		
		    \addplot coordinates {
                (38.610, -2.522878745)
                (5.138, -0.031517051)
                (3.125, 1.146995757)
                (1.09, 2.350348825)
                (0.34, 3.106145663)};
		    \addplot coordinates {
                (10.698617, -4)
                (1.768145, -2.522878745)
                (0.735743, -2)
                (0.439623, -1.397940009)
                (0.265091, -0.793174124)};
		    \end{axis}
		\end{tikzpicture} }
 \caption{Comparison between the binomial tree and the recursive projection method on an American call option written on a dividend-paying stock in the Black-Scholes case.
The option has a maturity of 3 years and a dividend $d=2$ is paid at the end of each year.
 Other parameters are set equal to $r=0.05$, $\sigma=0.2$, $K=100$.
 \label{Figure:BS_amcall_div}}
 \end{sidewaysfigure}



\begin{sidewaysfigure}
    \centering
 {		\pgfplotsset{domain=-1:1,
		legend style={
		at={(0.5,-0.2)}, anchor=north, legend columns=2, font=\tiny,
		},
		}
		\begin{tikzpicture}[baseline, scale=1.5]
		    \begin{axis}[
		        title = {$S_{0}$=80},
		        scale only axis,
                xmin=0,
                xmax=60,
                ymin=-2,
                ymax=4,
                ytick={-1,0,1,2,3},
                yticklabels={$10^{-1}$,$10^{0}$,$10^1$,$10^{2}$,$10^3$},
		        font=\tiny,
		        xlabel=Error (bp),
		        ylabel= Computational time (sec),
		        width=4cm, height=7cm
		            ]
		    \addplot plot coordinates {
            (16.619579, -0.293282218)
            (6.5143, 1.430800425)
            (9.9064, 2.099708258)
            (9.5363, 2.478081431)
            (9.5795, 2.744292201)
             };

		    \addplot plot coordinates {
            (51.838244, -1.366531544)
            (15.015252, -0.966576245)
            (5.973541, -0.632644079)
            (3.724238, -0.289036881)
            (3.162762, 0.067814511)
             };

		    \legend{Binomial tree\\Recursive projection\\}		
		
		\end{axis}
		\end{tikzpicture}
		\begin{tikzpicture}[baseline, scale=1.5]
		    \begin{axis}[
		        title = {$S_{0}$=100},
		        scale only axis,
                xmin=0,
                xmax=50,
                xlabel={Error (bp)},
                ymin=-2,
                ymax=4,
                ytick={-1,0,1,2,3},
                yticklabels={$10^{-1}$,$10^{0}$,$10^1$,$10^{2}$,$10^3$},
		        font=\tiny,
		        xlabel=Error (bp),
		        width=4cm, height=7cm,
		        legend style={legend columns=1}
		    ]
		
		    \addplot coordinates {
                    (44.4585, -0.297569464)
                    (19.1589, 1.436130829)
                    (14.5978, 2.097271196)
                    (10.6728, 2.487760371)
                    (10.5761, 2.743393728)};
		    \addplot coordinates {
                    (17.844330, -1.366531544)
                    (3.642263, -0.966576245)
                    (0.309681, -0.632644079)
                    (0.510702, -0.289036881)
                    (0.714960, 0.067814511)};
		    \end{axis}
		\end{tikzpicture}
		\begin{tikzpicture}[baseline, scale=1.5]
		    \begin{axis}[
		        title = {$S_{0}$=120},
		        scale only axis,
                xmin=0,
                xmax=2,
                xlabel={Error (bp)},
                ymin=-2,
                ymax=4,
                ytick={-1,0,1,2,3},
                yticklabels={$10^{-1}$,$10^{0}$,$10^1$,$10^{2}$,$10^3$},
		        font=\tiny,
		        xlabel=Error (bp),
		        width=4cm, height=7cm,
		        legend style={legend columns=1}
		    ]
		
		    \addplot coordinates {
                    (0.7399, -0.301029996)
                    (0.9424, 1.437718861)
                    (0.1287, 2.097062858)
                    (0.5679, 2.554635717)
                    (0.1223, 2.743381966)};
		    \addplot coordinates {
                (0.925933, -1.366531544)
                (0.264783, -0.966576245)
                (0.106413, -0.632644079)
                (0.067195, -0.289036881)
                (0.057415, 0.067814511)};
		    \end{axis}
		\end{tikzpicture} }
\caption{Comparison between the binomial tree and the recursive projection method on a Bermudan digital call option in the Black-Scholes case.
The option has a maturity of 10 years and can be exercised
 4 times per year. Other parameters are set equal to $r=0.1$, $\sigma=0.2$, and $K=100$.  \label{Figure:BS_berm_digitalcall}}
 \end{sidewaysfigure}
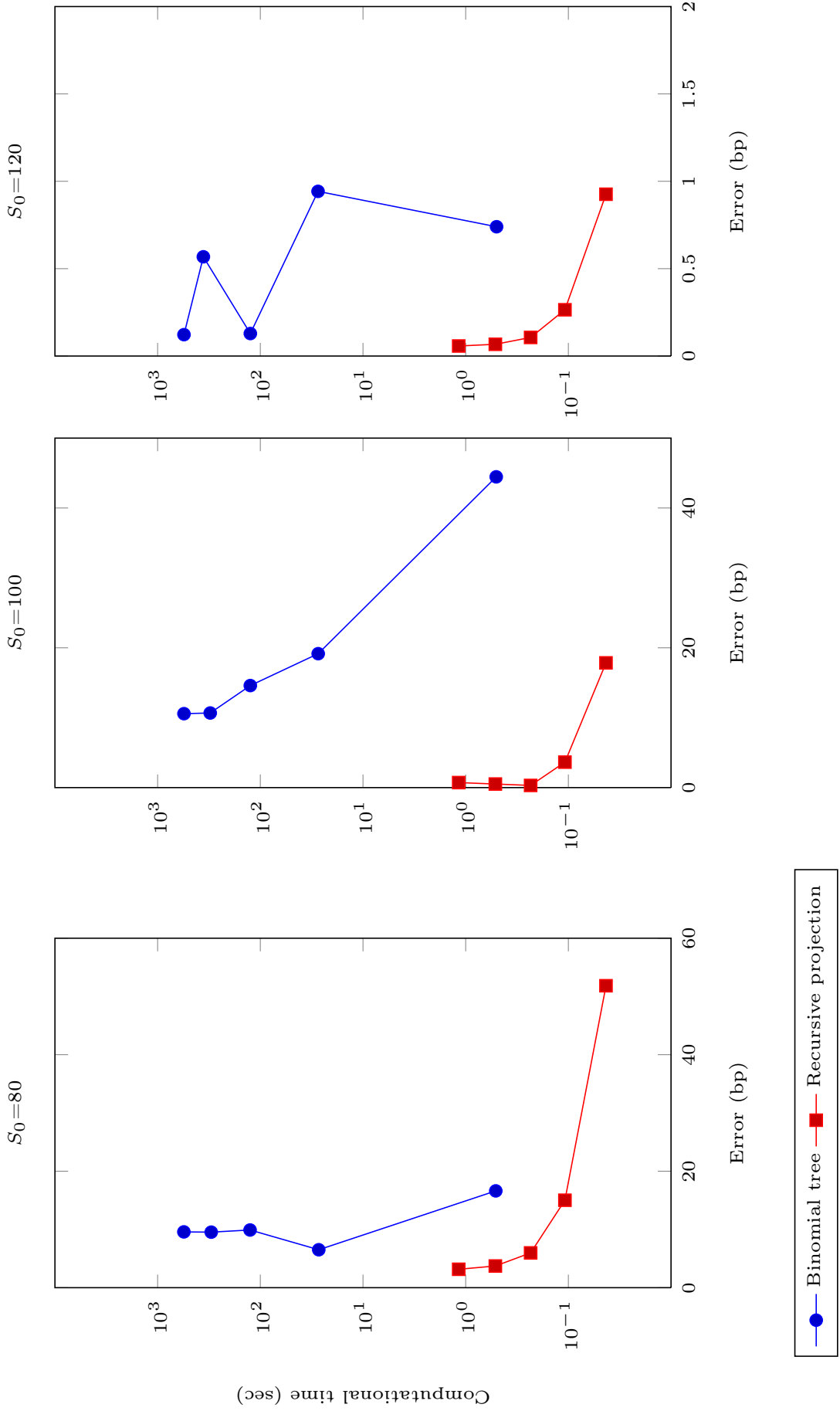



\begin{sidewaysfigure}
	\includegraphics[scale=0.9]{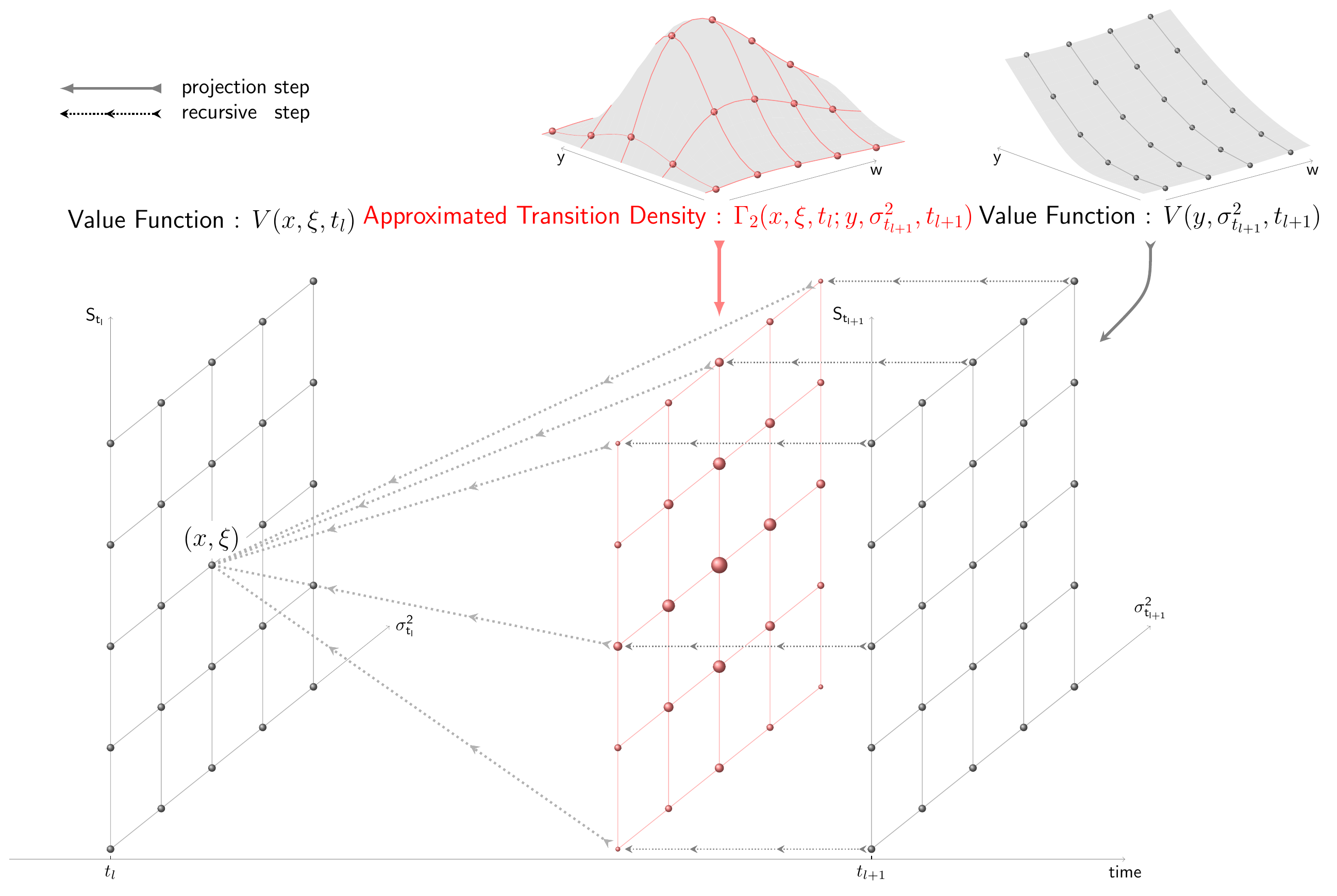}
	    \caption{
	    Recursive Projection scheme in the stochastic volatility case is composed of two steps. First the projection step (\emph{thick arrows}): the value function $V(y, w, t)$ at $t=t_{l+1}$ is sampled and the state price density function $G_{2}(x, \xi, t_{l}; y, w, t_{l+1})$ is approximated at $(y_{j}, w_{q})$ by $\Gamma_{2}(x, \xi, t_{l}; y_{j}, w_{q}, t_{l+1})$. Second the recursive step: the sampled values are multiplied by the transition weights to obtain the value function $V(x=y_{i}, \xi=w_{p}, t_{l})$ at $t=t_{l}$  (\emph{thin arrows}), which in turn will be the input for the following step of the algorithm. \label{fig:projection-2d}}
\end{sidewaysfigure}



	\begin{sidewaysfigure}		
		\pgfplotsset{domain=-1:1,
		legend style={
		at={(0.5,-0.3)}, anchor=north, legend columns=2, font=\tiny,
		},
		}		
		\begin{tikzpicture}[baseline, scale=2]
		
		    \begin{semilogyaxis}[
		    	title= Recursive Projections,
		        xlabel=Resolution level $J$,
		        ylabel=Absolute Relative Error (bp),
		        width=5cm, height=6cm,
		        legend style={legend columns=1},
		        log basis x= 2, log basis y=2, ytickten={-2,0, ...,6},xtickten={6,8,10},
		        font=\tiny,
		    ]
		
		    \addplot plot coordinates {
            (6, 225.3234)
			(7, 64.306165526595905)
			(8, 14.884584207645688)
			(9, 3.887641180947537)
			(10, 0.939904985045580)
			(11, 0.220514201184560)
               };
			\node[font=\tiny, coordinate,pin=right:{15bp, 0.8$s$}] at (axis cs:8,  14.884584207645688)	{};
			\node[font=\tiny, coordinate,pin=left:{4bp, 2$s$}] at (axis cs:9, 3.887641180947537)	{};
			\node[font=\tiny, coordinate,pin=left:{1bp, 8$s$}] at (axis cs:10, 0.939904985045580)	{};
            \node[font=\tiny, coordinate,pin=left:{0.2bp, 32$s$}] at (axis cs:11, 0.220514201184560)	{};
		
			\addplot[domain=6:11, color= red, thin, dashed] {2^(-2.078572297076128*x + 20.597050222584251)} ;
			\node[color=red] at (axis cs:8.5, 200) {\scalebox{0.8}{$\widehat{\log_{2}(err)} = -2.08J+20.6$}};

		    \end{semilogyaxis}
		\end{tikzpicture}\hspace{2cm}
		\begin{tikzpicture}[baseline, scale=2]
		
		    \begin{loglogaxis}[
			title= Finite Differences,
		        xlabel=Time steps $L_T$,
		        width=5cm, height=6cm,
		        log basis x=2, log basis y=2,  xtickten={6, 8,10},
		        font=\tiny,    ]
		
		    \addplot plot coordinates {
			(64, 254.7586)
			(128,  193.0398)
			(256,   161.5917)
			(512,  146.1948)
			(1024,   138.5890)
			(2048,  134.7347)
		    };
		
		    \addplot plot coordinates {
			(64,  72.2534)
			(128,  33.7862)
			(256,    12.2868)
			(512,   4.3254)
			(1024,    0.7252)
			(2048,   1.0782)
		    };
		
		    \addplot[color=brown,mark=triangle*, mark options={solid}] plot coordinates {
			(64,  53.0495)
			(128,   25.0731)
			(256,   11.6392)
			(512,   4.4119)
			(1024,    0.8199)
			(2048,   0.9836)
		    };
		
		    \node[coordinate, font=\tiny,pin=left:{{\color{red}1bp, 130s}}] at (axis cs:1024, 0.7252)	{};
		    \node[coordinate, font=\tiny,pin=left:{{\color{red}4bp, 65s}}] at (axis cs:512, 4.3254)	{};	

		    \legend{$m_{s}=1600$\\$m_{s}=3200$\\$m_{s}=6400$\\}		
		    \end{loglogaxis} \hspace{-5cm}
		\end{tikzpicture}
	\caption{Comparison between the finite difference scheme and the recursive projections on an American call option written on a dividend-paying stock in the Heston case. The option has a maturity of 1 year and a dividend $d = 2$ is paid at dates $t_{l}=0.25, 0.5, 0.75$. Other parameters are set equal to $S_{0}=100$, $K=100$, $r = 0.05$, $\sigma_{LT}=0.2$, $\beta=2$, $\omega = 0.2$. The parameter $m_{s}$ gives the umber of points in the $X_{t}= \log S_{t}$ grid for the \emph{FD} scheme, while the $\sigma^{2}$ grid has $m_{w}=31$ points. The sampling grid for the $X_{t}= \log S_{t}$ variable in the recursive projections has size $2^{J}$. The resolution level in the $\sigma^{2}$ dimension is $J_{w}=4$.
The dashed line in the left panel corresponds to a fitted linear regression, and shows that the estimated slope is close to the slope of $-2$ predicted by the theoretical convergence results.	
	\label{fig:Hes3divConv}}
	\end{sidewaysfigure}



	\begin{sidewaysfigure}
		\pgfplotsset{domain=-1:1,
		legend style={
		at={(0.5,-0.3)}, anchor=north, legend columns=2, font=\tiny,
		},
		}
		
		\begin{tikzpicture}[baseline, scale=2]
		    \begin{semilogyaxis}[
		        title = Recursive Projections,
		        log basis y=2, ytickten={-4,-2, ...,6, 8},
		        font=\tiny,
		        xlabel=Resolution level $J$,
		        ylabel= Absolute Relative Error (bp),
		        width=5cm, height=6cm,
		        legend style={legend columns=1}
		    ]
		
		    \addplot coordinates {
			(6, 93.012748850687)
			(7, 28.060331654882)
			(8, 6.238016986234)
			(9, 1.427137835960)
			(10, 0.403207161629)
			(11, 0.114928765042)   };
				\node[pin={right:6bp, 0.3$s$}] at (axis cs:8, 6.23)	{};
				\node[font=\tiny, coordinate,pin=left:{1bp, 0.8$s$}] at (axis cs:9, 1.427)	{};
				\node[font=\tiny, coordinate,pin=left:{0.4bp, 3$s$}] at (axis cs:10, 0.403)	{};
				\node[font=\tiny, coordinate,pin=left:{0.1bp, 10$s$}] at (axis cs:11, 0.1149)	{};
			\addplot[domain=6:11, color= red, thin, dashed] {2^(-2.062089632466064*x + 19.186607439706588)} ;
			\node[font=\tiny, color=red] at (axis cs:8.5, 64) {\scalebox{0.8}{$\widehat{\log_{2}(err)} = -2.1J+19$}};
		    \end{semilogyaxis}
		\end{tikzpicture}
		\hspace{2cm}
		\begin{tikzpicture}[baseline, scale=2]
		
		    \begin{loglogaxis}[
	          	title = Finite Differences,
		        xlabel=Time steps $L_T$,
		        width=5cm, height=6cm,
		        log basis y=2, log basis x=2, xtickten={6,8,...,11},
		        font=\tiny
		            ]
		    \addplot plot coordinates {
			(64,98.071571336248198)
			(128,36.600568803175825)
			(256, 5.859589814083233)
			(512, 9.512269111101862)
			(1024, 17.198883289014720)
            (2048,21.041505662650838)
             };

		    \addplot plot coordinates {
			(64,125.3834960047469)
			(128,63.5769829650675)
			(256, 32.6114173511439)
			(512, 17.1231568216251)
			(1024, 9.3215104700180)
            (2048,5.4213720095353)
             };

		    \addplot[color=brown,mark=triangle*, mark options={solid}] plot coordinates {
			(64,129.1905131817646)
			(128,67.0772476789007)
			(256, 36.0473188249336)
			(512, 20.5289308213515)
			(1024, 12.7409787761372)
            (2048,8.8449486075715)
             };    	

			\node[coordinate, font=\tiny, pin=left:{\color{red} 5bp, 48s}] at (axis cs:2048, 5.4213720095353)	{};
            \node[coordinate, font=\tiny, pin=85:{\color{red} 17bp, 12s}] at (axis cs:512, 17.1231568216251)	{};
		
		    \legend{$m_{s}=200$\\$m_{s}=400$\\$m_{s}=800$\\}		
		
		\end{loglogaxis}
		\end{tikzpicture}		
		\caption{Comparison between the finite difference scheme and the recursive projections on an American call option written on a dividend-paying stock in the Heston case. The option has a maturity of 1 year and a dividend $d = 10$ is paid at $t_{l}=0.5$. Other parameters are set equal to $S_{0}=100$, $K=100$, $r = 0.05$, $\sigma_{LT}=0.2$, $\beta=2$, $\omega = 0.2$. The parameter $m_{s}$ gives the number of points in the $X_{t}= \log S_{t}$ grid for the \emph{FD} scheme, while the $\sigma^{2}$ grid has $m_{w}=31$ points. The sampling grid for the $X_{t}= \log S_{t}$ variable in the recursive projections has size $2^{J}$. The resolution level in the $\sigma^{2}$ dimension is $J_{w}=4$.
The dashed line in the left panel corresponds to a fitted linear regression, and shows that the estimated slope is close to the slope of $-2$ predicted by the theoretical convergence results.
		\label{fig:Hes1divConv}}
\end{sidewaysfigure}



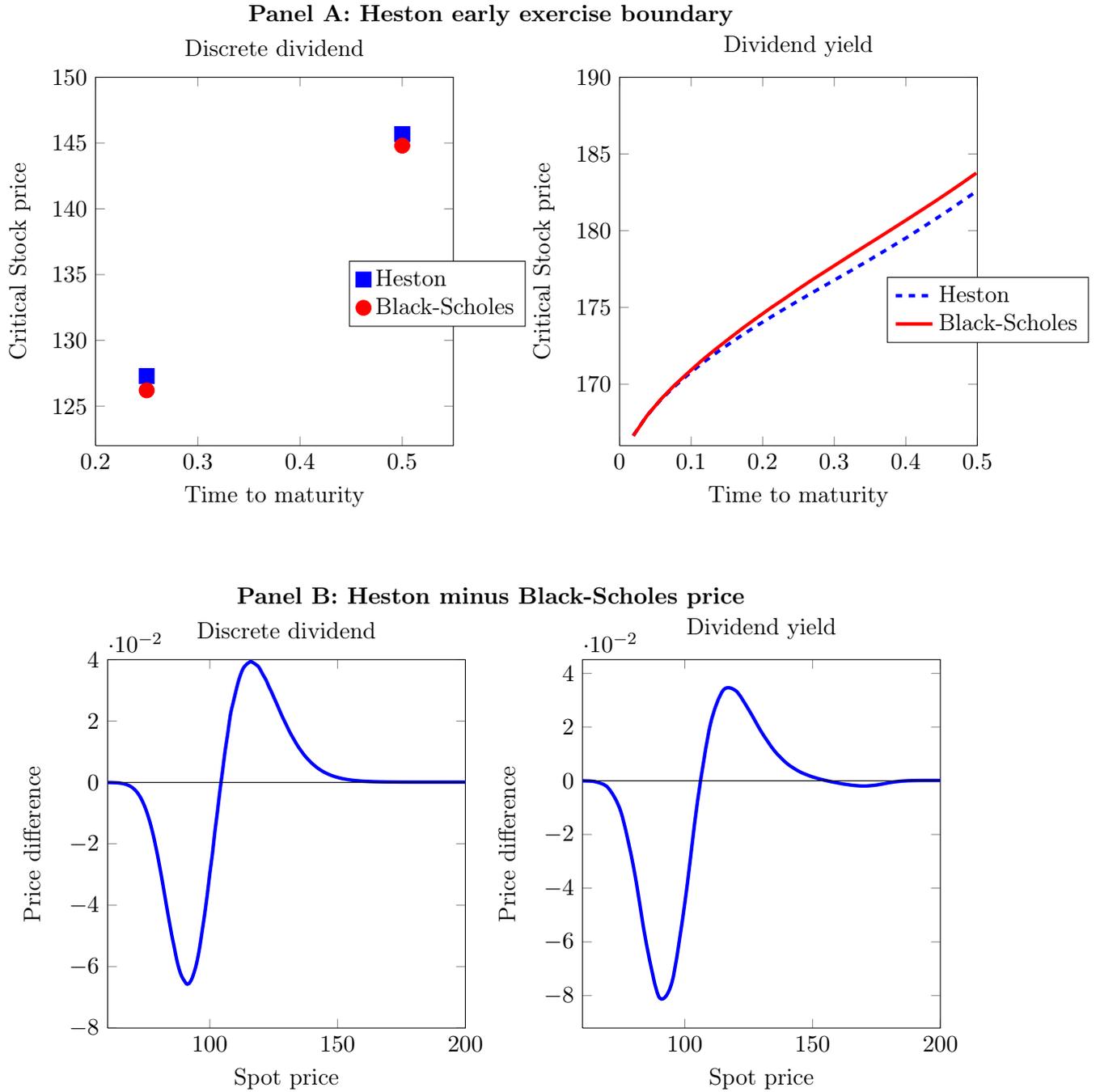
\begin{figure}
	\begin{center}
\title{\textbf{Panel A: Heston early exercise boundary}}
  \definecolor{mycolor1}{rgb}{0,0.498039215803146,0}%
\begin{tikzpicture}

\begin{axis}[%
width=2.91432291666667in,
height=3in,
xmin=0.2,
xmax=0.55,
xlabel={Time to maturity},
ymin=122,
ymax=150,
ylabel={Critical Stock price},
name=plot1,
title={Discrete dividend},
legend style={at={(0.708494389689963,0.322918036309861)},anchor=south west,draw=black,fill=white,legend cell align=left}
]
\addplot [
color=blue,
mark size=3.5pt,
only marks,
mark=square*,
mark options={solid,draw=blue}
]
table[row sep=crcr]{
0.25 127.29\\
0.5 145.68\\
};
\addlegendentry{Heston};

\addplot [
color=red,
mark size=3.5pt,
only marks,
]
table[row sep=crcr]{
0.25 126.2\\
0.5 144.8\\
};
\addlegendentry{Black-Scholes};

\end{axis}

\begin{axis}[%
width=2.91432291666667in,
height=3in,
xmin=0,
xmax=0.5,
xlabel={Time to maturity},
ymin=166,
ymax=190,
ylabel={Critical Stock price},
at=(plot1.right of south east),
anchor=left of south west,
title={Dividend yield},
legend style={at={(0.747601540616246,0.278998778998779)},anchor=south west,draw=black,fill=white,legend cell align=left}
]
\addplot [
color=blue,
solid,
dashed,
ultra thick,
]
table[row sep=crcr]{
0.0191780821917808 166.65\\
0.0383561643835616 167.96\\
0.0575342465753425 168.98\\
0.0767123287671233 169.85\\
0.0958904109589041 170.63\\
0.115068493150685 171.35\\
0.134246575342466 172.01\\
0.153424657534247 172.639999999999\\
0.172602739726027 173.239999999999\\
0.191780821917808 173.809999999999\\
0.210958904109589 174.359999999999\\
0.23013698630137 174.899999999999\\
0.249315068493151 175.419999999999\\
0.268493150684932 175.939999999999\\
0.326027397260274 177.469999999999\\
0.345205479452055 177.989999999999\\
0.364383561643836 178.519999999999\\
0.402739726027397 179.599999999999\\
0.441095890410959 180.749999999998\\
0.46027397260274 181.339999999998\\
0.479452054794521 181.949999999998\\
0.498630136986301 182.559999999998\\
};
\addlegendentry{Heston};

\addplot [
color=red,
solid, ultra thick
]
table[row sep=crcr]{
0.0191780821917808 166.64\\
0.0383561643835616 167.95\\
0.0575342465753425 169.01\\
0.0767123287671233 169.93\\
0.0958904109589041 170.75\\
0.115068493150685 171.55\\
0.134246575342466 172.28\\
0.153424657534247 172.979999999999\\
0.172602739726027 173.669999999999\\
0.191780821917808 174.319999999999\\
0.210958904109589 174.959999999999\\
0.23013698630137 175.559999999999\\
0.249315068493151 176.179999999999\\
0.268493150684932 176.779999999999\\
0.326027397260274 178.499999999999\\
0.345205479452055 179.059999999998\\
0.364383561643836 179.619999999998\\
0.402739726027397 180.759999999998\\
0.441095890410959 181.919999999998\\
0.46027397260274 182.519999999998\\
0.479452054794521 183.129999999998\\
0.498630136986301 183.759999999998\\
};
\addlegendentry{Black-Scholes};

\end{axis}

\end{tikzpicture}%
\vspace{0.5cm} 

  \title{\textbf{Panel B: Heston minus Black-Scholes price}}
%
%
\begin{tikzpicture}

\begin{axis}[%
width=2.91432291666667in,
height=3in,
xmin=60,
xmax=200,
xlabel={Spot price},
ymin=-0.08,
ymax=0.04,
name=plot1,
ylabel={Price difference},
title={Discrete dividend }
]
\addplot [
color=blue,
solid, ultra thick,
forget plot, smooth
]
table[row sep=crcr]{
60 -1.37268986766929e-005\\
61 -2.56084076991461e-005\\
62 -4.63038487589107e-005\\
63 -8.07858747644334e-005\\
64 -0.000137220678365984\\
65 -0.000227144935902789\\
66 -0.000364554477977526\\
67 -0.000568388999358634\\
68 -0.00087082776803634\\
69 -0.00129765592560227\\
70 -0.00188390360642661\\
71 -0.00267105290740116\\
72 -0.00371940301467325\\
73 -0.00506778507318347\\
74 -0.00677674976026913\\
75 -0.0088438286810181\\
76 -0.0113505893364224\\
77 -0.0142222128791843\\
78 -0.017630272619781\\
79 -0.0214833333006747\\
80 -0.0256662111864778\\
81 -0.0301962895177031\\
82 -0.0349806737488172\\
83 -0.0398875009214248\\
84 -0.0446056543034896\\
85 -0.0492387512376216\\
86 -0.0532916058585351\\
87 -0.0574773273389418\\
88 -0.0607959206979403\\
89 -0.0631263463948849\\
90 -0.0645604409772731\\
91 -0.0656677951631863\\
92 -0.065353285498629\\
93 -0.0639269806399736\\
94 -0.0615314400067999\\
95 -0.0583675605012686\\
96 -0.0539484687998164\\
97 -0.0485097509138837\\
98 -0.0429188448146745\\
99 -0.0366684334923555\\
100 -0.029757843904366\\
101 -0.0234681042267546\\
102 -0.0160144868237042\\
103 -0.00949394690239735\\
104 -0.00243931762229099\\
105 0.003740706473363\\
106 0.0103282167516472\\
107 0.015712577747875\\
108 0.0217579291543135\\
109 0.0254901628183966\\
110 0.0289794839376949\\
111 0.0323578661166568\\
112 0.0350133633868239\\
113 0.0370445454708364\\
114 0.0381515952918203\\
115 0.0388859553674017\\
116 0.0393406799984461\\
117 0.0389887870940235\\
118 0.0384067490230393\\
119 0.0377773588329511\\
120 0.0363596897245415\\
121 0.0348047607600321\\
122 0.0334023942037227\\
123 0.0314386286346711\\
124 0.0297304065408426\\
125 0.0278349541046516\\
126 0.0259843917692706\\
127 0.0240529590186043\\
128 0.0221109679328535\\
129 0.0203299783297766\\
130 0.0186358443252708\\
131 0.0168861947801311\\
132 0.0153180557347028\\
133 0.0138896177425991\\
134 0.0124414887322146\\
135 0.011162550922144\\
136 0.010004830615685\\
137 0.00887737571770941\\
138 0.00791389682328258\\
139 0.00702969408064291\\
140 0.00620736452496118\\
141 0.00549193591961483\\
142 0.00482242542559419\\
143 0.00423079398603932\\
144 0.00371875623620355\\
145 0.00325085252563184\\
146 0.00284816707010549\\
147 0.00247940903351918\\
148 0.00215950605870319\\
149 0.00188533574020511\\
150 0.0016395217508034\\
151 0.00142821645371782\\
152 0.00123950252237393\\
153 0.00108180799767155\\
154 0.000939721084932899\\
155 0.000819251146900513\\
156 0.000715570677115807\\
157 0.000625605072279711\\
158 0.000548322964320391\\
159 0.000482621759864799\\
160 0.00042561829580734\\
161 0.000377460527715812\\
162 0.000336363255051708\\
163 0.000300879275336285\\
164 0.000271442011410272\\
165 0.000245755136617731\\
166 0.000224672330517706\\
167 0.000206422021989283\\
168 0.000191298123269235\\
169 0.000178440207790231\\
170 0.000167719513569864\\
171 0.000158745282945461\\
172 0.000151248650610114\\
173 0.000145056540134192\\
174 0.000139901551705179\\
175 0.000135694250474216\\
176 0.00013222781369393\\
177 0.000129432297271137\\
178 0.000127175689925707\\
179 0.000125371734824853\\
180 0.00012398753486309\\
181 0.000122901500404282\\
182 0.000122120627196409\\
183 0.000121553753956505\\
184 0.000121187846787052\\
185 0.000120986559593916\\
186 0.000120915143625666\\
187 0.000120969238281532\\
188 0.000121114689562773\\
189 0.000121338771975843\\
190 0.00012163560973022\\
191 0.000121985714869766\\
192 0.000122382331255722\\
193 0.0001228212323241\\
194 0.000123290579750801\\
195 0.000123788031956451\\
196 0.000124309933809741\\
197 0.000124849221307954\\
198 0.000125404739179658\\
199 0.000125974694668685\\
200 0.000126554942823986\\
};
\addplot [
color=black,
solid,
smooth,
forget plot
]
table[row sep=crcr]{
60 0\\
65 0\\
70 0\\
75 0\\
80 0\\
85 0\\
90 0\\
95 0\\
100 0\\
105 0\\
110 0\\
115 0\\
120 0\\
125 0\\
130 0\\
135 0\\
140 0\\
145 0\\
150 0\\
155 0\\
160 0\\
165 0\\
170 0\\
175 0\\
180 0\\
185 0\\
190 0\\
195 0\\
200 0\\
};
\end{axis}

\begin{axis}[%
width=2.91432291666667in,
height=3in,
xmin=60,
xmax=200,
xlabel={Spot price},
ylabel={Price difference},
at=(plot1.right of south east),
anchor=left of south west,
title={Dividend yield}
]
\addplot [
color=blue,
solid, ultra thick,
forget plot, smooth
]
table[row sep=crcr]{
60 -2.41970856090781e-005\\
65 -0.000347586731462181\\
70 -0.00262579082533956\\
75 -0.0115761329620721\\
80 -0.0323852633123489\\
85 -0.0608424678304537\\
90 -0.0806501974421425\\
95 -0.0750100265171954\\
100 -0.0457033174567574\\
105 -0.00761330763384471\\
110 0.0209305249943395\\
115 0.0336694695623265\\
120 0.0334845301837703\\
125 0.0264726903387356\\
130 0.0181718827172048\\
135 0.0112492892558507\\
140 0.00641960200420044\\
145 0.00335492753091415\\
150 0.00143878039473577\\
155 0.000144288213647314\\
160 -0.000839004945255795\\
165 -0.00159898409670234\\
170 -0.00196362757756674\\
175 -0.00158588757685152\\
180 -0.000711620179913552\\
185 -0.000106027253380603\\
190 8.7182459949986e-005\\
195 0.00012103853583767\\
200 0.000126988260220173\\
};
\addplot [
color=black,
solid,
forget plot
]
table[row sep=crcr]{
60 0\\
65 0\\
70 0\\
75 0\\
80 0\\
85 0\\
90 0\\
95 0\\
100 0\\
105 0\\
110 0\\
115 0\\
120 0\\
125 0\\
130 0\\
135 0\\
140 0\\
145 0\\
150 0\\
155 0\\
160 0\\
165 0\\
170 0\\
175 0\\
180 0\\
185 0\\
190 0\\
195 0\\
200 0\\
};
\end{axis}

\end{tikzpicture}%
  \caption{Panel A. Comparison between the early exercise boundary in the Heston and Black-Scholes models of an American call with maturity 6 months, in the case in which the stock pays a dividend yield $r_d=0.03$ (right) and in the case in which the stock distributes an equivalent quarterly discrete dividend of $d=1.38$ (left). The remaining parameters are: $K=100$, $r=0.05$, $\sigma_{0}=0.2$, $\omega=0.1$, $\sigma_{LT}=0.3$, $\beta=4$, $\rho=-0.5$. We set the volatility parameter in the Black-Scholes model equal to the volatility of the underlying return over the life of the option in the Heston model.
Panel B. Heston minus Black-Scholes price of an American call with $T=0.25$ for different values of $S _{0}$ in the case of discrete dividend (left) and continuous dividend yield (right). The remaining parameters are the same as in Panel A.\label{Heston_exercise_boundary}}
	
\end{center}
\end{figure}



\begin{figure}
	\begin{center}
\title{\textbf{Panel A: Merton early exercise boundary}}
  \begin{tikzpicture}

\begin{axis}[%
width=2.91432291666667in,
height=3in,
xmin=0.2,
xmax=0.55,
xlabel={Time to maturity},
ymin=60,
ymax=78,
name=plot1,
ylabel={Critical Stock price},
title={Discrete dividend},
legend style={at={(0.7,0.2)},anchor=south west,draw=black,fill=white,legend cell align=left}
]
\addplot [
color=blue,
mark size=4pt,
only marks,
mark=square*,
]
table[row sep=crcr]{
0.25 63.56\\
0.5 74.59\\
};
\addlegendentry{Merton};

\addplot [
color=red,
mark size=4pt,
only marks,
]
table[row sep=crcr]{
0.25 61.34\\
0.5 73.97\\
};
\addlegendentry{Black-Scholes};

\end{axis}
\begin{axis}[%
width=2.91432291666667in,
height=3in,
xmin=0,
xmax=0.5,
xlabel={Time to maturity},
ymin=60,
ymax=95,
ylabel={Critical Stock price},
at=(plot1.right of south east),
anchor=left of south west,
title={Dividend yield},
legend style={at={(0.65,0.25)},anchor=south west, draw=black, fill=white,legend cell align=left}
]
\addplot [
color=blue,
dashed,
ultra thick,
]
table[row sep=crcr]{
0.0191780821917808 67.72\\
0.0383561643835616 69.4399999999999\\
0.0575342465753425 70.93\\
0.0767123287671233 72.29\\
0.0958904109589041 73.54\\
0.115068493150685 74.71\\
0.134246575342466 75.84\\
0.153424657534247 76.9000000000001\\
0.172602739726027 77.9200000000001\\
0.191780821917808 78.9100000000001\\
0.210958904109589 79.87\\
0.23013698630137 80.8\\
0.249315068493151 81.69\\
0.268493150684932 82.57\\
0.287671232876712 83.35\\
0.306849315068493 84.2\\
0.326027397260274 85.08\\
0.345205479452055 85.8899999999999\\
0.364383561643836 86.68\\
0.383561643835616 87.38\\
0.402739726027397 88.2099999999999\\
0.421917808219178 88.8899999999999\\
0.441095890410959 89.69\\
0.46027397260274 90.3999999999999\\
0.479452054794521 91.1199999999999\\
0.498630136986301 91.8199999999999\\
};
\addlegendentry{Merton};

\addplot [
color=red,
solid, ultra thick, smooth
]
table[row sep=crcr]{
0.0191780821917808 63.99\\
0.0383561643835616 65.33\\
0.0575342465753425 66.42\\
0.0767123287671233 67.53\\
0.0958904109589041 68.7700000000001\\
0.115068493150685 70.1400000000001\\
0.134246575342466 71.5500000000001\\
0.153424657534247 72.9500000000001\\
0.172602739726027 74.3400000000001\\
0.191780821917808 75.6800000000001\\
0.210958904109589 76.9800000000002\\
0.23013698630137 78.2400000000002\\
0.249315068493151 79.4700000000002\\
0.268493150684932 80.6700000000002\\
0.287671232876712 81.6800000000002\\
0.306849315068493 82.8200000000002\\
0.326027397260274 84.0400000000002\\
0.345205479452055 85.0900000000003\\
0.364383561643836 86.1300000000003\\
0.383561643835616 87.0300000000003\\
0.402739726027397 88.1200000000003\\
0.421917808219178 88.9900000000003\\
0.441095890410959 90.0100000000003\\
0.46027397260274 90.9400000000003\\
0.479452054794521 91.8300000000003\\
0.498630136986301 92.7100000000003\\
};
\addlegendentry{Black-Scholes};

\end{axis}

\end{tikzpicture}%
\vspace{0.5cm} 

  \title{\textbf{Panel B: Merton minus Black-Scholes price}}
   \input{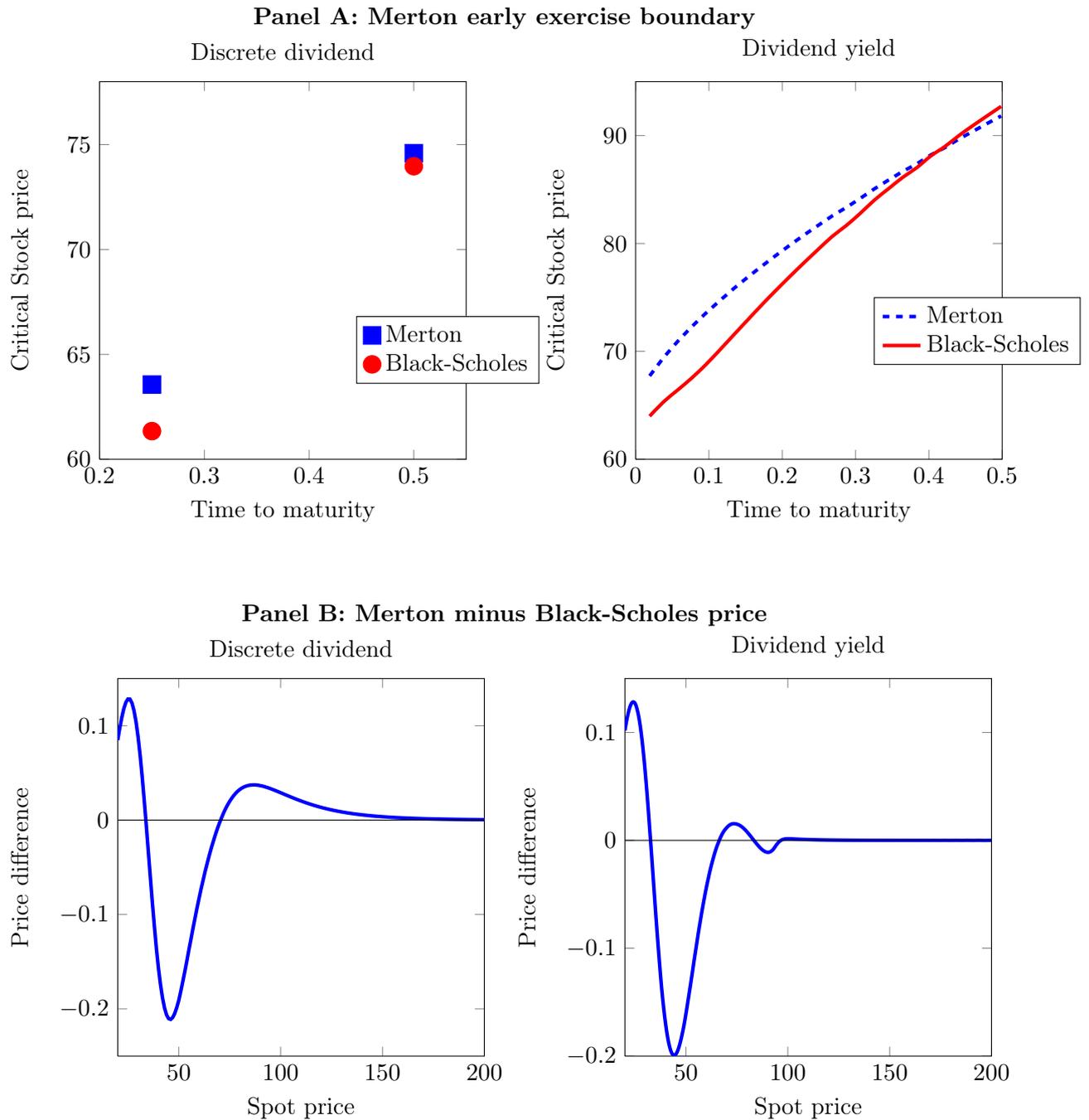}
  \caption{Panel A. Comparison between the early exercise boundary in the Merton and Black-Scholes models of an American call with maturity 6 months, in the case in which the stock pays a dividend yield $r_d=0.05$ (right) and in the case in which the stock distributes an equivalent quarterly discrete dividend of $ d=1.125$ (left). The other parameters are the following: $K=40$, $T=0.5$, $r=0.08$, $\gamma=5$, $\sigma_{M}^2=0.05$, $\sigma_{\psi}^2=0.05$, $\mu_{\psi}=0$. We set the volatility parameter in the Black-Scholes model equal to the volatility of the underlying return over the life of the option in the Merton model.
Panel B. Merton minus Black-Scholes price of an American call with the same parameters as those used in Panel A but different values of $S_0$ in the case of discrete dividend (left) and continuous dividend yield (right).\label{Merton_exercise_boundary}}
	
\end{center}
\end{figure}



\begin{table}

\begin{tabular}{lcllllllll}
 &  &  & \multicolumn{7}{c}{{\scriptsize Number of option quotes}}\tabularnewline
\cline{4-10}
 &  &  & \multicolumn{3}{c}{{\scriptsize maturity<60 days}} &  & \multicolumn{3}{c}{{\scriptsize maturity>60 days}}\tabularnewline
\cline{4-6} \cline{8-10}
 & {\scriptsize Ticker} & {\scriptsize Stock} & {\scriptsize OTM} & {\scriptsize ATM} & {\scriptsize ITM} &  & {\scriptsize OTM} & {\scriptsize ATM} & {\scriptsize ITM}\tabularnewline
\hline
{\scriptsize 1} & {\scriptsize{}   AA} & {\scriptsize{}    ALCOA INC} & {\scriptsize 42827} & {\scriptsize 6379} & {\scriptsize 41206} &  & {\scriptsize 28788} & {\scriptsize 9621} & {\scriptsize 29371}\tabularnewline
{\scriptsize 2} & {\scriptsize{}    AXP} & {\scriptsize{}    AMERICAN EXPRESS CO} & {\scriptsize 43360} & {\scriptsize 10386} & {\scriptsize 59656} &  & {\scriptsize 27744} & {\scriptsize 14941} & {\scriptsize 50002}\tabularnewline
{\scriptsize 3} & {\scriptsize{}    BAC} & {\scriptsize{}    BANK OF AMERICA CO} & {\scriptsize 53785} & {\scriptsize 7767} & {\scriptsize 57794} &  & {\scriptsize 48365} & {\scriptsize 15446} & {\scriptsize 58196}\tabularnewline
{\scriptsize 4} & {\scriptsize{}    BA} & {\scriptsize{}    BOEING CO} & {\scriptsize 41616} & {\scriptsize 7385} & {\scriptsize 52422} &  & {\scriptsize 36022} & {\scriptsize 13312} & {\scriptsize 50357}\tabularnewline
{\scriptsize 5} & {\scriptsize{}    CAT} & {\scriptsize{}    CATERPILLAR INC DEL} & {\scriptsize 40592} & {\scriptsize 8651} & {\scriptsize 55672} &  & {\scriptsize 34438} & {\scriptsize 15377} & {\scriptsize 54361}\tabularnewline
{\scriptsize 6} & {\scriptsize{}    CHV} & {\scriptsize{}    CHEVRON CORPORATION} & {\scriptsize 33444} & {\scriptsize 5922} & {\scriptsize 50092} &  & {\scriptsize 28088} & {\scriptsize 11238} & {\scriptsize 51676}\tabularnewline
{\scriptsize 7} & {\scriptsize{}    CSCO} & {\scriptsize{}    CISCO SYS INC} & {\scriptsize 58276} & {\scriptsize 8194} & {\scriptsize 63177} &  & {\scriptsize 44295} & {\scriptsize 14142} & {\scriptsize 55598}\tabularnewline
{\scriptsize 8} & {\scriptsize{}    KO} & {\scriptsize{}    COCA COLA CO} & {\scriptsize 34710} & {\scriptsize 5055} & {\scriptsize 44407} &  & {\scriptsize 29054} & {\scriptsize 9249} & {\scriptsize 44421}\tabularnewline
{\scriptsize 9} & {\scriptsize{}    DIS} & {\scriptsize{}    DISNEY WALT CO} & {\scriptsize 40138} & {\scriptsize 7052} & {\scriptsize 47571} &  & {\scriptsize 31350} & {\scriptsize 10730} & {\scriptsize 40773}\tabularnewline
{\scriptsize 10} & {\scriptsize{}    XOM} & {\scriptsize{}    EXXON MOBIL CORP} & {\scriptsize 39509} & {\scriptsize 6427} & {\scriptsize 55468} &  & {\scriptsize 27944} & {\scriptsize 9425} & {\scriptsize 48375}\tabularnewline
{\scriptsize 11} & {\scriptsize{}    GE} & {\scriptsize{}    GENERAL ELECTRIC CO} & {\scriptsize 49637} & {\scriptsize 6734} & {\scriptsize 58779} &  & {\scriptsize 42696} & {\scriptsize 13072} & {\scriptsize 58291}\tabularnewline
{\scriptsize 12} & {\scriptsize{}    HWP} & {\scriptsize{}    HEWLETT PACKARD CO} & {\scriptsize 54913} & {\scriptsize 9182} & {\scriptsize 53498} &  & {\scriptsize 47707} & {\scriptsize 17077} & {\scriptsize 51029}\tabularnewline
{\scriptsize 13} & {\scriptsize{}    HD} & {\scriptsize{}    HOME DEPOT INC} & {\scriptsize 42968} & {\scriptsize 7052} & {\scriptsize 56732} &  & {\scriptsize 36047} & {\scriptsize 13413} & {\scriptsize 53921}\tabularnewline
{\scriptsize 14} & {\scriptsize{}    INTC} & {\scriptsize{}    INTEL CORP} & {\scriptsize 55825} & {\scriptsize 8666} & {\scriptsize 61243} &  & {\scriptsize 43796} & {\scriptsize 14327} & {\scriptsize 52367}\tabularnewline
{\scriptsize 15} & {\scriptsize{}    IBM} & {\scriptsize{}    INTER. BUS. MACHS} & {\scriptsize 70249} & {\scriptsize 9060} & {\scriptsize 85569} &  & {\scriptsize 45355} & {\scriptsize 13330} & {\scriptsize 62979}\tabularnewline
{\scriptsize 16} & {\scriptsize{}    JNJ} & {\scriptsize{}    JOHNSON \& JOHNSON} & {\scriptsize 34175} & {\scriptsize 4551} & {\scriptsize 47524} &  & {\scriptsize 22974} & {\scriptsize 7074} & {\scriptsize 43322}\tabularnewline
{\scriptsize 17} & {\scriptsize{}    MCD} & {\scriptsize{}    MCDONALDS CORP} & {\scriptsize 33730} & {\scriptsize 5674} & {\scriptsize 47969} &  & {\scriptsize 27630} & {\scriptsize 10145} & {\scriptsize 48000}\tabularnewline
{\scriptsize 18} & {\scriptsize{}    MRK} & {\scriptsize{}    MERCK \& CO INC} & {\scriptsize 41559} & {\scriptsize 7612} & {\scriptsize 53181} &  & {\scriptsize 31901} & {\scriptsize 11486} & {\scriptsize 47534}\tabularnewline
{\scriptsize 19} & {\scriptsize{}    MSFT} & {\scriptsize{}    MICROSOFT CORP} & {\scriptsize 68396} & {\scriptsize 8725} & {\scriptsize 78393} &  & {\scriptsize 50253} & {\scriptsize 14163} & {\scriptsize 65728}\tabularnewline
{\scriptsize 20} & {\scriptsize{}    MMM} & {\scriptsize{}    3M CO} & {\scriptsize 36835} & {\scriptsize 6737} & {\scriptsize 46521} &  & {\scriptsize 24402} & {\scriptsize 10205} & {\scriptsize 38750}\tabularnewline
{\scriptsize 21} & {\scriptsize{}    JPM} & {\scriptsize{}    MORGAN J P \& CO INC} & {\scriptsize 6991} & {\scriptsize 3180} & {\scriptsize 9940} &  & {\scriptsize 4428} & {\scriptsize 5874} & {\scriptsize 10568}\tabularnewline
{\scriptsize 22} & {\scriptsize{}    PFE} & {\scriptsize{}    PFIZER INC} & {\scriptsize 51100} & {\scriptsize 6170} & {\scriptsize 54751} &  & {\scriptsize 47414} & {\scriptsize 12109} & {\scriptsize 52358}\tabularnewline
{\scriptsize 23} & {\scriptsize{}    PG} & {\scriptsize{}    PROCTER \& GAMBLE CO} & {\scriptsize 36971} & {\scriptsize 5904} & {\scriptsize 52782} &  & {\scriptsize 25422} & {\scriptsize 8974} & {\scriptsize 46956}\tabularnewline
{\scriptsize 24} & {\scriptsize{}    T} & {\scriptsize{}    AT\&T INC} & {\scriptsize 42547} & {\scriptsize 5542} & {\scriptsize 52123} &  & {\scriptsize 32851} & {\scriptsize 9496} & {\scriptsize 46524}\tabularnewline
{\scriptsize 25} & {\scriptsize{}    TRV} & {\scriptsize{}    TRAVELERS COMPANIES INC} & {\scriptsize 21404} & {\scriptsize 3803} & {\scriptsize 27783} &  & {\scriptsize 15069} & {\scriptsize 5811} & {\scriptsize 23646}\tabularnewline
{\scriptsize 26} & {\scriptsize{}    UTX} & {\scriptsize{}    UNITED TECHNOLOGIES CORP} & {\scriptsize 34765} & {\scriptsize 6366} & {\scriptsize 48200} &  & {\scriptsize 28973} & {\scriptsize 11712} & {\scriptsize 45580}\tabularnewline
{\scriptsize 27} & {\scriptsize{}    UNH} & {\scriptsize{}    UNITEDHEALTH GROUP INC} & {\scriptsize 39924} & {\scriptsize 8864} & {\scriptsize 58885} &  & {\scriptsize 32757} & {\scriptsize 16724} & {\scriptsize 56451}\tabularnewline
{\scriptsize 28} & {\scriptsize{}    VZ} & {\scriptsize{}    VERIZON COMMUNICATIONS INC} & {\scriptsize 39642} & {\scriptsize 6461} & {\scriptsize 56527} &  & {\scriptsize 34306} & {\scriptsize 11100} & {\scriptsize 54563}\tabularnewline
{\scriptsize 29} & {\scriptsize{}    WMT} & {\scriptsize{}    WAL-MART STORES INC} & {\scriptsize 37172} & {\scriptsize 5668} & {\scriptsize 51904} &  & {\scriptsize 31679} & {\scriptsize 10429} & {\scriptsize 52890}\tabularnewline
{\scriptsize 30} & {\scriptsize{}    DD} & {\scriptsize{}     DU PONT E I NEMOURS \& CO} & {\scriptsize 35262} & {\scriptsize 7404} & {\scriptsize 46302} &  & {\scriptsize 24077} & {\scriptsize 11097} & {\scriptsize 41173}\tabularnewline
\hline
\end{tabular}

\caption{{\small Number of observations for in-the-money (ITM), at-the-money (ATM)
and out-of-the-money (OTM) call option quotes for the stocks which are
the constituents of the Dow Jones Industrial Average Index (DJIA). The data are further broken down by maturity.
According to the classification of \cite{bollen2004does}, a call
option is considered OTM if its delta is less then 0.375,
ATM if its delta ranges between 0.375 and 0.625 and ITM
if its delta is above 0.625.\label{Table:Db_description}}}

\end{table}



\begin{table}
\begin{tabular}{c|l|cccc|cccccccc}
{\footnotesize Underlying } & \multirow{1}{*}{{\footnotesize BS}} & \multicolumn{4}{c|}{{\footnotesize MRT}} & \multicolumn{8}{c}{{\footnotesize BTS}}\tabularnewline
 & {\footnotesize $\sigma_{BS}$} & {\footnotesize $\gamma$} & {\footnotesize $\sigma_{M}$} & {\footnotesize $\sigma_{\psi}$} & {\footnotesize $\mu_{\psi}$} & {\footnotesize $\gamma$} & {\footnotesize $\sigma_{\psi}$} & {\footnotesize $\mu_{\psi}$} & {\footnotesize $\omega$} & {\footnotesize $\sigma_{LT}$} & {\footnotesize $\beta$} & {\footnotesize $\rho$} & {\footnotesize $\sigma_{0}$}\tabularnewline
\hline
{\footnotesize All stocks} & {\footnotesize 0.29} & {\footnotesize 1.33} & {\footnotesize 0.22} & {\footnotesize 0.16} & {\footnotesize -0.12} & {\footnotesize 0.50} & {\footnotesize 0.18} & {\footnotesize -0.12} & {\footnotesize 0.75} & {\footnotesize 0.32} & {\footnotesize 1.52} & {\footnotesize -0.35} & {\footnotesize 0.28}\tabularnewline
{\footnotesize SP500{*}} & {\footnotesize 0.18} & {\footnotesize NA} & {\footnotesize NA} & {\footnotesize NA} & {\footnotesize NA} & {\footnotesize 0.61} & {\footnotesize 0.14} & {\footnotesize -0.09} & {\footnotesize 0.4} & {\footnotesize 0.2} & {\footnotesize 3.93} & {\footnotesize -0.52} & {\footnotesize 0.2}\tabularnewline
\hline
\end{tabular}

\caption{{\small Average values of the parameters of the models of Black-Scholes (BS),
Merton (MRT) and Bates (BTS) calibrated at each day before the ex-dividend date on the options written on
the dividend-paying stocks belonging to the Dow Jones Industrial Average Index (DJIA). In total we computed 1701 calibrations
and the reported values are the averages across these calibrations. \protect \\
The in-sample sum of squared error is on average equal to 0.26 for
the Black-Scholes model, 0.20 for the Merton model,
and 0.16 for the Bates model with stochastic
volatility.\protect \\
\protect \\
{*}Calibrated parameters of the SP500 dynamics are from \cite{bakshi1997empirical}.\label{Table:Param}}}
\end{table}



\begin{table}

\begin{centering}
\begin{tabular}{lccccccc}
 & \multicolumn{3}{c}{{\scriptsize Average rational implied fee}} &  & \multicolumn{3}{c}{{\scriptsize \%Implied fee > 0.4446}}\tabularnewline
\cline{2-4} \cline{6-8}
{\scriptsize Underlying} & {\scriptsize BS} & {\scriptsize MRT} & {\scriptsize BTS} &  & {\scriptsize BS} & {\scriptsize MRT} & {\scriptsize BTS}\tabularnewline
\hline
{\scriptsize All stocks} & {\scriptsize 7.54} & {\scriptsize 7.27} & {\scriptsize 7.23} &  & {\scriptsize 94\%} & {\scriptsize 93\%} & {\scriptsize 93\%}\tabularnewline
\hline
\end{tabular}
\par\end{centering}{\scriptsize \par}

\caption{{\small The table reports the average implied fee per share which would explain
the non-exercise behavior of investors in each model: Black-Scholes (BS), Merton (MRT) and Bates (BTS).
The average implied fee is calculated for each option that should
be exercised but which is not optimally exercised by some of the investors
as the value of the trading costs $\mathcal{F}$ which makes the continuation
value of the option equal to the early exercise proceeds:
$	C(S-d,K+\mathcal{F},T)=(S-K-\mathcal{F}). $
In the last three columns of the table we report the percentage of
options for which the fee that would explain the suboptimal non-exercise behavior is higher than the conservative
fee of $0.4446$ dollar per share estimated by \cite{pool2008failure}.\label{table:Fee}}}

\end{table}



\begin{table}
\begin{center}
\begin{tabular}{cccc}
 & \multirow{1}{*}{{\small Model}} & {\small Without fee} & {\small With fee }\tabularnewline
\hline 
{\small Contracts outstanding} & &\multicolumn{2}{c}{406 414 980}\tabularnewline
\hline
{\small Total market value} & &\multicolumn{2}{c}{99 392 927 000}\tabularnewline
\hline
{\small Contracts that should be} & BS & 38 527 586 & 31 551 786\tabularnewline
exercised &  & (9.48\%) & (7.76\%)\tabularnewline
\cline{2-4}
 & MRT & 30 633 542 & 25 340 009\tabularnewline
 &  & (7.54\%) & (6.23\%)\tabularnewline
\cline{2-4}
 & BTS & 30 486 666 & 25 050 616\tabularnewline
 &  & (7.5\%) & (6.16\%)\tabularnewline
\hline
{\small Contracts that are left} & BS & 15 214 908 & 11 077 913\tabularnewline
{\small suboptimally non-exercised} &  & (39.49\%) & (35.11\%)\tabularnewline
\cline{2-4}
 & MRT & 9 404 406 & 7 118 002\tabularnewline
 &  & (30.70\%) & (28.09\%)\tabularnewline
\cline{2-4}
 & {\small BTS} & 8 786 524 & 6 702 799\tabularnewline
 &  & (28.82\%) & (26.76\%)\tabularnewline
\hline
{\small Money available} & BS & 770 287 766 & 647 088 372\tabularnewline
{\small due to exercise opportunity} &  & (0.77\%) & (0.65\%)\tabularnewline
\cline{2-4}
 & MRT & 555 922 918 & 460 980 123\tabularnewline
 &  & (0.56\%) & (0.46\%)\tabularnewline
\cline{2-4}
 & {\small BTS} & 617 173 686 & 509 713 158\tabularnewline
 &  & (0.62\%) & (0.51\%)\tabularnewline
\hline
{\small Total loss due to} & {\small BS} & 209 284 628 & 172 424 967\tabularnewline
{\small suboptimal non-exercise} &  & (27.17\%) & (26.65\%)\tabularnewline
 &  & (0.21\%) & (0.17\%)\tabularnewline
\cline{2-4}
 & {\small MRT} & 133 130 786 & 108 951 652\tabularnewline
 &  & (23.95\%) & (23.63\%)\tabularnewline
 &  & (0.13\%) & (0.11\%)\tabularnewline
\cline{2-4}
 & {\small BTS} & 147 480 996 & 123 615 413\tabularnewline
 &  & (23.9\%) & (24.25\%)\tabularnewline
 &  & (0.15\%) & (0.12\%)\tabularnewline
\hline
\end{tabular}
\end{center} 

\caption{{\small
Summary results of the total loss due to suboptimal non-exercise for the short-term
call option series. The numbers are calculated for each series and each day
before the ex-dividend date separately, and then pooled together.
%
\newline
The number of contracts outstanding is the total open interest of all contracts two days before the ex-dividend date. The contracts that should be exercised under a specific model, i.e. Black-Scholes (BS), Merton (MRT) and Bates (BTS), are the contracts outstanding for which the continuation value is lower than the exercise proceeds. The number of contracts that are left suboptimally non-exercised is the sum of the open interests one day before ex-dividend of the contracts that should have been exercised.
We compute the other quantities in the table in the following way:
\newline
Total market value = Contracts outstanding$\times$Market price$\times$100,
 \newline
 Money available = $\max\{0,(S-K-\mathcal{F}-\text{Continuation value})\times \text{Contracts outstanding}\times100\}$,
 \newline
 Total loss = $\max\{0,(S-K-\mathcal{F}- \text{Continuation value})\times \text{Open interest}_{t-1}\times 100\}$,
 \newline
where $\mathcal{F}$ is the exercise fee. In the first column the results are computed considering $\mathcal{F}=0$, while in the second column the results are computed considering the conservative fee of $0.44$ dollar.
The first percentage in
parenthesis in the Total loss due to non-exercise is computed with
respect to the money available due to exercise opportunities, while the second one is computed with
respect to the total market value.\label{table:NII_lost}}}
\end{table}

\clearpage

%
%
%



\begin{figure}
	\begin{center}
  \title{\textbf{Percentage of suboptimal non-exercise as function of delta moneyness}}
  \begin{tikzpicture}

\begin{axis}[%
width=4.52083333333333in,
height=3.565625in,
scale only axis,
xmin=0.91,
xmax=1,
xlabel={Delta},
ymin=20,
ymax=70,
ytick={20,25,30,35,40,45,50,55,60,65,70},
yticklabels={20\%,25\%,30\%,35\%,40\%,45\%,50\%,55\%,60\%,65\%,70\%},
ylabel={Subptimal non-exercise},
]
\addplot [
color=blue,
mark size=3.5pt,
only marks,
mark options={solid,draw=blue}
]
table[row sep=crcr]{
0.91 64.9858812884619\\
0.93 61.7468236142281\\
0.95 61.2933508362636\\
0.97 54.2337243125713\\
0.985 47.8986052918035\\
0.9949995 32.2092693292697\\
0.9999995 22.2914794972761\\
};
\end{axis}
\end{tikzpicture}%
  \caption{Percentage of contracts suboptimally non-exercised in the model of Black-Scholes as function of the delta moneyness of the contracts. \label{Figure:moneyness}}
	\end{center}
\end{figure}
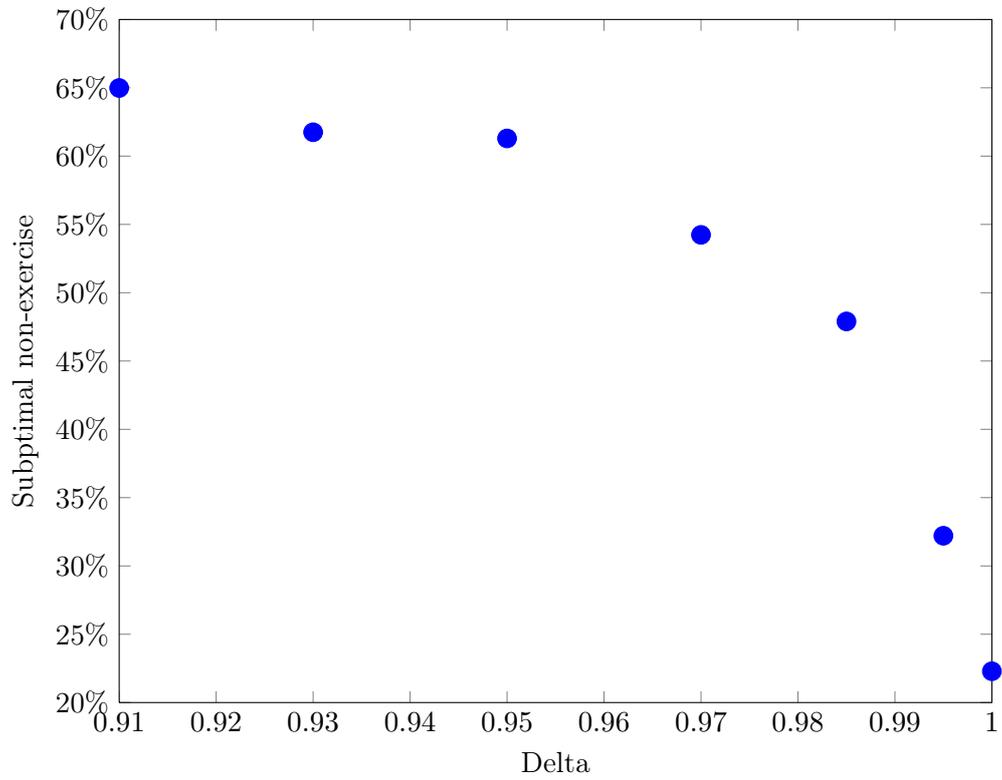

\clearpage



\end{document}